%% file: test_separability_arXiv-v4.tex
\pdfoutput=1
\documentclass[11pt,authoryear,a4paper]{article}
\usepackage{geometry}
\usepackage[authoryear]{natbib}
\usepackage[english]{babel}
\usepackage{inputenc}
\usepackage{times}
\usepackage{graphicx}
\usepackage{amsmath}
\usepackage{mathtools}
\usepackage{amsthm}
\usepackage{amssymb}
\usepackage{amsfonts}
\usepackage{dsfont}
\usepackage{url}
\usepackage{bm}
\usepackage{setspace}
\usepackage{xspace}
\usepackage{enumerate}
\usepackage{authblk}
\PassOptionsToPackage{dvipsnames}{xcolor}
\RequirePackage{xcolor} 

\usepackage{algorithm}
\allowdisplaybreaks

\usepackage{caption}
\usepackage{subcaption}

\numberwithin{equation}{section}

\newtheorem{thm}{theorem}[section]

\newtheorem{rmk}[thm]{remark}
\newtheorem{cor}[thm]{corollary}
\newtheorem{prop}[thm]{proposition}
\newtheorem{lma}[thm]{lemma}
\newtheorem{cond}[thm]{Condition}

\newcommand{\tp}{\mathsf{T}}

\DeclareMathOperator*{\tensort}{\otimes_2}  
\DeclareMathOperator*{\btensort}{ {\bigotimes}_2} 

  \newcommand{\tb}{\tilde \beta}

\DeclareMathOperator*{\tensor}{\otimes}  
\DeclareMathOperator*{\kprod}{\: \widetilde{\otimes} \:}  
\DeclareMathOperator*{\bkprod}{\: \widetilde{\bigotimes} \:}  
\renewcommand{\sc}[1]{\left\langle #1 \right\rangle}
\newcommand{\ssc}[1]{{\langle #1 \rangle}} 
\newcommand{\HSsc}[1]{\left \langle #1 \right \rangle_{\!\!\mathcal{S}_{2}}}

\newcommand{\ind}[1]{\mathbf{1}_{\{ #1 \} }}

\newcommand{\LL}[1]{{L^2\left( #1, \mathbb{R} \right)}\xspace}

 \DeclareMathOperator{\ee}{ \:{\mathbb{E} }}
 \newcommand{\eee}[1]{\,\:{\mathbb{E}\left[ #1 \right]}}
\DeclareMathOperator{\cov}{ {\mathrm{cov}}}
\newcommand{\covv}[1]{ {{\cov}\left[ #1 \right] }}

\newcommand{\rainf}{\rightarrow \infty}
\newcommand{\raz}{\rightarrow 0}
\newcommand{\convd}{\stackrel{d}{\longrightarrow}}
\newcommand{\convp}{\stackrel{p}{\longrightarrow}}

\newcommand{\id}{\mathrm{Id}}

\DeclareMathOperator{\trace}{{\mathrm{Tr}}}
\renewcommand{\vec}{{\mathrm{Vec}}}

\newcommand{\schatten}{\mathcal{S}}
\newcommand{\tc}[1]{\schatten_1(#1)}
\newcommand{\bounded}[1]{\schatten_\infty(#1)}
\newcommand{\HS}[1]{\schatten_2(#1)}

\newcommand{\hnorm}[1]{{\left\lVert #1 \right\rVert}}   
\newcommand{\snorm}[1]{{\left\vert\kern-0.2ex\left\vert\kern-0.2ex\left\vert #1
  \right\vert\kern-0.2ex\right\vert\kern-0.2ex\right\vert}}
\newcommand{\ssnorm}[1]{{\vert\kern-0.2ex\vert\kern-0.2ex\vert #1
  \vert\kern-0.2ex\vert\kern-0.2ex\vert}} 
\newcommand{\opnorm}[1]{\snorm{#1}_\infty}
\newcommand{\tnorm}[1]{\snorm{#1}_1}
\newcommand{\hsnorm}[1]{\snorm{#1}_2}

\newcommand{\vep}{{\varepsilon}}

\newcommand{\bR}{\mathbb{R}}

\newcommand{\simiid}{\stackrel{\text{i.i.d.}}{\sim}}

\newcommand{\bT}{\mathbf{T}}

\title{Tests for separability in nonparametric covariance operators of random surfaces}
\author{J. A. D. Aston\footnote{Research Supported by EPSRC grant EP/K021672/2.}, D. Pigoli and S. Tavakoli\footnote{Address for correspondence: Shahin Tavakoli,
Statistical laboratory, Department of Pure Mathematics and
Mathematical Statistics, University of Cambridge, Wilberforce Road,
Cambridge, CB3 0WB, United Kingdom. Email:
s.tavakoli@statslab.cam.ac.uk}}
\affil{Statistical
Laboratory\\ Department of Pure Mathematics and Mathematical Statistics\\
University of Cambridge}

\begin{document}

\date{}
\maketitle


\abstract{
  .\hspace*{1mm}The assumption of separability of the covariance operator for a random image or
hypersurface can be of substantial use in applications, especially in situations where the accurate estimation
of the full covariance structure is unfeasible, either for computational reasons, or due to a small sample
size. However, inferential tools to verify this assumption are somewhat lacking in high-dimensional or
functional {data analysis} settings, where this assumption is most relevant. We propose here to test separability by focusing
on $K$-dimensional projections of the difference between the covariance operator and {a} nonparametric
separable approximation. The subspace we project onto is one generated by the eigenfunctions of the covariance
operator estimated under
the separability hypothesis, negating the need to ever estimate the full non-separable covariance. We show
that the rescaled difference of the sample covariance operator with its separable approximation is
asymptotically Gaussian. As a by-product of this result, we derive asymptotically pivotal tests under Gaussian
assumptions, and propose bootstrap methods for approximating the distribution of the test statistics. We probe
the finite sample performance through simulations
studies, and present an application to log-spectrogram images from a phonetic linguistics dataset.
}


\section{Introduction}

Many applications involve hypersurface data, data that is both functional {\citep[as in functional data
  analysis, see e.g.][]{ramsay:2005,ferraty:2006,Horvath:2012,wang:2015review}} and multidimensional. Examples
abound and include images from medical devices such as MRI \citep{lindquist2008statistical} or PET
\citep{worsley1996unified}, spectrograms derived from audio signals \citep[][and as in the application we
consider in Section \ref{sec:results}]{rabiner1978digital} or geolocalized data \citep[see, e.g.,
][]{Secchi2015}. In these kinds of problem, the number of available observations {(hypersurfaces)} is often
small relative to the high-dimensional nature of the individual observation, and not usually large enough to
estimate a full multivariate covariance function.

It is usually, therefore, necessary to make some simplifying
assumptions about the data or their covariance
structure. If the covariance structure is of interest, such as for PCA or network modeling, for
instance, it is usually assumed to have some kind of lower dimensional structure.
Traditionally, this translates into a \emph{sparsity} assumption: one assumes that
most entries of the covariance matrix or function are zero. Though
being  relevant for a number of applications \citep{tibshirani2014praise}, this traditional definition of sparsity
may not be appropriate in some cases, such as in imaging, as this can give rise to artefacts in the analysis
(for example, holes in an image). In such problems, where the data is multidimensional,
a natural assumption that can be made is that the covariance is
\emph{separable}. This assumption greatly simplifies both the
estimation and the computational cost in dealing with multivariate
covariance functions, while still allowing for a positive definite
covariance to be specified. In the context of space-time data
$X(s,t)$, for instance, where $s \in [-S,S]^d$, $S
>0$, denotes the location in space, and $t \in [0,T]$,  $T > 0$, is
the time index, the assumption of separability translates into
\begin{equation}
\label{eq:intro-separable-structure}
  c(s,t,s',t')=c_1(s,s')c_2(t,t'), \quad s,s' \in [-S,S]^d; t,t' \in [0,T],
\end{equation}
where  $c$, $c_1$, and $c_2$, are respectively  the full covariance
function, the  space covariance function and the time covariance
function. In words, this means that the full covariance function
factorises as a product of the spatial covariance function with the
time covariance function.

The separability assumption \citep[see e.g.][]{gneiting2006geostatistical,genton2007separable} simplifies the
covariance structure of
the process and makes it far easier to estimate; in some sense, the separability assumption results in a
  estimator of the covariance which has less variance, at the expense of a possible bias.
As an illustrative example, consider that we observe a discretized version of the process through measurements
on a two dimensional grid (without loss of generality, as the same arguments apply for any dimension greater than
$2$)
being a $q\times p$ matrix {(of course, the functional data analysis approach taken here does
  \emph{not} assume that
the replications of the process are observed on same grid, nor that they are observed on a grid)}. Since we are not assuming a parametric
form for the covariance, the degrees of freedom in the full
covariance are  $qp(qp+1)/2$, while the separability assumption
reduces them to $q(q+1)/2 + p(p+1)/2$. This reflects a dramatic
reduction in the dimension of the problem even for moderate value of
$q,p$, and overcomes both computational and estimation problems due to the relatively small sample sizes
available in applications. For example, for $q=p = 10$, we have
$qp(qp+1)/2 = 5050$ degrees of freedom, however, if the
separability holds, then we have only $q(q+1) + p(p+1)= 110$ degrees of
freedom. Of course, this is only one example, and our approach is not restricted to data on a grid, but
this illustrates the computational savings that such assumptions can possess.

\begin{figure}[h]
  \begin{center}
    \includegraphics{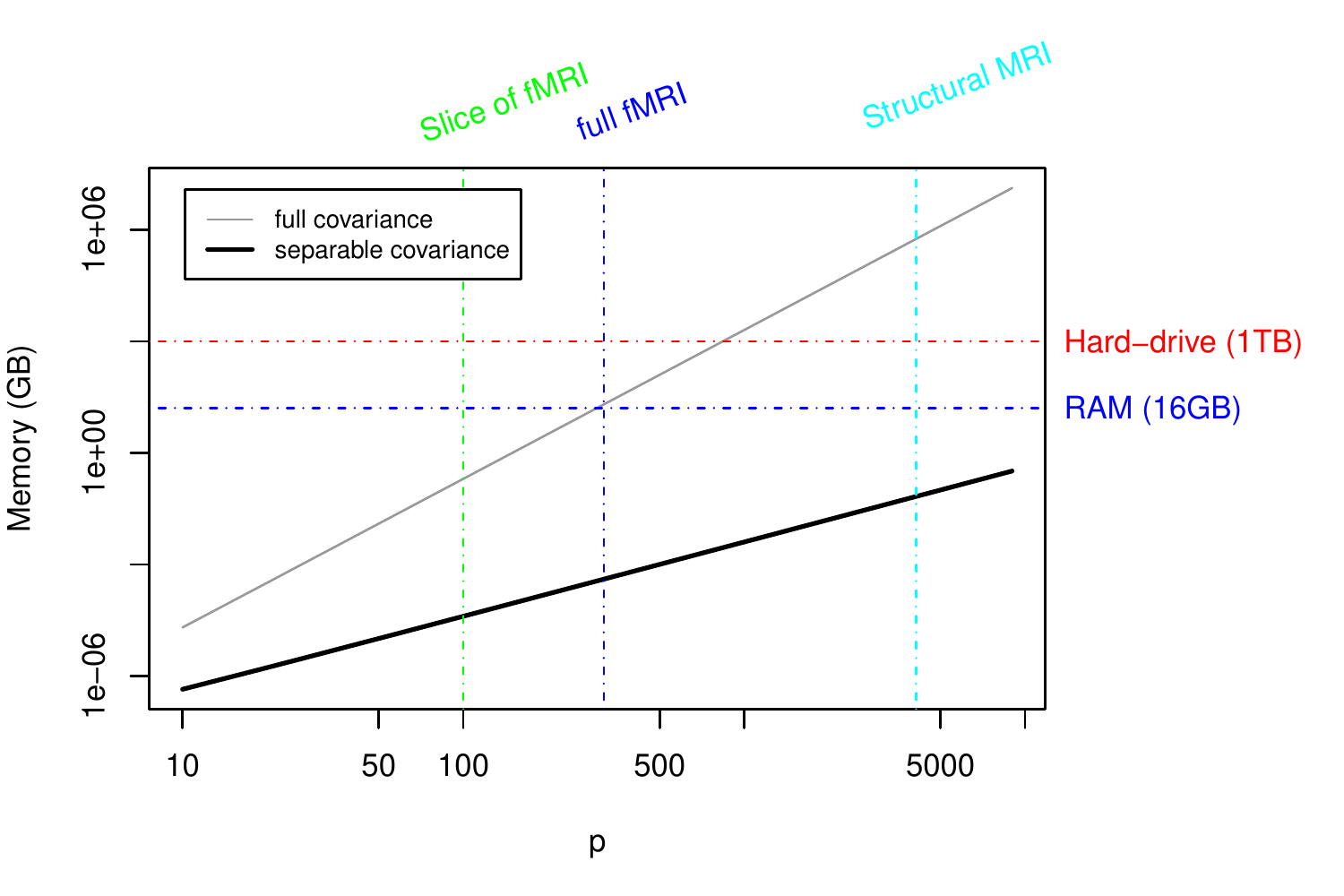}
  \end{center}
  \caption{Memory required to store the full covariance and the separable covariance of $p \times p$ matrix data, as a function of $p$. Several types of data related to Neuroimaging (structural and functional Magnetic Resonance Imaging) are used as exemplars of data sizes, as they naturally have multidimensional structure.}
  \label{fig:memory_usage}
\end{figure}

Three related computational classes of problem can be identified. In the first case, the full covariance structure
can be computed and stored. In the second one, it is still possible,
although burdensome, to compute the full covariance matrix but it
can not be stored, while the last class includes problems where even
computation of the full covariance is infeasible. The values of $q,p$ that set the
boundaries for these classes depend of course on the available
hardware and they are rapidly changing. At the present time however,
for widely available systems, storage is feasible up to $q,p
 \approx 100$ while computation becomes unfeasible when $q,p$ get close to
 $1000$ (see Figure~\ref{fig:memory_usage}). On the contrary, a separable covariance structure can be usually
 both computed and stored without effort even for these sizes of problem. We would like to stress however that the constraints coming from the
need for statistical accuracy are usually tighter. The estimation of
the full covariance structure even for $q,p=100$ presents about
$5\times 10^7$ unknown parameters, when typical sample sizes are in
the order of hundreds at most. If we are able to assume
separability, we can rely on far more accurate estimates.

While the separability assumption can be very useful, and is indeed
often implicitly made in many higher dimensional applications when using
isotropic smoothing \citep{worsley1996unified,lindquist2008statistical}, very little has been done to
develop tools to assess its validity on a case by case basis. In the
classical multivariate setup, some tests for the separability
assumption are available. These have been mainly developed in the
field of spatial statistics \citep[see][and references therein]{Lu2005,Fuentes2006}, where the discussion of separable covariance
functions is well-established, or for applications involving repeated
measures \citep{Mitchell2005}. These methods, however, rely on the
estimation of the full multidimensional covariance structure, which
can be troublesome.  It is sometimes possible to circumvent this
problem by considering a parametric model for the full covariance
structure \citep[][]{Simpson2010,Simpson2014,LiuRH2014}. On the contrary, when
the covariance is being non-parametrically specified, as will be
the case in this paper, estimation of the full covariance is at
best computationally complex with large estimation errors, and in
many cases simply computationally infeasible. Indeed, we highlight
that, while the focus of this paper is on checking the viability of
a separable structure for the covariance, this is done
{without any parametric assumption on the form of}
$c_1(s,s')$ and $c_2(t,t')$,
thus allowing for the maximum flexibility. This is opposed to
assuming a parametric separable form with only few unknown parameters,
which is usually too restrictive in many applications, something that
has led to separability being rightly criticised and viewed with suspicion in the spatio-temporal
statistics literature \citep{gneiting2002nonseparable,gneiting2006geostatistical}. {Moreover, the
methods we develop here are aimed to applications typical of functional data, where replicates from the
underlying random process are available. This is different from the spatio-temporal setting, where
usually only one realization of the process is observed.}
See also \citet{constantinou:2015} for
another approach to test for separability in functional data.

It is important to notice that a separable covariance structure
({or equivalently, a separable correlation structure})
is not necessarily connected with the original data being separable.
Furthermore, sums or differences of separable hypersurfaces are not
necessarily separable. On the other hand, the error structure may be
separable even if the mean is not. Given that in many applications
of functional data analysis, the estimation of the covariance is the
first step in the analysis, we concentrate on covariance
separability. Indeed, covariance separability is an extremely useful
assumption as it implies separability of the eigenfunctions,
allowing computationally efficient estimation of the eigenfunctions
(and principal components). Even if separability is misspecified, separable eigenfunctions can still form a
basis representation for the data, they simply no longer carry
optimal efficiency guarantees in this case \citep{aston2012evaluating}, but can often have near-optimality under the appropriate assumptions \citep{chen2015modeling} .

In this paper, we propose a test to verify if the data at hand are
in agreement with a separability assumption. Our test does not
require the estimation of the full covariance structure, but only
the estimation of the separable structure
\eqref{eq:intro-separable-structure}, thus avoiding both the
computational issues and the diminished accuracy involved in the
former. To do this, we rely on a strategy from Functional Data
Analysis
\citep{ramsay:2002,ramsay:2005,ferraty:2006,ramsayfunctional,horvath2012}, which
consists in projecting the observations onto a carefully chosen
low-dimensional subspace. The key fact for the success of our
approach is that, under the null hypothesis, it is possible to
determine this subspace using only the marginal covariance
functions. {While the optimal choice for the dimension of
this subspace is a non-trivial problem, some insight can be obtained through our
extensive simulation studies (Section~\ref{sec:simulations}). Ultimately, the proposed test checks
the separability in the chosen subspace, which will often be the
focus of following analyses.
}

The paper proceeds as follows. In Section~\ref{s:SepCovs}, we
examine the ideas behind separability,  propose a separable
approximation of a covariance operator, and study the asymptotics of
the difference between the sample covariance operator and its
separable approximation. This difference will be the building block
of the testing procedures introduced in Section~\ref{sec:bootstrap},
and whose distribution we propose to approximate by bootstrap
techniques. In Section~\ref{sec:results}, we investigate by means of
simulation studies the finite sample behaviour of our testing
procedures  and apply our methods to acoustic phonetic data. A
conclusion, given in Section~\ref{sec:conclusion}, summarizes the
main contributions of this paper. Proofs  are collected in
appendices~\ref{app:asymptotic-covariance}, \ref{app:proofs}, and
\ref{app:partial-traces}, while implementation details,
theoretical background and additional figures can be found 
in the appendices~\ref{app:implementation},
\ref{app:background-results-main-paper} and \ref{app:additional-simulations-studies}.
All the tests introduced in the paper are available as an R package \texttt{covsep} \citep{R:covsep},
available on the Comprehensive R Archive Network (CRAN).

For notational simplicity, the proposed method will be described for
two dimensional functional data (e.g.\ random surfaces), hence a
four dimensional covariance structure (i.e.\ the covariance of a
random surface), but the generalization to higher dimensional cases
is straightforward.  The methodology is developed in general for
data that take values in a Hilbert space, but the case of square
integrable surfaces---being relevant for the case of acoustic
phonetic data---is used throughout the paper as a demonstration.
{We recall that the proposed approach is not restricted to data observed on a
regular grid, although for simplicity of exposition we consider here the
case where data are observed densely and a pre-processing smoothing
step allows to consider the smooth surfaces as our observations, as
happens for example the case of the acoustic phonetic data described
in Section \ref{sec:results}. If data are observed sparsely, the
proposed approach can still be applied but there may be the need to
use more appropriate estimators for the marginal covariance
functions \citep[see, e.g.][]{yao:2005} and these need to satisfy
the properties described in Section \ref{s:SepCovs}.}

\section{Separable Covariances: definitions, estimators and asymptotic results}\label{s:SepCovs}

While the general idea of the factorization of a multi-dimensional
covariance structure as the product of lower dimensional covariances
is easy to describe, the development of a testing procedure asks for
a rigorous mathematical definition and the introduction of some
technical results. In this section we propose a definition of
separability for covariance operators, show how it is possible to
estimate a separable version of a covariance operator and evaluate
the difference between the empirical covariance operator and its
separable version. Moreover, we derive some asymptotic results for
these estimators. To do this, we
first set the problem in the framework of random elements in
Hilbert spaces and their covariance operators. {The benefit
in doing this is twofold. First, our results become applicable in more general settings (e.g.\ multidimensional
functional data, data on multidimensional grids, fixed size rectangular random matrices) and do not depend on a specific choice of smoothness of the
data (which is implicitly assumed when modeling the data  as e.g.\ square integrable surfaces). They only rely
on the Hilbert space structure of the space in which the data lie.
Second, it highlights
the importance of the  \emph{partial trace} operator in the
estimation of the separable covariance structure, and how the
properties of the partial trace (Appendix~\ref{app:partial-traces}) play a crucial role in the asymptotic
behavior of the proposed test statistics.} However, to ease
explanation, we use the case of the
Hilbert space of square integrable surfaces (which shall be used in our linguistic application, see
Section~\ref{sec:results}) as an illustration of our testing procedure.


\subsection{{Notation} }

Let us first introduce some definitions and notation about operators in a Hilbert space
\citep[see e.g.][]{gohberg:1990,Kadison:1997vol1,Ringrose:1971}.
  Let $H$ be a real separable Hilbert space (that
is, a Hilbert space with a countable orthonormal basis), whose inner
product and norm are denoted by $\sc{\cdot, \cdot}$ and
$\hnorm{\cdot}$, respectively. The space of bounded (linear) operators on $H$
is denoted by $\bounded{H}$, and its norm is $\opnorm{T} = \sup_{x
  \neq 0} \hnorm{Tx}/\hnorm{x}$. The space of Hilbert--Schmidt operators on $H$
is denoted by $\HS{H}$, and is a Hilbert space with the
inner-product $\HSsc{S,T} = \sum_{i \geq 1}\sc{S e_i, T e_i}$ and
induced norm $\hsnorm{\cdot}$, where $(e_i)_{i \geq 1} \subset H$ is
an orthonormal basis of $H$.  The space of trace-class operator on
$H$ is denoted by $\tc{H}$, and consists of all compact operators
$T$ with finite trace-norm, i.e.\ $\tnorm{T} = \sum_{n \geq 1}
s_n(T) < \infty$, where $s_n(T) \geq 0$ denotes the $n$-th singular
value of $T$. For any trace-class operator $T \in \tc{H}$, we define
its trace by $\trace(T) = \sum_{i \geq 1} \sc{Te_i,
  e_i}$, where $(e_i)_{i \geq 1} \subset H$ is an orthonormal basis, and the sum is independent of the choice
of the orthonormal basis.

If $H_1, H_2$ are real separable Hilbert spaces, we denote by $H =
H_1 \tensor H_2$ their tensor product Hilbert space, which is
obtained by the completion of all finite sums $\sum_{i,j=1}^N
u_i \tensor
  v_j$, $u_i \in H_1, v_j \in H_2$, under the inner-product $\sc{u \tensor v, z \tensor
    w}=\sc{u,z}\sc{v,w}, u,z \in H_1, z,w \in H_2$ \citep[see e.g.][]{Kadison:1997vol1}.
  If $C_1 \in \bounded{H_1}$, $C_2 \in \bounded{H_2}$, we
  denote by  $C_1
\kprod C_2$  the unique linear operator on $ H_1 \tensor H_2 $  satisfying
\begin{equation}
  \label{eq:kprod-defn}
  \left( C_1 \kprod C_2 \right)(u \tensor v) = C_1u \tensor C_2v, \quad \text{for all } u \in H_1, v \in H_2.
\end{equation}
It is a bounded operator on $H$, with $\opnorm{C_1 \kprod C_2} = \opnorm{C_1} \opnorm{C_2}$. Furthermore,
if  $C_1 \in \tc{H_1}$ and $C_2 \in \tc{H_2}$, then $C_1 \kprod C_2 \in \tc{H_1 \tensor H_2}$ and $\tnorm{C_1 \kprod
  C_2} = \tnorm{C_1}\tnorm{C_2}$. We denote by
$\trace_{1}: \tc{H_1 \tensor H_2} \rightarrow \tc{H_2}$ the \emph{partial trace with respect to $H_1$}. It is
the unique bounded linear operator satisfying $\trace_{1}(A \kprod B) = \trace(A)B$, for all $A \in \tc{H_1}, B
\in \tc{H_2}$.  $\trace_2: \tc{H_1 \tensor H_2} \rightarrow \tc{H_1}$ is defined symmetrically (see
Appendix~\ref{app:partial-traces} for more details).

If $X \in H$ is a random element
with $\ee \hnorm{X} < \infty$, then $\mu = \ee X \in H$, the mean of $X$, is well defined. Furthermore, if
$\ee \hnorm{X}^2 < \infty$, then $C = \eee{(X-\mu) \tensort (X-\mu)}$ defines the \emph{covariance operator}
of $X$, where $f \tensort g$ is the operator on $H$ defined by $(f \tensort g)h = \sc{h,g}f$,
for $f,g,h \in H$. The covariance operator $C$ is a trace-class hermitian operator on $H$, and encodes all the
second-order fluctuations of $X$ around its mean.

Using this nomenclature, we are going to deal with random variables
belonging to a tensor product Hilbert space. This framework
encompasses the situation where $X$ is a random surface, for example
a space-time indexed data, i.e.\ $X = X(s,t), s \in [-S,S]^d, t \in
[0,T]$,  $S, T
> 0$, by setting $H = \LL{[-S,S]^d \times [0,T]}$, for instance {(notice however that additional smoothness
  assumptions on $X$ would lead to assume that $X$ belongs to some other Hilbert space)}. In
this case, the covariance operator of the random element $X \in
\LL{[-S,S]^d \times [0,T]}$ satisfies
\[
  Cf(s,t) = \int_{[-S,S]^d} \int_{0}^T c(s,t,s',t') f(s', t') ds' dt', \quad s \in [-S,S]^d, t \in [0,T],
\]
$f \in \LL{[-S,S]^d \times [0,T]}$,
where $c(s,t,s',t') = \covv{X(s,t), X(s',t')}$ is the
\emph{covariance function} of $X$. The space of square integrable
surfaces,
\[
  \LL{[-S,S]^d \times [0,T]},
\]
is a tensor product Hilbert space
because it can can be identified with
\[
  \LL{[-S,S]^d} \tensor \LL{[0,T]}.
\]

\subsection{ {Separability} }
We recall now that we want to define separability so that the
covariance function can be written as
  $c(s,t,s',t') = c_1(s,s') c_2(t,t'),$
for some $c_1 \in \LL{[-S,S]^d \times [-S,S]^d}$ and $c_2 \in \LL{[0,T]
\times [0,T]}$. This can be extended to the covariance operator of a
random elements $X \in H = H_1 \tensor H_2$, where $H_1, H_2$ are
arbitrary separable real Hilbert spaces. We call its covariance
operator $C$  \emph{separable} if
\begin{equation}
  C = C_1 \kprod C_2,
  \label{eq:abstract-definition-of-separable}
\end{equation}
where $C_1$, respectively $C_2$, are trace-class operators on $H_1$,
respectively on $H_2$, and $C_1 \kprod C_2$ is defined in
\eqref{eq:kprod-defn}. Notice that though the decomposition
\eqref{eq:abstract-definition-of-separable} is not unique, since
$C_1 \kprod C_2 = (\alpha C_1) \kprod (\alpha^{-1} C_2)$ for any
$\alpha \neq 0$, this will not cause any problem at a later stage
since we will ultimately be dealing with the product $C_1 \kprod C_2$,
which is identifiable.

In practice, neither $C$ nor $C_1 \kprod C_2$ are known. If $X_1, \ldots, X_N \simiid X$ and
\eqref{eq:abstract-definition-of-separable} holds, the sample covariance operator $\widehat{C}_N$ is not
necessarily separable in finite samples. However, we can estimate a separable approximation of it by
\begin{equation}
  \label{eq:separable-version-of-Cn}
  \widehat{C}_{1,N} \kprod \widehat{C}_{2,N},
\end{equation}
where $\widehat{C}_{1,N} =
\trace_2(\widehat{C}_N)/\sqrt{\trace(\widehat{C}_N)}$,
$\widehat{C}_{2,N} =
\trace_1(\widehat{C}_N)/\sqrt{\trace(\widehat{C}_N)}$. The intuition
behind \eqref{eq:separable-version-of-Cn} is that
\[
  \trace(T)T = \trace_2(T) \kprod \trace_1(T),
\]
for all $T \in \tc{H_1 \tensor H_2}$ of the form $T = A \kprod B$,
$A \in \tc{H_1}, B \in \tc{H_2}$, with $\trace(T) \neq 0$.

Let us consider again what this means when $X$ is a random element of $\LL{[-S,S]^d
  \times [0, T]}$---i.e.\ the realization of a space-time process---of which we observe $N$ i.i.d.\ replications
$X_1,\ldots, X_N \sim X$. In this case,
Proposition~\ref{prop:partial-trace-of-kernel-operator-with-continous-kernels}
tells us that if the covariance function is continuous, the
operators $\widehat{C}_{1,N}$ and $\widehat{C}_{2,N}$  are defined by
\begin{align*}
  \widehat{C}_{1,N}f(s) &= \int_{[-S,S]^d} \widehat{c}_{1,N}(s,s')f(s)ds, \quad f \in \LL{[-S,S]^d},
     \\
     \widehat{C}_{2,N}g(t) &= \int_{0}^T \widehat{c}_{2,N}(t,t')g(t)dt, \quad g \in \LL{[0,T]},
\end{align*}
where
  \begin{align*}
    \widehat{c}_{1,N}(s,s') &= \frac{\tilde c_{1,N}(s,s')}{\sqrt{\int_{[-S,S]^d} \tilde c_{1,N}(s,s)ds}},
&    \widehat{c}_{2,N}(t,t') &= \frac{\tilde
c_{2,N}(t,t')}{\sqrt{\int_{0}^T \tilde c_{2,N}(t,t)dt}},
\end{align*}
and
\begin{align*}
  \tilde c_{1,N}(s,s') &= \frac{1}{N} \sum_{i=1}^N \int_{0}^T \left( X_i(s,t) - \overline{X}(s,t)
  \right)\left( X_i(s',t) - \overline{X}(s',t)\right)dt  = \int_0^T c_N(s,t,s',t) dt,
  \\  \tilde c_{2,N}(t,t') &= \frac{1}{N} \sum_{i=1}^N \int_{[-S,S]^d} \left( X_i(s,t) - \overline{X}(s,t)
  \right)\left( X_i(s,t') - \overline{X}(s,t')\right)ds  = \int_{[-S,S]^d} c_N(s,t,s,t') ds,
\\ \overline{X}(s,t) &= \frac{1}{N} \sum_{i=1}^N X_i(s,t),  \widehat{c}_N(s,t,s',t') = \frac{1}{N} \sum_{i=1}^N \left(
X_i(s,t) - \overline{X}(s,t) \right)\left( X_i(s',t') -
\overline{X}(s',t')\right),
  \end{align*}
  for all $s, s' \in [-S,S]^d, t,t' \in [0,T]$.
The assumption of separability here means that the estimated covariance is
written as a product of a purely spatial component and a purely
temporal component, thus making both modeling and estimation easier
in many practical applications.

We stress again that we aim to develop a test statistic that solely
relies on the estimation of the separable components $C_1$ and $C_2$,
and does not require the estimation of the full covariance $C$. We
can expect that under the null hypothesis $H_0 : C = C_1 \kprod
C_2$, the difference $D_N = \widehat{C}_N - \widehat{C}_{1,N} \kprod
\widehat{C}_{2,N}$ between the sample covariance operator and its
separable approximation should take small values. We propose
therefore to construct our test statistic by projecting $D_N$ onto
the first eigenfunctions of $C$, since these encode the directions along
which $X$ has the most variability. If we denote by $C_1 = \sum_{i
\geq 1} \lambda_i u_i \tensort u_i$ and $C_2 = \sum_{j \geq 1}
\gamma_j v_j \tensort v_j$ the Mercer decompositions of $C_1$ and
$C_2$, we have
\[
  C = C_1 \kprod C_2 = \sum_{i,j \geq 1} \lambda_i \gamma_j (u_i \tensor v_j) \btensort (u_i \tensor v_j),
\]
where we have used results from Appendix~\ref{app:Hilbert-spaces}.
The eigenfunctions of $C$ are therefore of the form $u_r \tensor
v_s$, where $u_r \in H_1$ is the $r$-th eigenfunction of $C_1$ and
$v_s \in H_2$ is the $s$-th eigenfunction of $C_2$. We define a test
statistic based on the projection
\begin{equation}
  T_N(r,s) = \sqrt{N}\sc{D_N (\hat u_r \tensor \hat v_s), \hat u_r \tensor \hat v_s}, \quad r,s \geq 1 \text{
    fixed},
  \label{eq:test-statistic-projected}
\end{equation}
where we have replaced the eigenfunctions of $C_1$ and $C_2$ by
their empirical counterpart, i.e.\ the Mercer decompositions of
$\widehat{C}_{1,N}$, respectively $\widehat{C}_{2,N}$, are given by
 $\widehat{C}_{1,N} = \sum_{i \geq 1} \hat \lambda_i \hat u_i \tensor \hat u_i$, respectively
$\widehat{C}_{2,N} = \sum_{j \geq 1} \hat \gamma_j \hat v_j \tensor
\hat v_j$. Notice that though the eigenfunctions of
$\widehat{C}_{1,N}$ and $\widehat{C}_{2,N}$ are defined up to a
multiplicative constant $\alpha = \pm 1$, our test statistic is well
defined. The key fact for the practical implementation of the method
is that $T_N(r,s)$ can be computed without the need to estimate (and
store in memory) the operator $D_N$, since
  $
    T_N(r,s) = \sqrt{N} \left( \frac{1}{N} \sum_{k=1}^N \sc{ X_k-\overline{X}_N, \widehat{v}_i \otimes
        \widehat{u}_j}^2 - \hat
      \lambda_r \hat \gamma_s\right).
  $
In particular, the computation of $T_N(r,s)$ does \emph{not} require an estimation of the full covariance
operator $C$, but only the estimation of the marginal covariance operators $C_1$ and $C_2$, and their
eigenstructure.

\subsection{{Asymptotics}}
The theoretical justification for using a projection of $D_N$ to define a
test procedure is that, under the null hypothesis $H_0: C = C_1
\kprod C_2$, we have $\tnorm{D_N} \convp 0$ as $N \rightarrow \infty$, i.e.\ $D_N$  convergences in probability to zero
with respect to the trace norm. In fact, we will show in
Theorem~\ref{thm:general-asymptotic-distn} that $\sqrt{N}D_N$  is
asymptotically Gaussian under the following regularity conditions:
\begin{cond}
  \label{cond:asymptotics}
 $X$ is a random element of the real Hilbert space $H$ satisfying
  \begin{equation}
    \sum_{j = 1}^\infty \left( \eee{ \sc{X, e_j}^4 } \right)^{1/4} < \infty,
    \label{eq:condition-asymptotics}
  \end{equation}
  for some orthonormal basis $(e_j)_{j \geq 1}$ of $H$.
\end{cond}
The implications of this condition can be better understood in light of
the following remark.
\begin{rmk}[\citet{mas2006sufficient}] \mbox{}
  \begin{enumerate}
    \item Condition \ref{cond:asymptotics} implies
      that $\ee{ \hnorm{X}^4 } < \infty$.
    \item If $\ee{ \hnorm X ^4 } < \infty$, then $\sqrt{N}(C_N - C)$
      converges in distribution to a Gaussian random element of $\HS{H}$ for $N \rightarrow \infty$, with respect to the
      Hilbert--Schmidt topology. Under Condition~\ref{cond:asymptotics}, a stronger form of convergence holds: $\sqrt{N}(C_N - C)$ converges in
      distribution to a random element of $\tc{H}$ for $N \rightarrow \infty$, with respect to the trace-norm topology.
    \item If $X$ is Gaussian and $(\lambda_j)_{j \geq 1}$ is the sequence of eigenvalues of its covariance
      operator, a sufficient condition for \eqref{eq:condition-asymptotics} is $\sum_{j \geq 1 }
      \sqrt{\lambda_j} < \infty$.
  \end{enumerate}
\end{rmk}
Condition~\ref{cond:asymptotics} requires fourth order moments rather than the usual second order moments often assumed in functional data, as in this case we are interested in investigating the variation of the second moment, and hence require assumptions on the fourth order structure.
Recall that $ \widehat{C}_N =
\frac{1}{N} \sum_{j = 1}^N (X_i - \overline{X}) \tensort (X_i -
\overline{X})$, where $\overline X = N^{-1} \sum_{k = 1}^N X_k$. The
following result establishes the asymptotic distribution of $D_N =
\widehat{C}_N - \frac{\trace_2(\widehat{C}_N)\kprod
\trace_1(\widehat{C}_N)}{\trace(\widehat{C}_N)}$:
\begin{thm}
  \label{thm:general-asymptotic-distn}
  Let $H_1, H_2$ be separable real Hilbert spaces,   $X_1,\ldots, X_N \sim X$ be i.i.d. random elements on
  $H_1 \tensor H_2$ with covariance operator $C$, and $\trace{C} \neq 0$.

  If $X$ satisfies Condition~\ref{cond:asymptotics} (with $H = H_1 \tensor H_2$),
  then, under the null hypothesis
  \[
    H_0 : C = C_1 \kprod C_2, \qquad C_1 \in \tc{H_1}, C_2 \in \tc{H_2},
  \]
  we have
  \begin{equation}
    \sqrt{N}\left(  \widehat{C}_N - \frac{\trace_2( \widehat{C}_N) \kprod \trace_1( \widehat{C}_N)}{\trace( \widehat{C}_N)} \right)
    \convd Z, \quad \text{as } N \rightarrow \infty,
    \label{eq:conv-in-d-general}
  \end{equation}
  where $Z$ is a Gaussian random element of $\tc{H_1 \tensor H_2}$ with mean zero, whose covariance structure
  is given in Lemma~\ref{lma:general-asymptotic-covariance}.
\end{thm}
Condition~\ref{cond:asymptotics} is used here because we need
$\sqrt{N}(\widehat{C}_N - C)$ to converge in distribution in the
topology of the space $\tc{H_1 \tensor H_2}$; it could be replaced
by any (weaker) condition ensuring such convergence. The assumption
$\trace{C} \neq 0$ is equivalent to assuming that $X$ is not almost
surely constant.
\begin{proof}[Proof of Theorem~\ref{thm:general-asymptotic-distn}]
  First, notice that $C = C_1 \kprod C_2 = \frac{\trace_2( C) \kprod \trace_1( C)}{\trace( C)}$ under $H_0$.
   Therefore, using the linearity of the partial trace, we get
  \begin{align*}
    \sqrt{N}\left(  \widehat{C}_N - \frac{\trace_2( \widehat{C}_N) \kprod \trace_1( \widehat{C}_N)}{\trace( \widehat{C}_N)} \right)
    &=  \sqrt{N}(  \widehat{C}_N - C )
    \\ & \quad +  \sqrt{N} \left(\frac{\trace_2( C) \kprod \trace_1( C)}{\trace( C)} + \frac{\trace_2( \widehat{C}_N) \kprod \trace_1( \widehat{C}_N)}{\trace( \widehat{C}_N)} \right)
    \\ &=  \sqrt{N}(  \widehat{C}_N - C ) + \frac{\trace\left( \sqrt{N}(\widehat{C}_N - C)
      \right) C }{\trace(\widehat{C}_N)}
    \\ & \quad -   \frac{\trace_2\left( \sqrt{N}(\widehat{C}_N - C) \right) \kprod \trace_1( C)}{\trace( \widehat{C}_N)}
    \\ & \quad - \frac{\trace_2( \widehat{C}_N) \kprod \trace_1\left( \sqrt{N}(\widehat{C}_N - C) \right)}{\trace( \widehat{C}_N)}.
    \\ &= \Psi\left( \sqrt{N}(\widehat{C}_N - C), \widehat{C}_N \right),
  \end{align*}
  where
  \[
    \Psi(T,S) = T + \frac{\trace(T)C}{\trace(S)} - \frac{\trace_2(T) \kprod \trace_1(C)}{\trace(S)} -
    \frac{\trace_2(S) \kprod \trace_1(T)}{\trace(S)};
  \quad T, S \in \tc{H_1 \tensor H_2}.
  \]
  Notice that the function $\Psi: \tc{H_1 \tensor H_2} \times \tc{H_1 \tensor H_2} \rightarrow \tc{H_1 \tensor
  H_2}$ is continuous at $(T,S) \in \tc{H_1 \tensor H_2} \times \tc{H_1 \tensor H_2}$ in each coordinate,
  with respect to the trace norm, provided  $\trace(S) \neq 0$.
Since $\sqrt{N}(\widehat{C}_N - C)$ converges in
distribution---under Condition~\ref{cond:asymptotics}---to a
Gaussian random element $Y \in \tc{H_1 \tensor H_2}$, with respect
to the trace norm $\tnorm{\cdot}$
\citep[see][Proposition~5]{mas2006sufficient}, $\Psi\left(
\sqrt{N}(\widehat{C}_N - C), \widehat{C}_N \right)$ converges in
distribution to
\begin{equation}
  \label{eq:limiting-distn-in-proof}
  \Psi(Y,C) = Y + \frac{\trace(Y)C}{\trace(C)} - \frac{\trace_2(Y) \kprod \trace_1(C)}{\trace(C)} -
  \frac{\trace_2(C) \kprod \trace_1(Y)}{\trace(C)}
\end{equation}
by the continuous mapping theorem in metric spaces
\citep{billingsley:1999}. $\Psi(Y, C)$ is Gaussian because each of
the  summands of \eqref{eq:limiting-distn-in-proof} are Gaussian.
Indeed, the first and second summands are obviously Gaussian, and
the last two summands are Gaussian by
Proposition~\ref{prop:partial-trace-of-Gaussian}, and
Proposition~\ref{prop:kronecker-product-of-gaussian-and-fixed-operator}.
\end{proof}
We can now give the asymptotic distribution of $T_N(r,s)$, defined
in \eqref{eq:test-statistic-projected} as the (scaled) projection of
$D_N$ in a direction given by the tensor product of the empirical
eigenfunctions $\hat{u}_r$ and $\hat{v}_s$. The proof of the
following result is given in Appendix~\ref{app:proofs}.
\begin{cor}
  \label{cor:asymptotic-distribution-Tn-general-case}
  Under the conditions of Theorem~\ref{thm:general-asymptotic-distn}, if
  $\mathcal I \subset \left\{ (i,j) : i,j \geq 1 \right\}$ is a \emph{finite} set of indices such that
  $\lambda_r \gamma_s > 0$ for each $(r,s) \in \mathcal I$, then
  \begin{equation*}
    \left( T_N(r,s) \right)_{(r,s) \in \mathcal I} \convd N(0, \Sigma), \quad \text{as } N \rightarrow \infty.
  \end{equation*}
  \noindent This means that the vector  $\left( T_N(r,s) \right)_{(r,s) \in \mathcal I}$ is asymptotically
  multivariate Gaussian, with asymptotic variance-covariance matrix $\Sigma = \left( \Sigma_{ (r,s),
        (r',s')  } \right)_{(r,s), (r',s') \in \mathcal I}$ is given by
  \begin{align*}
    \Sigma_{ (r,s), (r',s')  } &= \tb_{rsr's'}
    + \frac{\alpha_{rs}\tb_{r's'\cdot \cdot}  + \alpha_{r's}\tb_{r \cdot \cdot s'} + \alpha_{rs'}\tb_{r' \cdot \cdot s} + \alpha_{r's'}\tb_{rs\cdot \cdot}}{\trace(C)}
    \\ & \qquad +  \frac{\alpha_{rs}\alpha_{r's'} \tb_{\cdot\cdot\cdot\cdot}}{ \trace(C)^2}
     + \frac{\lambda_r \lambda_{r'} \tb_{\cdot s \cdot s'}}{\trace(C_1)^2}
 + \frac{\gamma_s \gamma_{s'} \tb_{r \cdot r' \cdot }}{\trace(C_2)^2}
    \\ & \qquad - \frac{\lambda_r \tb_{r's'\cdot s} + \lambda_{r'} \tb_{r s \cdot s'}}{\trace(C_1)}
     - \frac{\gamma_s \tb_{r's' r \cdot } + \gamma_{s'} \tb_{r s r' \cdot } }{\trace(C_2)}
     \\ & \qquad - \frac{\alpha_{rs}}{\trace(C)}
      \left( \frac{\gamma_{s'} \tb_{r' \cdot \cdot \cdot
          }}{\trace(C_2)} + \frac{\lambda_{r'} \tb_{\cdot s' \cdot \cdot }}{\trace(C_1)} \right)
    \\ & \qquad - \frac{\alpha_{r's'} }{\trace(C) } \left( \frac{\gamma_{s} \tb_{r \cdot \cdot \cdot
         }}{\trace(C_2)} + \frac{\lambda_{r} \tb_{\cdot s \cdot \cdot }}{\trace(C_1)} \right)
  \end{align*}
  where $\mu = \eee{X}$, $\alpha_{rs} = \lambda_r \gamma_s$,
  \[
    \tb_{ijkl} = \eee{\sc{X - \mu, u_i \tensor v_j}^2 \sc{X - \mu,
        u_k \tensor v_l}^2},
  \]
  and `$\,\cdot$' denotes summation over the corresponding index, i.e.\ $\tb_{r\cdot
    jk} = \sum_{i \geq 1} \tb_{rijk}$.
\end{cor}
We note that the asymptotic variance-covariance of  $\left( T_N(r,s) \right)_{(r,s) \in \mathcal I}$ depends on the
second and fourth order moments of $X$, which is not surprising
since it is based on estimators of the covariance of $X$. Under the
additional assumption that $X$ is Gaussian, the asymptotic variance-covariance
of $\left( T_N(r,s) \right)_{(r,s) \in \mathcal I}$ can be entirely expressed in terms of the covariance
operator $C$.
The proof of the following result is given in
Appendix~\ref{app:proofs}.
\begin{cor}
  \label{cor:asymptotic-distn-Tn-gaussian-case}
  Assume the conditions of Theorem~\ref{thm:general-asymptotic-distn} hold, and that $X$ is Gaussian.
 If
  $\mathcal I \subset \left\{ (i,j) : i,j \geq 1 \right\}$ is a \emph{finite} set of indices such that
  $\lambda_r \gamma_s > 0$ for each $(r,s) \in \mathcal I$, then
  \begin{equation*}
    \left( T_N(r,s) \right)_{(r,s) \in \mathcal I} \convd N(0, \Sigma), \quad \text{as } N \rightarrow \infty.
  \end{equation*}
  where
  \begin{align*}
    \Sigma_{(r,s), (r',s')} & = \frac{2 \lambda_r \lambda_{r'} \gamma_s \gamma_{s'}}{\trace(C)^2} \left( \delta_{r r'} \trace(C_1)^2 + \hsnorm{C_1}^2 - ( \lambda_r + \lambda_{r'}) \trace(C_1) \right)
      \\ & \qquad \times   \left( \delta_{s s'}\trace(C_2)^2 + \hsnorm{C_2}^2 - ( \gamma_s + \gamma_{s'})
        \trace(C_2) \right),
  \end{align*}
  and $\delta_{ij} = 1$ if $i=j$, and zero otherwise.
{In particular, notice that $\Sigma$ itself is separable.}
\end{cor}

It will be seen in the next section that even in the case where we use a bootstrap test, knowledge of the asymptotic distribution can be very useful to establish a pivotal bootstrap test, which will be seen to have very good performance in simulation.

\section{Separability Tests and Bootstrap Approximations}
\label{sec:bootstrap}

In this section we use the estimation procedures and the theoretical
results presented in Section \ref{s:SepCovs} to develop a test for
$H_0 : C = C_1 \kprod C_2$, against the alternative that $C$ cannot
be written as a tensor product.

First, it is straightforward to
define a testing procedure when $X$ is Gaussian. Indeed, if we let
\begin{equation}
  \label{eq:asymptotically-sigma2-times-chi2-test}
G_N(r,s) = T_N^2(r,s)
= N \left( \frac{1}{N} \sum_{k=1}^N \sc{ X_k-\overline{X}, \widehat{u}_r \otimes \widehat{v}_s}^2 - \hat
  \lambda_r \hat \gamma_s\right)^2,
\end{equation}
 and
\begin{multline}
  \hat \sigma^2(r,s) = \left( {\trace(\widehat{C}_{1,N})^2\trace(\widehat{C}_{2,N})^2 } \right)^{-1}  2 \hat
    \lambda_r^2 \hat \gamma_s^2
    \\ \times \left( \trace(\widehat{C}_{1,N})^2 + \hsnorm{\widehat{C}_{1,N}}^2 - 2
      \hat \lambda_r \trace(\widehat{C}_{1,N}) \right)
    \\ \times \left( \trace(\widehat{C}_{2,N})^2 + \hsnorm{\widehat{C}_{2,N}}^2 - 2 \hat \gamma_s
      \trace(\widehat{C}_{2,N}) \right),
\end{multline}
then $\hat \sigma^{-2}(r,s) G_N(r,s)$ is asymptotically $\chi^2_1$
distributed, and $\{G_N^2(r,s) > \hat \sigma^2(r,s)
\chi^2_{1}(1-\alpha)\}$, where $\chi^2_1(1-\alpha)$ is the
$1-\alpha$ quantile of the $\chi^2_1$ distribution, would be a
rejection region of level approximately $\alpha$, for $\alpha \in
[0,1]$ and $N$ large.

Apart for the distributional assumption for $X$ to be Gaussian, this
approach suffers also the important limitation that it only tests
the separability assumption along \emph{one} eigendirection. It is
possible to extend this approach to take into account several
eigendirections. {
  For simplicity, let us consider the case $\mathcal I = \left\{ 1, \ldots, p \right\} \times \left\{ 1,
    \ldots, q \right\}$. Denote by $\bT_N(\mathcal I)$ the $p \times q$ matrix with entries
  $( \bT_N(\mathcal I) )_{ij} = T_N(i,j)$, and let
    \begin{equation}
    \label{eq:test-stat-multiple-directions-stud-full}
      \widetilde G_N(\mathcal I) = \left| \hat \Sigma_{L,\mathcal I}^{-1/2} \bT_N(\mathcal I) \hat
      \Sigma_{R,\mathcal I}^{-\tp/2}  \right|^2,
    \end{equation}
  where $|A|^2$ denotes the sum of squared entries of a matrix $A$, $A^{-1/2}$ denotes the inverse of (any)
  square root of the matrix $A$, $A^{-\tp/2} = (A^{-1/2})^\tp$, and the matrices $\hat
  \Sigma_{L,\mathcal I}$, respectively $\hat \Sigma_{R,\mathcal I}$, which are estimators of the row, resp. column,
  asymptotic covariances of $\bT_N(\mathcal I)$,  are defined in Appendix~\ref{app:implementation}.
  Then $\widetilde G_N(\mathcal I)$ is asymptotically $\chi^2_{pq}$ distributed.
  In the
simulation studies (Section~\ref{sec:simulations}), we consider also an approximate version of
this Studentized test statistics,
  $
    \widetilde G_N^a(\mathcal I) = \sum_{(r,s) \in \mathcal I} T_N^2(r,s)/\hat \sigma^2(r,s),
  $
which are  obtained simply by standardizing marginally each entry $T_N^2(r,s)$,
thus ignoring the dependence between the test statistics associated with different directions.
In order to assess the advantage of Studentization, we also consider the
non-Studentized test statistic
\[
  G_N(\mathcal I) = \sum_{(r,s) \in \mathcal I} T_N^2(r,s).
\]
 The computation details for $\widetilde G_N$,
 $T_N$, $\hat \sigma^2(r,s)$, $\hat \Sigma_{L, \mathcal I}$ and $\hat \Sigma_{R, \mathcal I}$ are described in
Appendix~\ref{app:implementation}.
}

\begin{rmk}
  Notice that the only test whose asymptotic distribution is parameter free is $\widetilde G_N(\mathcal
  I)$, under Gaussian assumptions. It would in principle be possible to construct an analogous test without
  the Gaussian assumptions (using Corollary~\ref{cor:asymptotic-distribution-Tn-general-case}). However, due to the
  large number of parameters that would need to be estimated in this case, we expect the asymptotics to come
  into force only for very large sample sizes (this is actually the case under Gaussian assumptions, specially
  if the set of projections $\mathcal I$ is large, as can be seen in
  Figure~\ref{fig:simulations-gaussian-and-t-R3}). For these reasons, we shall
  investigate bootstrap approximations to the test statistics.
\end{rmk}

{The choice of the number of eigenfunctions $K$ (the number
of elements in $\mathcal I$) onto which one should project is not
trivial. The popular choice of including enough eigenfunctions to
explain a fixed percentage of the variability in the dataset may seem
inappropriate in this context, because under the alternative
hypothesis there is no guarantee that the separable eigenfunctions
explain that percentage of {variation}.

For fixed $K$, notice that the test at least guarantees the separability in the subspace of the respective $K$ eigenfunctions, which is where the following analysis will be often focused. On the other hand, since our test statistic looks at an estimator of the non-separable component
\[
  D = C - \frac{\trace_2(C) \kprod \trace_1(C)}{\trace(C)},
\]
restricted to the
subspace spanned by the eigenfunctions $u_r \tensor v_s$, the test takes small values (and thus lacks power) when
\[
  \sc{D (u_r \tensor v_s), u_r \tensor v_s} = \HSsc{D, (u_r \tensort u_r) \kprod (v_s \tensort v_s)} = 0,
\]
that is when the non-separable component $D$ is orthogonal to
\[
  (u_r \tensort u_r) \kprod (v_s \tensort v_s)
\]
with respect to the Hilbert--Schmidt inner product.
Thus the proposed
test statistic $G_N(\mathcal I)$ is powerful when $D$ is
not orthogonal to the subspace
\[
  V_{\mathcal I} = \mathrm{span}\{ (u_i \tensort u_i) \kprod (v_j \tensort v_j), (i,j)\in \mathcal{I}\},
\]
and in general the power of the test for finite
sample size depends on the properly rescaled norm  of the projection of $D$ onto $V_{\mathcal I}$.

In practice, it seems reasonable to use the subset of eigenfunctions
that it is possible to estimate accurately given the available
sample sizes. The accuracy of the estimates for the eigendirections
can be in turn evaluated with bootstrap methods, see e.g.\
\citet{Hall2006} for the case of functional data. A good strategy
may also be to consider more than one subset of eigenfunctions and
then summarize the response obtained from the different tests using
a Bonferroni correction.}

As an alternative to these test statistics (based on
projections of $D_N = C_N - C_{1,N} \kprod C_{2,N}$), we consider
also a test based on the squared Hilbert--Schmidt norm
of $D_N$, i.e.\ $\hsnorm{D_N}^2$, whose null distribution will be approximated by a bootstrap procedure (this
test will be referred to as
\emph{Hilbert--Schmidt test} hereafter). Though it seems that such
tests would require one to store the full sample covariance of the
data (which could be infeasible), we describe in Appendix~\ref{app:implementation}
 a way of circumventing such problem, although the computation of each entry  of
the full covariance is still needed.  Therefore this could be used only for applications in which the
dimension of the discretized covariance matrix is not too large.

In the following, we propose also a bootstrap approach to approximate the distribution of the test
statistics  $\widetilde{G}_N(\mathcal{I})$, $\widetilde{G}^a_N(\mathcal I)$ and $G_N(\mathcal{I})$, with the aim
to improve the finite sample  properties of the procedure and to relax the distributional assumption on $X$.

\subsection{Parametric Bootstrap}

If we assume we know the distribution of $X$ up to its mean
$\mu$ and its covariance operator $C$, i.e.\ $ X \sim F(\mu; C)$,  we can approximate
the distribution of $\widetilde G_N(\mathcal I)$, $\widetilde{G}^a_N(\mathcal I)$, $G_N(\mathcal I)$  and
$\hsnorm{D_N}^2$ under the separability
hypothesis via a parametric bootstrap procedure.
Since
$C_{1,N} \kprod C_{2,N}$, respectively $\overline X$,  is an estimate of $C$, respectively $\mu$,
we simulate $B$ bootstrap samples
$X^b_1,\ldots,X^b_N \simiid F\left( \:\overline X, C_{1,N} \kprod C_{2,N}\right)$, for $b=1,\ldots,B$. For each sample, we compute
$H_N^b = H_N(X_1^b, \ldots, X_N^b)$, where $H_N = G_N(\mathcal I)$, $H_N = \widetilde
G_N(\mathcal I)$, $H_N = \widetilde
G^a_N(\mathcal I)$ respectively $H_N = \hsnorm{D_N}^2$, if we wish to use the non-Studentized projection test,
the Studentized projection test, the approximated Studentized version or the Hilbert--Schmidt test, respectively.
A formal description of the algorithm for obtaining the
$p$-value of the test based on the statistic $H_N = H_N(X_1,\ldots,
X_N)$ with the parametric bootstrap can be found in Appendix~\ref{app:implementation}, along with the details
for the computation of
$H_N$. We highlight that this procedure does not ask for the
estimation of the full covariance structure, but only of its
separable approximation, with the exception of the Hilbert--Schmidt
test (and even in this case, it is possible to avoid the storage of
the full covariance).

\subsection{Empirical Bootstrap}

In many applications it is not possible to assume a distribution for
the random element $X$, and a non-parametric approach is therefore
needed. In this setting, we can use the empirical bootstrap to
estimate the distribution of the test statistic
$G_N(\mathcal I), \widetilde G_N(\mathcal I)$ or $\hsnorm{D_N}^2$ under the null
hypothesis $H_0 : C = C_1 \kprod C_2$. Let $H_N$ denote the test statistic whose distribution is of interest.
Based on an i.i.d. sample $X_1, \ldots, X_N \sim X$, we wish to approximate the distribution of $H_N$ with
 the distribution of some test statistic $\Delta_N^* = \Delta_N(X_1^*, \ldots, X_N^*)$,
where $X_1^*, \ldots, X_N^*$ is obtained by drawing with replacement from the set $\left\{ X_1,\ldots, X_N
\right\}$.
Though it is tempting to use $\Delta^*_N = H_N(X_1^*, \ldots, X_N^*)$, this is not an appropriate choice.
Indeed, let us look at the case $H_N = G_N(i,j)$. Notice that the true covariance of $X$ is
\begin{equation}
  \label{eq:general-C-with-non-separable-part}
  C = \frac{\trace_2(C) \kprod \trace_1(C)}{\trace(C)} + D,
\end{equation}
where  $D$ is a possibly non-zero operator, and that
\[
  H_N^* = G_N(i,j|X_1^*, \ldots, X_N^*) = N \sc{(C_N^* - C_{1,N}^* \kprod C_{2,N}^*) (\hat u_i \tensor \hat v_j), \hat u_i \tensor
    \hat  v_j}^2,
\]
where $C_N^*=C_N(X_1^*, \ldots, X_N^*), C_{1,N}^*=C_{1,N}(X_1^*, \ldots, X_N^*),$ and $C_{2,N}^*=C_{2,N}(X_1^*,
\ldots, X_N^*)$.
Since $(C_N^* - C_{1,N}^* \kprod C_{2,N}^*) \approx (C_N - C_{1,N} \kprod C_{2,N}) \approx D$, the statistic
$H_N^*$ would approximate the distribution of $H_N$ under the hypothesis
\eqref{eq:general-C-with-non-separable-part}, which is not what we want.
We therefore propose the following choices of $\Delta_N^* = \Delta_n(X_1^*, \ldots, X_N^*; X_1, \ldots, X_N)$, depending on the choice of $H_N$:
\begin{enumerate}
  \item $H_N = G_N(\mathcal I)$, $\Delta_N^* = \sum_{(i,j) \in \mathcal I}  \left( T_N^*(i,j) - T_N(i,j)
    \right)^2$.
    \item {$H_N = \widetilde G_N(\mathcal I)$,
          $\Delta_N^* =
          \left| \left( \hat  \Sigma^*_{L,\mathcal I} \right)^{-1/2} \left(  \bT^*_N(\mathcal I) -
            \bT_N(\mathcal I) \right) \left( \hat \Sigma^*_{R,\mathcal I}  \right)^{-\tp/2} \right|^2,$
          where $\hat  \Sigma^*_{L,\mathcal I} = \hat  \Sigma_{L,\mathcal I}(X_1^*, \ldots, X_N^*)$, and
          $\hat  \Sigma^*_{R,\mathcal I} = \hat  \Sigma_{R,\mathcal I}(X_1^*, \ldots, X_N^*)$.
 are the row, resp. column, covariances estimated from the bootstrap sample.}
  \item $H_N = \widetilde G^a_N(\mathcal I)$,
      $\Delta_N^* = \sum_{(i,j) \in \mathcal I}  \left( T_N^*(i,j) - T_N(i,j)
      \right)^2/ \hat \sigma_*^2(i,j),$
    where $\hat \sigma_*^2(i,j) = \hat \sigma^2(i,j| X_1^*, \ldots, X_N^*)$.
  \item $H_N = \hsnorm{D_N}^2$, $\Delta_N^* = \hsnorm{D_N^* - D_N}^2$, where $D_N^* = D_N(X_1^*, \ldots,
    X_N^*)$.
\end{enumerate}

The algorithm to approximate the $p$-value of $H_N$ by the empirical
bootstrap is described in detail. The
basic idea consists of generating $B$ bootstrap samples, computing
$\Delta_N^*$ for each bootstrap sample and looking at the proportion
of bootstrap samples for which $\Delta_N^*$ is larger than the test
statistic $H_N$ computed from the original sample.

\section{Empirical demonstrations of the method}
\label{sec:results}

\subsection{Simulation studies}
\label{sec:simulations}
We investigated the finite sample behavior of our testing procedures through an intensive reproducible
simulation study (its running time is equivalent to approximately 401 days on a single CPU computer).
We compared the test based on the asymptotic distribution of \eqref{eq:asymptotically-sigma2-times-chi2-test},
as well as the tests based on $G_N(\mathcal I), \widetilde G_N(\mathcal I), \widetilde G^a_N(\mathcal I)$,  and $\hsnorm{D_N}^2$, with the
$p$-values obtained via the parametric bootstrap or the empirical bootstrap.

We generated discretized functional data $X_1,\ldots, X_N \in
\bR^{32 \times 7}$ under two scenarios. In the first scenario
(Gaussian scenario), the data were generated from a multivariate
Gaussian distribution $\mathcal N (0, \mathbf{C})$. In the second
scenario (Non-Gaussian scenario), the data were generated from a
multivariate $t$ distribution with $6$ degrees of freedom and non centrality parameter equal to zero.
In the Gaussian scenario, we set $\mathbf C = \mathbf{C}^{(\gamma)}$, where
\begin{multline}
\mathbf{C}^{(\gamma)}(i_1, j_1, i_2, j_2)= (1-\gamma)c_1(i_1,i_2)c_2(j_1,j_2) \\
+ \gamma\frac{1}{( j_1-j_2
  )^2+1}\exp\left\{-\frac{(i_1-i_2)^2}{(j_1-j_2)^2+1}\right\},
\end{multline}
$\gamma \in [0,1]; i_1, i_2 = 1, \ldots, 32; j_1, j_2 = 1,\ldots,
7$. The covariances $c_1$ and $c_2$ used in the simulations can be
seen in Figure~\ref{fig:c_xy}. For the Non-Gaussian scenario, we chose a multivariate $t$ distribution with
the correlation structure implied by $\mathbf{C}^{(\gamma)}$, $\gamma \in [0,1]; i_1, i_2 = 1, \ldots, 32;
j_1, j_2 = 1,\ldots, 7$.
The parameter $\gamma \in [0,1]$ controls the departure from the separability of the covariance
$\mathbf{C}^{(\gamma)}$: $\gamma = 0$ yields a separable covariance, whereas $\gamma = 1$ yields a complete
non-separable covariance structure \citep{Cressie1999}.
{All the simulations have been performed using the R package \texttt{covsep} \citep{R:covsep},
  available on CRAN, which implements the tests presented in the paper}.

\begin{figure}[h]
  \centerline{ \makebox{
    \includegraphics[scale=0.65, trim=10 50 0 50, clip=TRUE]{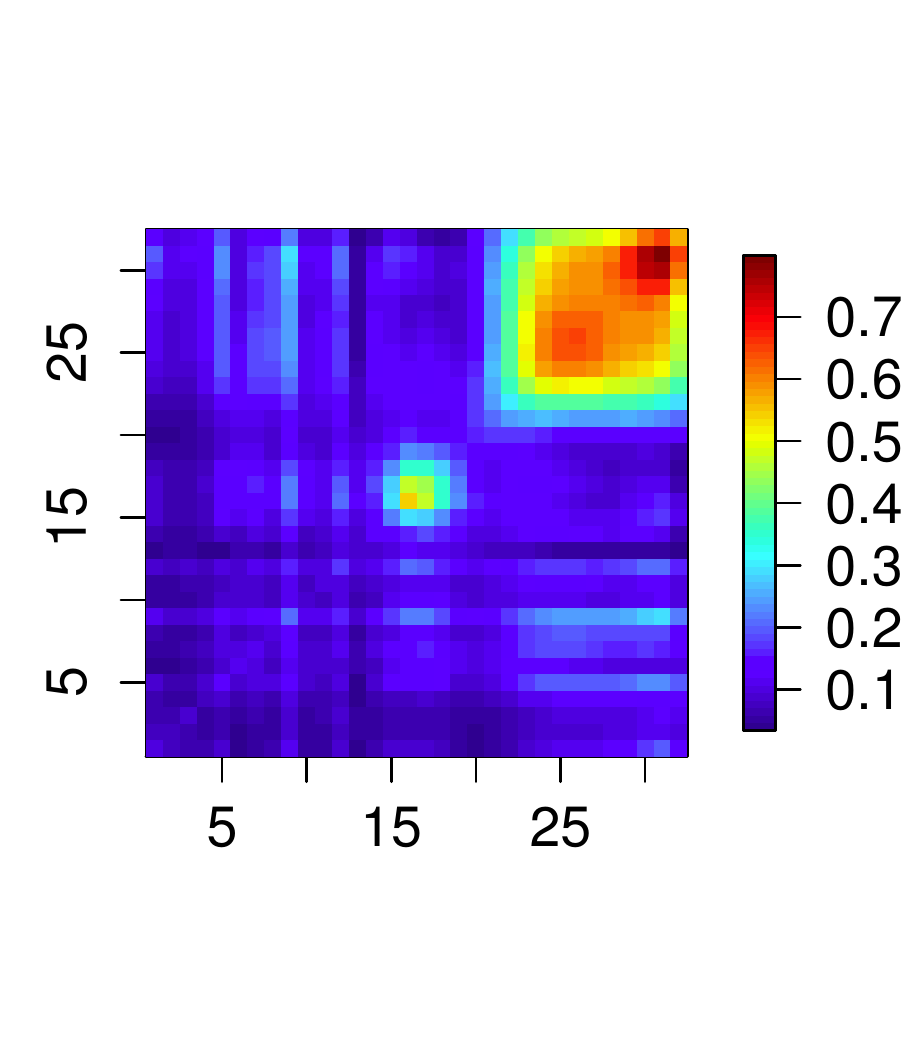}
    \hspace{.1cm}
    \includegraphics[scale=0.65, trim=0 50 0 50, clip=TRUE]{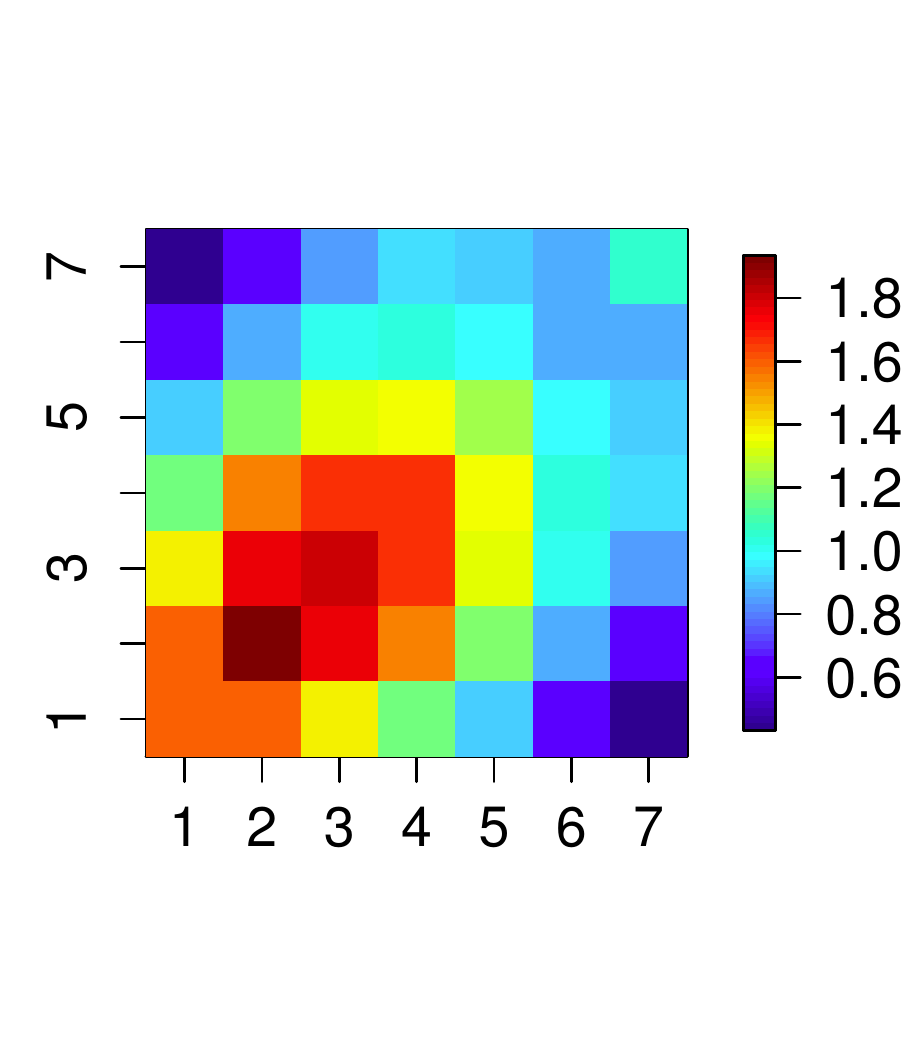}
  }
}
  \caption{Covariance functions $c_1$ (left) and $c_2$ (right) used in
    the simulation study.}\label{fig:c_xy}
\end{figure}

For each value of $\gamma \in \{0,0.01, 0.02, \ldots, 0.1\}$  and $N
\in \{10,25,50,100 \}$, we performed $1000$ replications for each of
the above simulations, and estimated the power of the tests based on
the asymptotic distribution of
\eqref{eq:asymptotically-sigma2-times-chi2-test}.

\begin{figure}[p]
  \centering 
  \begin{subfigure}[c]{\linewidth}
    {\includegraphics[width=\linewidth, trim=0 0 0 5, clip=TRUE]{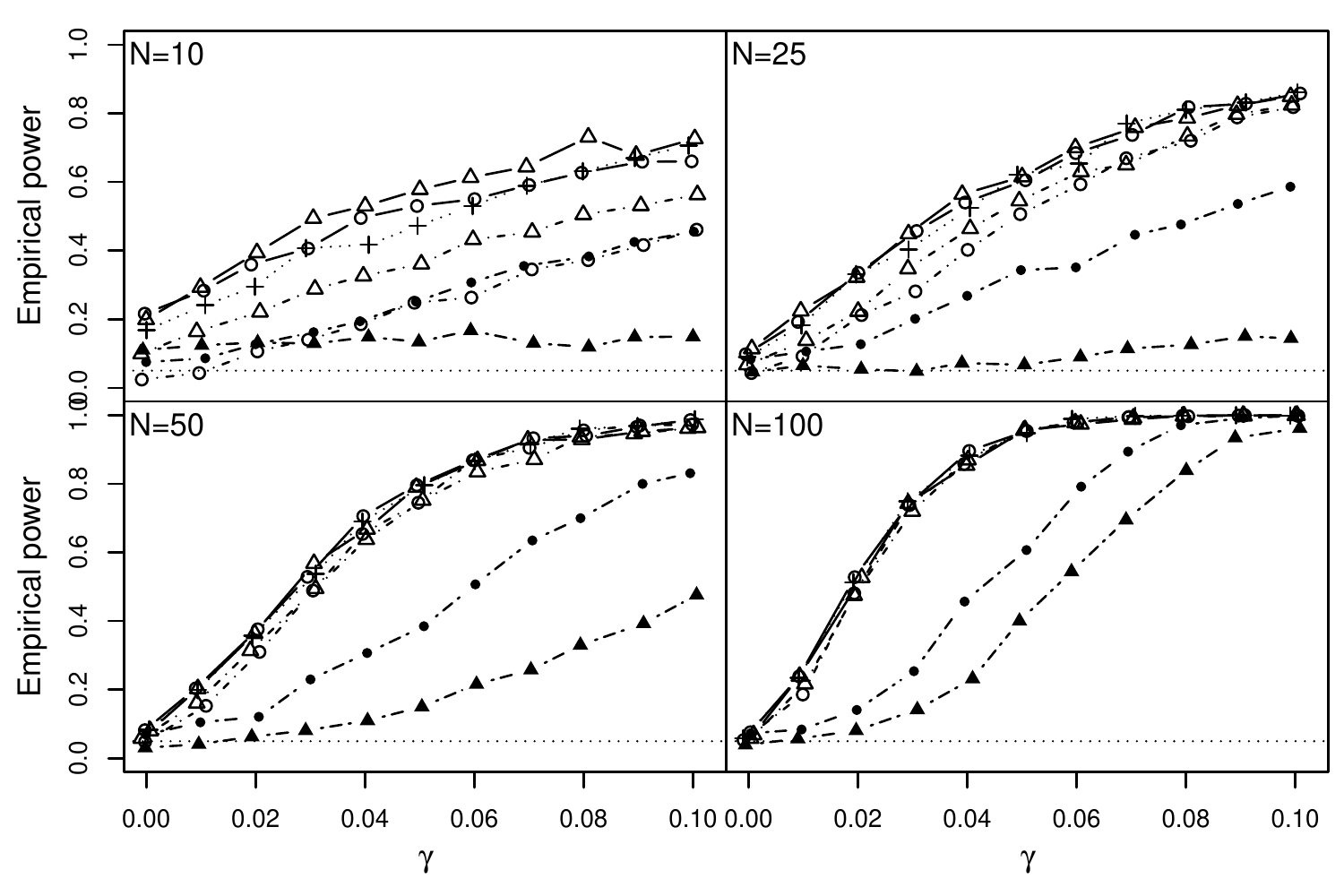}}
    \caption{Gaussian scenario}
  \end{subfigure}
  \begin{subfigure}[c]{\linewidth}
    {\includegraphics[width=\linewidth, trim=0 0 0 5, clip=TRUE]{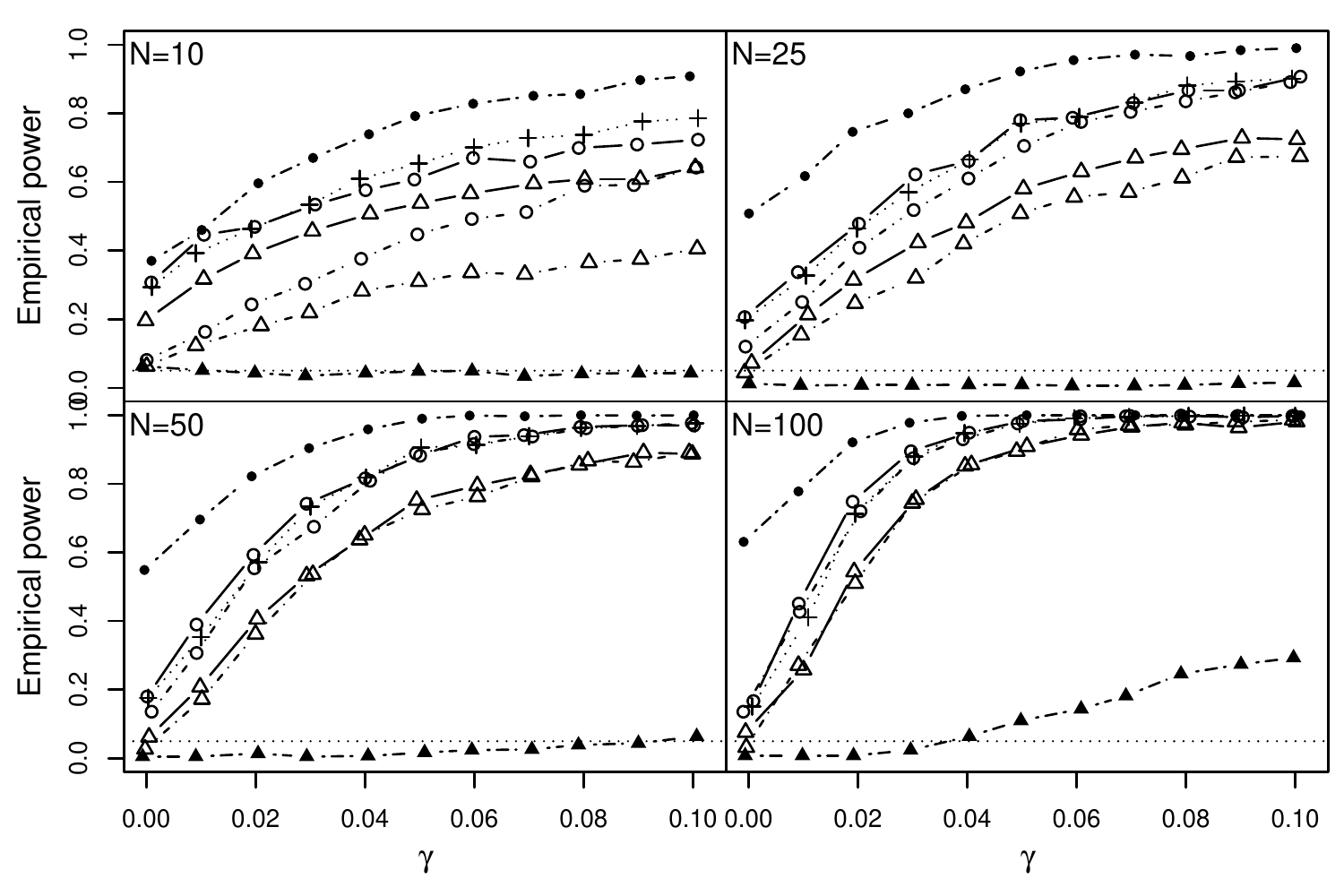}}
    \caption{Non Gaussian scenario}
  \end{subfigure}
  \caption{Empirical power of the testing procedures in the
    \emph{Gaussian} scenario (panel (a)) and \emph{non-Gaussian} scenario (panel (b)), for $N=10, 25, 50, 100$ and $\mathcal I = \mathcal I_1$. The
    results shown correspond to the test \eqref{eq:asymptotically-sigma2-times-chi2-test} based on its
    asymptotic distribution ($\cdot\!\cdot\!\cdot\!\cdot\!+\!\cdot\!\cdot\!\cdot\!\cdot$), the Gaussian
    parametric bootstrap test (solid line with empty circles) and its studentized version (dash-dotted line with
    empty circles), the empirical parametric bootstrap test (---$\bigtriangleup$---) and its Studentized
    version (-- --$\bigtriangleup$-- --), the Gaussian parametric Hilbert--Schmidt test (dash-dotted line with
    filled circles) and the empirical Hilbert--Schmidt test (dash-dotted line with filled triangles). The
    horizontal dotted line indicates the nominal level ($5\%$) of the test. 
  Note that the points have been horizontally jittered for better visibility.
  }
  \label{fig:simulations-gaussian-and-t-R1}
\end{figure}

We first also estimated
the power of the tests $\widetilde G_N(1,1)$, $G_N(1,1)$,  and
$\hsnorm{D_N}$, with distributions  approximated by a Gaussian
parametric bootstrap, and the empirical bootstrap, with $B = 1000$.
The results are shown in Figure~\ref{fig:simulations-gaussian-and-t-R1}.
In the Gaussian scenario (Figure~\ref{fig:simulations-gaussian-and-t-R1}, panel (a)), the
empirical size of all the proposed tests gets closer to the nominal
level ($5\%$) as $N$ increases (see also
Table~\ref{tab:empirical-size-r1}). Nevertheless, the non-Studentized tests
$G_N(1,1)$, for both parametric and empirical bootstrap, seem to
have a slower convergence with respect to the Studentized version,
and even for $N=100$ the level of these tests appear still higher
than the nominal one (and a CLT-based 95\% confidence interval for
the true level does not contain the nominal level in both cases).
The empirical bootstrap version of the Hilbert--Schmidt test also
fails to respect the nominal level at $N=100$, but its parametric
bootstrap counterpart respects the level, even for $N=25$. For
$N=25, 50, 100$, the most powerful tests (amongst those who respect
the nominal level) are the parametric and empirical bootstrap
versions of $\widetilde G_N(1,1)$, and they seem to have equal
power. The power of the Hilbert--Schmidt test based on the
parametric bootstrap seems to be competitive only for $N=100$ and
$\gamma=0.1$, and is much lower for other values of the parameters.
The test based on the asymptotic distribution does not respect the
nominal level for small $N$ but it does when $N$ increases. Indeed,
the convergence to the nominal level seems remarkably fast and its
power is comparable with those of the parametric and empirical
bootstrap tests based on $\widetilde G_N(1,1)$. Despite being based
on an asymptotic result, its performance is quite good also in
finite samples, and it is less computationally demanding than the
bootstrap tests.

In the non-Gaussian scenario
(Figure~\ref{fig:simulations-gaussian-and-t-R1}, panel (b)), only
the empirical bootstrap version of $\widetilde G_N(1,1)$ and of the Hilbert--Schmidt test seem to respect the
level for $N=10$ (see also
Table~\ref{tab:empirical-size-r1}).
Amongst these tests, the most powerful one is
clearly the empirical bootstrap test based on $\widetilde G_N(1,1)$.
Although the Gaussian parametric bootstrap test has higher empirical
power, it does not have the correct level (as expected) and thus
cannot be used in a non-Gaussian scenario.  Notice also that the
test based on the asymptotic distribution of $\widetilde G_N(1,1)$
(under Gaussian assumptions) does not respects the level of the test
even for $N=100$. The same holds for the Gaussian bootstrap version
of the Hilbert--Schmidt test. Finally, though the empirical
bootstrap version of the Hilbert--Schmidt test respects the level
for $N=10, 25, 50, 100$, it has virtually no power for $N= 10, 25,
50$, and has very low power for $N=100$ (at most $0.3$ for
$\gamma=0.1$).

As mentioned previously, there is no guarantee that a violation in
the separability of $C$ is mostly reflected in the first separable
eigensubspace. {Therefore, we consider also a larger
subspace for the test.
Figure~\ref{fig:simulations-gaussian-and-t-R2} shows the empirical
power for the asymptotic test, the parametric and empirical
bootstrap tests based on the test statistic $\widetilde
G_N(\mathcal{I}_2)$, as well as parametric and bootstrap tests based
on the test statistics $ G_N(\mathcal{I})$, $\widetilde
G^a_N(\mathcal{I}_2)$ where $\mathcal{I}_2= \{ (i,j) : i,j = 1,2 \}$. In
the Gaussian scenario, the asymptotic test is much slower in
converging to the correct level compared to its univariate version
based on $\widetilde G_N(1,1)$. For larger $N$ its power is
comparable to that of the parametric and empirical bootstrap based
on the Studentized test statistics $\widetilde G_N(\mathcal{I}_2)$,
which in addition respects the nominal level, even for $N=10$. It is interesting to note that the approximated
Studentized
bootstrap tests $\widetilde{G}_N^a(\mathcal{I}_2)$  have a performance
which is better than the non Studentized  bootstrap tests
$G_N(\mathcal{I}_2)$ but far worse than that of the Studentized tests $\widetilde G_N(\mathcal I_2)$.
The Hilbert--Schmidt test is again outperformed by all the other
tests, with the exception of the non-Studentized bootstrap test when
$N=10,25$. The results are similar for the non-Gaussian scenario,
apart for the fact that the asymptotic test does not respect the
nominal level (as expected, since it asks for $X$ to be Gaussian).} 

To investigate the difference between projecting on one or several
eigensubspaces, we also compare the power of the empirical bootstrap
version of the tests $\widetilde G_N(\mathcal I)$ for increasing
projection subspaces, i.e.\ for $\mathcal I = \mathcal I_l,
l=1,2,3$, where  $\mathcal I_1 = \{(1,1)\}, \mathcal I_2 = \{ (i,j)
: i,j = 1,2 \}$ and $\mathcal I _3 = \{ (i,j) : i = 1,\ldots, 4;
  j=1,\ldots, 10 \}$. The results are shown in
  Figure~\ref{fig:sim-Gaussian-increasing-proj} for the Gaussian scenario and
  Figure~\ref{fig:sim-tdist-increasing-proj} for the non-Gaussian scenario.
{In the Gaussian scenario, for  $N=10$, the most powerful test is
$\widetilde G_N(\mathcal I_2)$. In this case, projecting onto a
larger eigensubspace decreases the power of the test dramatically.
However, for $N\geq 25$ the power of the test is the largest for
$\widetilde G_N(\mathcal I_3)$, albeit only significantly larger than that of $\widetilde G_N(\mathcal I_2)$
when $\gamma=0.01$. Our interpretation is
that when the sample size is too small, including too many
eigendirection is bound to add only noise that degrades the
performance of the test. However, as long as the separable
eigenfunctions are estimated accurately, projecting in a larger eigenspace improves
the performance of test.}
See also Figure~\ref{fig:simulations-gaussian-and-t-R3} for 
the complete simulation results of the projection set $\mathcal I_3$.

\begin{figure}[h!]
\centering
\makebox{
\includegraphics[width=\linewidth, trim=0 0 0 5,
clip=TRUE]{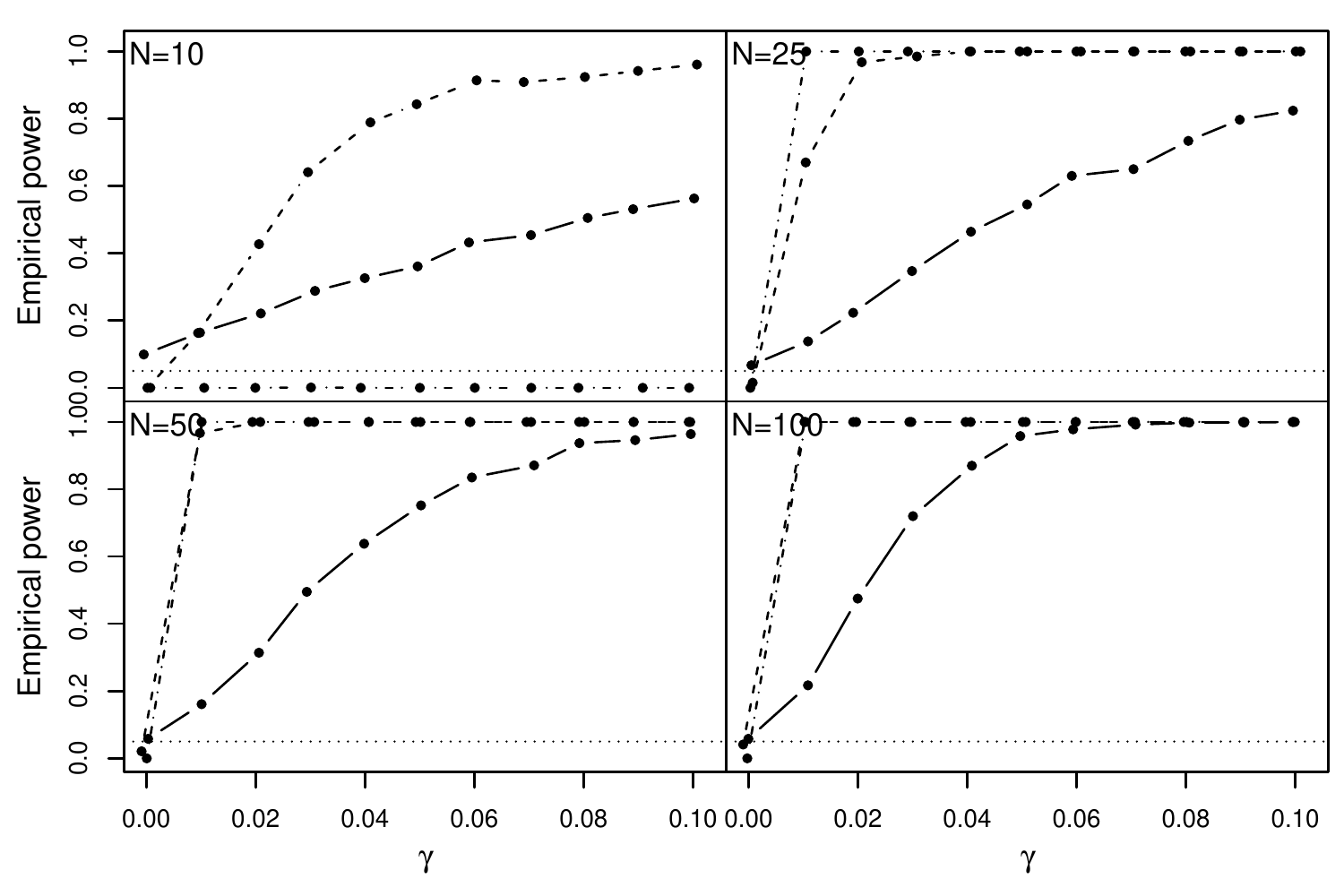}
} 
\caption{Empirical power of the empirical bootstrap version of $\widetilde G_N(\mathcal I_l)$, for $l = 1$
  (solid line), $l=2$ (dashed line) and $l=3$ (dash-dotted line), in the \emph{Gaussian} scenario.  The
  horizontal dotted line indicates the nominal level ($5\%$) of the test.
  Note that the points have been horizontally jittered for better visibility.
}
\label{fig:sim-Gaussian-increasing-proj}
\end{figure}

\begin{figure}[h]
  \begin{center}
\includegraphics[width=\linewidth, trim=0 0 0 5,
clip=TRUE]{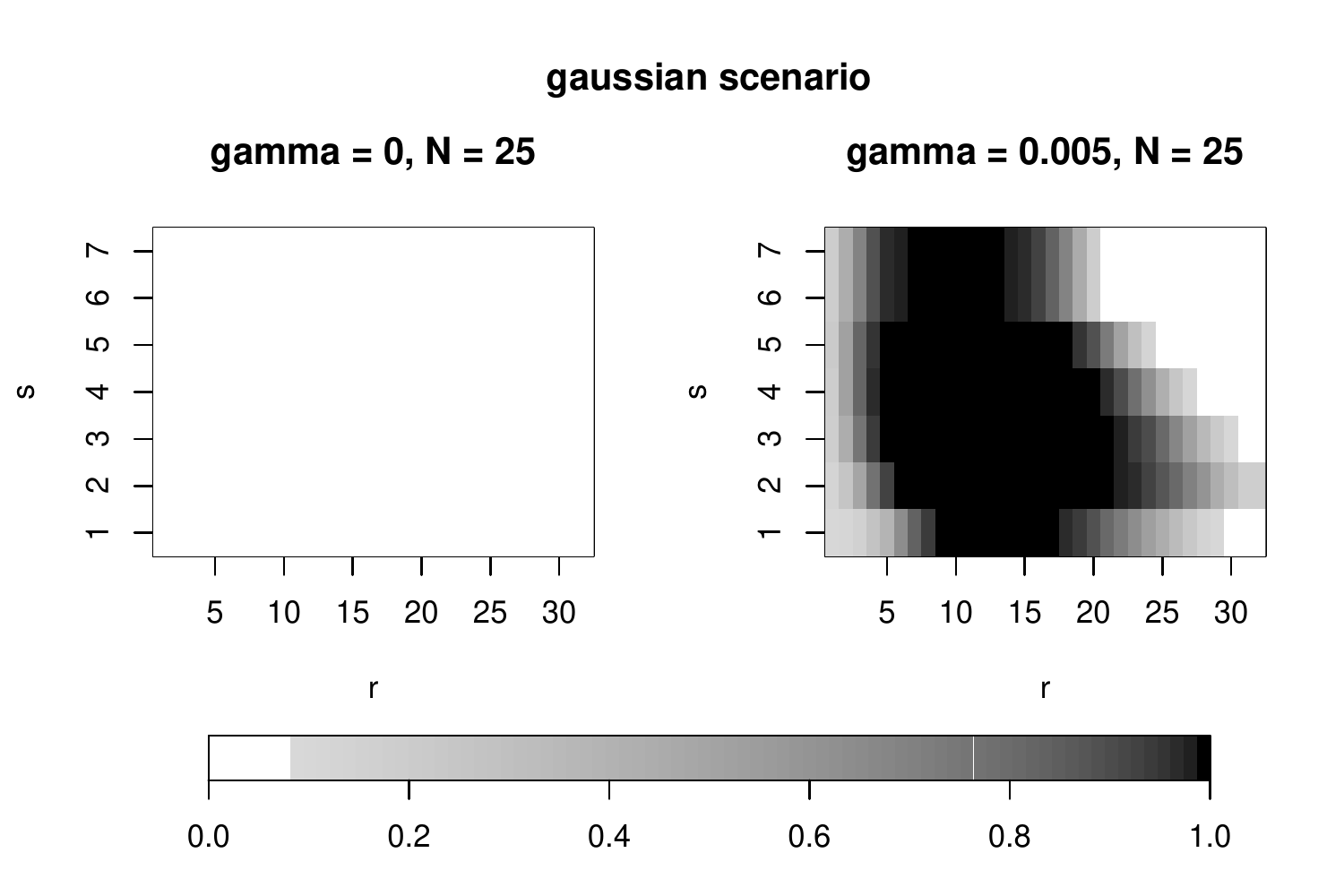}
  \end{center}
  \caption{Empirical size (left) and power (right) of the separability test as functions of the projection set
    $\mathcal I$. The test used is $\widetilde G_N(\mathcal I)$, with distribution approximated by the
    empirical bootstrap with $B=1000$. The left plot, respectively the right plot, was simulated from the Gaussian scenario with $\gamma
    =0$, respectively $\gamma = 0.005$, and $N=25$ . Each $(r,s)$ rectangle represents the level/power of the test based on the projection
    set $\mathcal I = \left\{ (i,j) : 1 \leq i \leq r, 1 \leq j \leq s \right\}$.}
  \label{fig:grid-level-power-gaussian-N25}
\end{figure}

This prompts us to investigate how the power of the test varies across all projection subsets
\[
  \mathcal I_{r,s} = \left\{ (i,j) : 1 \leq i \leq r, 1 \leq j \leq s \right\},
\]
$r=1,\ldots, 32, s=1,\ldots, 7$,=.
The test used is $\widetilde G_N(\mathcal I)$, with distribution approximated by the empirical bootstrap with
$B=1000$. Figure~\ref{fig:grid-level-power-gaussian-N25} shows the empirical size and
power of the separability test in the Gaussian scenario for sample size $N=25$, and
Figure~\ref{fig:grid-level-power-gaussian-3x2}, respectively
Figure~\ref{fig:grid-level-power-student-3x2}, shows the power for different sample sizes in the Gaussian scenario, respectively the
non-Gaussian scenario.

\subsubsection{Discussion of simulation studies}

The simulation studies above illustrate how the empirical bootstrap
test based on the test statistics $\widetilde G_N(\mathcal{I})$
usually outperforms its competitors, albeit it is also much more
computationally expensive than the asymptotic test, whose
performance are comparable in the Gaussian scenario for large enough
number of observations.

The choice of the best set of eigendirections
to use in the definition of the test \mbox{statistics} is difficult.
It seems that $K$ should be ideally chosen to
be increasing with $N$. This is reasonable, because larger values of
$N$ increase the accuracy of the estimation of the eigenfunctions
and therefore we will be able to detect departures from the
separability in more eigendirections, including ones
not only associated with the largest eigenvalues. However, the optimal rate
at which $K$ should increase with $N$ is still an open problem, and
will certainly depend in a complex way on the eigenstructure of the true underlying
covariance operator $C$.

This is confirmed by the results reported in Figure~\ref{fig:grid-level-power-gaussian-N25} and
Figures~\ref{fig:grid-level-power-gaussian-3x2} and
\ref{fig:grid-level-power-student-3x2}. These indeed show that taking into account too few
eigendirections can result in smaller power, while including too many of them can also decrease the power.


As an alternative to tests based on projections of $D_N$, the tests
based on the squared Hilbert--Schmidt norm of $D_N$, i.e.\
$\hsnorm{D_N}^2$, could potentially detect any departure from the
separability hypothesis---as opposed to the tests $\widetilde
G_N(\mathcal I)$. But as the simulation study illustrates, they
might be far less powerful in practice, particularly in situations
where the departure from separability is reflected in only in a few
eigendirections. Moreover, this approach still requires the
computation of the full covariance operator (although not its
storage) and is therefore not feasible for all applications.


\subsection{Application to acoustic phonetic data}\label{s:data}

An
interesting case where the proposed methods can be useful are
phonetic spectrograms. These data arise in the analysis of speech
records, since relevant features of recorded sounds can be better
explored in a two dimensional time-frequency domain.

In particular, we consider here the dataset of 23 speakers from five
different Romance languages that has been first described in
\citet{Pigoli2014}. The speakers were recorded while pronouncing the
words corresponding to the numbers from one to ten in their language
and the recordings are converted to a sampling rate of $16000$
samples per second. Since not all these words are available for all
the speakers, we have a total of $219$ speech records.  We focus on
the spectrum that speakers produce in each speech recording
$x^L_{ik}(t)$, where $L$ is the language, $i=1,\ldots,10$ the
pronounced word and $k=1,\ldots,n_L$ the speaker, $n_L$ being the
number of speakers available for language $L$. We then use a
short-time Fourier transform to obtain a two dimensional
log-spectrogram: we use a Gaussian window function $w(\cdot)$ with a window
size of $10$ milliseconds and we compute the short-time Fourier
transform as
$$
X^L_{ik}(\omega,t)=\int_{-\infty}^{+\infty}
x^L_{ik}(\tau)w(\tau-t)e^{-j\omega\tau}\mathrm{d}\tau.
$$
The spectrogram is defined as the magnitude of the Fourier transform
and the log-spectrogram (in decibel) is therefore
$$
\mathfrak{S}^L_{ik}(\omega,t)=10\log_{10}(|X^L_{ik}(\omega,t)|^2).
$$
The raw log-spectrograms are then smoothed \citep[with the robust
spline smoothing method proposed in][]{Garcia2010} and aligned in
time using an adaptation to 2-D of the procedure in
\citet{TangM2008}. The alignment is needed because a phase
distortion can be present in acoustic signals, due to difference in
speech velocity between speakers. Since the different words of each
language have different mean log-spectrograms, the focus of the
linguistic analysis---which is the study cross-linguistics changes---is on the residual log-spectrograms
$$
R_{ik}^L(\omega,t)=S_{ik}^L(\omega,t)- (1/n_{i})\sum_{k=1}^{n_i}
S_{ik}^L(\omega,t).
$$
Assuming that all the words within the language have the same
covariance structure, we disregard hereafter the information
about the pronounced words that generated the residual
log-spectrogram, and  use the surface data
$R_{j}^L(\omega,t), j=1,\ldots,N_L$, i.e.\ the set of observations
for the language $L$ including all speakers and words, for the separability test. These observations are
measured on an equispaced grid with $81$ points in the frequency
direction and $100$ points in the time direction. This translate on
a full covariance structure with about $33 \times 10^6$ degrees of
freedom. Thus, although the discretized covariance matrix is in
principle computable, its storage is a problem. More importantly, the
accuracy of its estimate is poor, since we have at most $50$
observations within each language. For these reasons, we would like
to investigate if a separable approximation of each covariance is appropriate.

We thus apply the Studentized version of the empirical bootstrap
test for separability to the residual log-spectrograms for each
language individually. Here, we take into consideration different
choices for set of eigendirections to be used in the definition of
the test statistic $\widetilde{G}_N(\mathcal{I})$,  namely
$\mathcal{I}=\mathcal{I}_1=\{(1,1)\}$,
$\mathcal{I}=\mathcal{I}_2=\{(r,s):1 \leq r \leq 2,1 \leq s \leq 3\}$,
$\mathcal{I}=\mathcal{I}_3=\{(r,s):1 \leq r \leq 8,1 \leq s \leq 10\}$. For all cases
we use $B=1000$ bootstrap replicates.

The resulting $p$-values for each language and for each set of indices
can be found in Table~\ref{tab:lang}. Taking into account the
multiple testings with a Bonferroni correction, we can conclude that
the separability assumption does not appear to hold. We can also see that the departure from separability is caught mainly on the
first component for the two Spanish varieties. In conclusion, a separable covariance structure is not a good
fit for these languages and thus, when practitioners use this
approximation for computational or modeling reasons, they should bear in mind that relevant aspects of the
covariance structure may be missed in the analysis.

\begin{table}
\caption{$P$-values 
for the test for the separability of the
covariance operators of the residual log-spectrograms of the five
Romance languages, using the Studentized version of the empirical
bootstrap. \label{tab:lang}
\vspace{.4cm}
}

\begin{tabular}{|l|c|c|c|c|c|}
\hline
$\mathcal{I}$ & French   &  Italian &  Portuguese & American Spanish & Iberian Spanish\\
    \hline
$\mathcal{I}_1$ & 0.65  & $<$ 0.001& $<$ 0.001 & $<$ 0.001 & $<$ 0.001\\
 \hline
$\mathcal{I}_2$ & 0.078  & 0.197& 0.022& 0.36& 0.013\\
\hline
$\mathcal{I}_3$ & 0.001  & 0.002& 0.001& 0.001& $<$ 0.001 \\
\hline
\end{tabular}
\end{table}

\section{Discussion and conclusions}
\label{sec:conclusion}

We presented tests to verify the separability assumption for the
covariance operators of random surfaces (or hypersurfaces) through
hypothesis testing. These tests are based on the difference between
the sample covariance operator and its separable
approximation---which we have shown to be asymptotically
Gaussian---projected onto subspaces spanned by the eigenfunctions of
the covariance of the data. While the optimal choice for this subspace is still an open problem and it may depend on the eigenstructure of the full covariance operator, it is however possible to give some advice on how to choose $\mathcal I$ in practice:

\begin{itemize}
\item in many cases, a dimensional reduction based on the separable eigenfunctions is needed also for the follow up analysis. Then, it is recommended to use the same subspace for the test procedure as well, so that we are guaranteed at least that the projection of the covariance structure in the subspace that will be used for the analysis is separable, as shown in Section \ref{sec:bootstrap}.
\item As mentioned in Section \ref{sec:bootstrap}, it is usually better to focus on the subset of eigenfunctions that it is possible to estimate accurately with the available data. These can be again identified with bootstrap methods such as the one described in
\citet{Hall2006} or considering the dimension of the sample size. As highlighted by the results of the simulation studies in Figure~\ref{fig:grid-level-power-gaussian-N25} and in Figures~\ref{fig:grid-level-power-gaussian-3x2} and \ref{fig:grid-level-power-student-3x2}, the empirical power of the test starts to decline when eigendirections that cannot be reasonably estimated with the available sample size are included.
\item When in doubt, it is also possible to apply the test to more than one subset of eigenfunctions and then
  summarize the response using a Bonferroni correction. We follow this approach in the data application
  described in Section \ref{s:data}.
\end{itemize}

Though an asymptotic distribution is
available in some cases, we also propose to approximate the
distribution of our test statistics using either a parametric
bootstrap (in case the distribution of the data is known) or an
empirical bootstrap. A simulation study suggests that the
Studentized version of the empirical bootstrap test gives the
highest power in non-Gaussian settings, and has power comparable to
its parametric bootstrap counterpart and to the asymptotic test in
the Gaussian setting. We therefore use the Studentized empirical
bootstrap for the application to linguistic data, since it is not
easy to assess the distribution of the data generating process. The
bootstrap test leads to the conclusion that the covariance structure
is indeed not separable.

Our present approach implicitly assumed that the functional
observations (e.g.\ the hypersurfaces) were densely observed. Though
this approach is not restricted to data observed on a grid, it
leaves aside the important class of functional data that are
sparsely observed \citep[e.g.][]{yao:2005}. However, the extension
of our methodology to the case of sparsely observed functional data
is also possible, as long as the estimator used for the full
covariance is consistent and satisfies a central limit theorem. Indeed, while we have only detailed the methods for 2-dimensional surfaces, the extension to higher-order multidimensional functions (such as 3-dimensional volumetric images from applications such as magnetic resonance imaging) is straightforward.


\appendix

\small

\section{The Asymptotic Covariance Structure}
  \label{app:asymptotic-covariance}

\begin{lma}
  \label{lma:general-asymptotic-covariance}

  The covariance operator of the random operator $Z$, defined in Theorem~\ref{thm:general-asymptotic-distn}, is
  characterized by the following equality, in which
  $\Gamma = \eee{(X \tensor X - C ) \bkprod (X \tensor X - C)}$:
  \begin{align}
    \lefteqn{
      \eee{ \trace\left[(A_1 \kprod A_2)Z\right] \trace\left[(B_1 \kprod B_2)Z\right]} =
  } \label{eq:charac-asymptotic-covariance}
    \\ & \trace\left[ (A \bkprod B) \Gamma \right]
     + \frac{\trace\left[ B C \right]}{\trace(C)}    \trace\left[ (A
    \bkprod \id_{H}) \Gamma \right]
 - \frac{\trace[ B_2 C_2]}{\trace[C_2]} \trace\left[ \left( A \bkprod (B_1 \kprod \id_{H_2}) \right)\Gamma  \right]
\nonumber      \\ &   - \frac{\trace[B_1 C_1]}{\trace[C_1]}
\trace\left[ \left(A \bkprod (\id_{H_1} \kprod B_2)  \right) \Gamma
\right]
 + \frac{\trace[A C]}{\trace[C]}\Bigg\{ \trace\left[ (\id_H \bkprod B)\Gamma \right]
 + \frac{\trace[BC]}{\trace[C]}\trace[\Gamma]
\nonumber  \\ & - \frac{\trace[B_2 C_2]}{\trace[C_2]}\trace\left[
(\id_H
      \bkprod (B_1 \kprod \id_{H_2}) ) \Gamma \right]
- \frac{\trace[B_1 C_1]}{\trace[C_1]}\trace\left[ (\id_H \bkprod
(\id_{H_1}
    \kprod B_2) ) \Gamma \right] \Bigg\}
\nonumber    \\ & - \frac{\trace[A_2 C_2]}{\trace[C_2]}\Bigg\{ \trace\left[ \left( (A_1 \kprod \id_{H_2}) \bkprod B
    \right)\Gamma \right]
+ \frac{\trace[BC]}{\trace[C]}\trace\left[ \left( (A_1 \kprod
\id_{H_2}) \bkprod
        \id_{H} \right)\Gamma \right]
\nonumber        \\ & - \frac{\trace[B_2 C_2]}{\trace[C_2]}
\trace\left[ \left( (A_1 \kprod \id_{H_2}) \bkprod (B_1 \kprod
    \id_{H_2}) \right)\Gamma \right]
\nonumber        \\ & - \frac{\trace[B_1 C_1]}{\trace[C_1]}
\trace\left[ \left( (A_1 \kprod \id_{H_2}) \bkprod (\id_{H_1} \kprod
B_2) \right)\Gamma \right]
    \Bigg\}
\nonumber    \\ & - \frac{\trace[A_1 C_1]}{\trace[C_1]} \Bigg\{
    \trace\left[ \left( (\id_{H_1} \kprod A_2) \bkprod B \right)\Gamma \right]
+ \frac{\trace[BC]}{\trace[C]} \trace\left[ \left( (\id_{H_1} \kprod
A_2) \bkprod \id_{H} \right)\Gamma \right] \nonumber        \\ & -
\frac{\trace[B_2 C_2]}{\trace[C_2]} \trace\left[ \left( (\id_{H_1}
\kprod A_2) \nonumber        \bkprod (B_1 \kprod \id_{H_2})
\right)\Gamma \right] \nonumber       \\ &  - \frac{\trace[B_1
C_1]}{\trace[C_1]} \trace\left[ \left( (\id_{H_1} \kprod A_2)
\bkprod (\id_{H_1} \kprod B_2) \right)\Gamma \right] \Bigg\},
\nonumber
  \end{align}
  where $A_1, B_1 \in \bounded{H_1}, A_2, B_2 \in \bounded{H_2}$, and $A = A_1 \kprod A_2$, $B = B_1 \kprod
  B_2$, $H = H_1 \tensor H_2$, and $\id_H$ denotes the identity operator on the Hilbert space $H$.

  \begin{proof}
By the linearity of the expectation and the trace, and by the properties of the partial trace, the
computation of \eqref{eq:charac-asymptotic-covariance}
boils down to the computation of expressions of the form
\[
  \eee{ \trace\left[ (A'_1 \kprod A'_2) Y \right]
    \trace\left[ (B'_1 \kprod B'_2) Y \right]},
\]
for general $A'_1, B'_1 \in \bounded{H_1}, A'_2, B'_2 \in
\bounded{H_2}$. Since $\ee \tnorm{Y}^2 < \infty$, we have
\begin{align*}
  \lefteqn{
    \ee{ ( \trace\left[ (A'_1 \kprod A'_2) Y \right] \trace\left[ (B'_1 \kprod B'_2) Y \right] )} =}
  \\ &= \trace\left[ \left( (A'_1 \kprod A'_2) \bkprod (B'_1 \kprod B'_2) \right) \ee{\left( Y \bkprod Y
      \right)} \right]
  \\ &= \trace\left[ \left( (A'_1 \kprod A'_2) \bkprod (B'_1 \kprod B'_2) \right)  \Gamma  \right],
\end{align*}
where $\Gamma = \eee{(X \tensor X - C ) \bkprod (X \tensor X - C)}$. The computation  of
\eqref{eq:charac-asymptotic-covariance} follows directly.
  \end{proof}
\end{lma}

\section{Proofs}
\label{app:proofs}

\begin{proof}[Proof of Corollary~\ref{cor:asymptotic-distribution-Tn-general-case}]
    To alleviate the notation, we shall assume without loss of generality that $\mu = \ee X = 0$.
    Using the properties of the tensor product (see Appendix~\ref{app:Hilbert-spaces}, we get that $T_{N}(r,s) = \trace\left[(\hat A_r \kprod \hat B_s) \sqrt{N} D_N
\right]$,
    where $\hat A_r= (\hat u_r \tensort \hat u_r)$, $\hat B_s =  (\hat v_s \tensort \hat v_s)$.
    Now notice that though $A_r = u_r \tensort u_r$ and $B_s = v_s \tensort v_s$ are not estimable separately (since $C_1$
    and $C_2$ are not identifiable), their $\kprod$-product is identifiable, and is consistently estimated by
    $\hat A_r \kprod \hat B_s$  (in trace norm).
    Slutsky's Lemma, Theorem~\ref{thm:general-asymptotic-distn} and the continuous mapping theorem imply
    therefore that $\left( T_N(r,s) \right)_{(r,s) \in \mathcal I}$ has the same asymptotic distribution of
    $\left( \widetilde T_N(r,s) \right)_{(r,s) \in \mathcal I}$, where
    $\widetilde T_N(r,s) = \trace\left[( A_r \kprod  B_s) \sqrt{N} D_N \right]$.
    This implies that
    \[
      \left( T_N(r,s) \right)_{(r,s) \in \mathcal I} \convd Z' = \left( \trace\left[ (A_r \kprod B_s) Z \right]
      \right)_{(r,s) \in \mathcal I},  \quad \text{as } N \rightarrow \infty,
    \]
    where $Z$ is a mean zero Gaussian random element of $\tc{H_1 \tensor H_2}$ whose covariance structure is given by
    Lemma~\ref{lma:general-asymptotic-covariance}.  $Z'$ is therefore also Gaussian random element, with mean zero and
    covariances
    \[
      \Sigma_{(r,s), (r', s')} =     \cov(Z'_{(r,s)}, Z'_{(r',s')}) =   \eee{ \trace\left[ (A_r \kprod B_s) Z
        \right] \trace\left[ A_{r'} \kprod B_{s'}) Z \right]}.
    \]
    Using Lemma~\ref{lma:general-asymptotic-covariance}, we see that the computation of $\Sigma_{(r,s), (r', s')}$ depends on
    the terms $\trace\left[ (A_r \kprod B_s) C \right] = \lambda_r \gamma_s$,
       $\trace\left[ A_r C_1 \right] = \lambda_r$,
       $\trace\left[ B_s C_2 \right] = \gamma_s$,
    as well as on the value of
    \begin{equation*}
      \trace\left[ \left( (A_1' \kprod B_1') \bkprod (A_2' \kprod B_2') \right)
        \Gamma \right]
    \end{equation*}
    for general $A_1', A_2' \in \bounded{H_1}, B_1', B_2' \in \bounded{H_2}$. Using the
    Karhunen--Lo\`eve expansion $X = \sum_{i,i' \geq 1} \xi_{ii'} u_i \tensor v_{i'}$, where $\xi_{ii'} =
    \sc{X, u_i \tensor v_{i'}}$,
    we get
    \begin{align*}
      \Gamma &= \ee\left( \left( X \tensort X - C \right) \kprod \left( X \tensort X - C \right) \right)
      \\ &= \sum_{i,i',j,j',k,k',l,l' \geq 1} \beta_{ii'jj'kk'll'} \left( u_{ij} \kprod v_{i'j'} \right)
      \bkprod \left( u_{kl} \kprod v_{k'l'} \right)
      \\ & \qquad  - \sum_{i,i',j,j'} \alpha_{ii'}\alpha_{jj'} \left( u_{ii} \kprod v_{i'i'} \right) \bkprod \left(
        u_{jj} \kprod v_{j'j'} \right)
    \end{align*}
    where we have written $u_{ij} = u_i \tensort u_j \in \tc{H_1}, v_{ij} = v_i \tensort v_j \in \tc{H_2}$,
    $\beta_{ii'jj'kk'll'} =
    \eee{\xi_{ii'}\xi_{jj'}\xi_{kk'}\xi_{ll'}}$, $\alpha_{ij} = \lambda_i \gamma_j$
    and used the identity  $ u_{ij}\kprod v_{i'j'} = (u_i \tensor v_{i'})\tensort (u_{j} \tensor
    v_{j'})$. Therefore,
    \begin{align*}
      \lefteqn{\trace\left[ \left( (A_1' \kprod A_2') \bkprod (B_1' \kprod B_2') \right)
          \Gamma \right] = }
      \\ &= \sum_{i,i',j,j',k,k',l,l' \geq 1} \beta_{ii'jj'kk'll'} \trace[A_1' u_{ij}] \trace[A_2'
      v_{i'j'}] \trace[B_1' u_{kl}] \trace[B_2' v_{k'l'}]
      \\ & \qquad - \sum_{i,i',j,j'} \alpha_{ii'}\alpha_{jj'} \trace[A_1' u_{ii}]\trace[B_1'
      u_{jj}]\trace[A_2' v_{i'i'}]\trace[B_2' v_{j'j'}],
    \end{align*}
    and the computation of the variance $\Sigma_{(r,s), (r', s')}$ follows from a straightforward (though
    tedious) calculation.
  \end{proof}

  \begin{proof}[Proof of Corollary~\ref{cor:asymptotic-distn-Tn-gaussian-case}]
    We only need to compute and substitute the values of the fourth order moments terms $\tb_{ijkl}$ in the
    expression given by Corollary~\ref{cor:asymptotic-distribution-Tn-general-case}. Since $\tb_{ijkl} =
    \eee{\xi_{ij}^2 \xi_{kl}^2} = 3 \alpha_{kl}^2$ if $(i,j) = (k,l)$, and $\tb_{ijkl} = \alpha_{ij}
    \alpha_{kl}$ if $(i,j) \neq (k,l)$, straightforward calculations give
    \begin{align*}
      \tb_{rs \cdot \cdot } &= 2 \alpha_{rs}^2 + \alpha_{rs}\trace(C) = \tb_{r\cdot \cdot s},
      &&&      \tb_{\cdot \cdot \cdot \cdot } &= \trace(C)^2 + 2 \hsnorm{C}^2,
      \\       \tb_{\cdot s\cdot s'} &= \gamma_s\gamma_{s'}(\trace(C_1)^2 + 2 \delta_{ss'} \hsnorm{C_1}^2),
      &&&     \tb_{r \cdot r' \cdot } &= \lambda_r\lambda_{r'}(\trace(C_2)^2 + 2 \delta_{rr'} \hsnorm{C_2}^2),
      \\       \tb_{r s  \cdot s'} &= 2 \delta_{ss'}\alpha_{rs}^2 + \alpha_{rs}\gamma_{s'} \trace(C_1),
      &&&      \tb_{r s r' \cdot} &= 2 \delta_{rr'}\alpha_{rs}^2 + \alpha_{rs}\lambda_{r'} \trace(C_2),
      \\      \tb_{\cdot\cdot\cdot s} &= 2 \gamma_s^2 \hsnorm{C_1}^2 + \gamma_s \trace(C_1)^2 \trace(C_2),
      &&&      \tb_{r \cdot\cdot\cdot } &= 2 \lambda_r^2 \hsnorm{C_2}^2 + \lambda_r \trace(C_1) \trace(C_2)^2,
    \end{align*}
  where $\delta_{ij} = 1$ if $i=j$, and zero otherwise.
    The  proof is finished by direct calculations.
  \end{proof}

\section{Partial Traces}
\label{app:partial-traces}

Letting $\tc{H_1 \tensor H_2}$ denote the space of trace-class
operators on $H_1 \tensor H_2$, we define the partial trace with
respect to $H_1$ as the unique linear operator $\trace_{1} :
\tc{H_1 \tensor H_2} \rightarrow \tc{H_2}$ satisfying $\trace_{1}(A
\kprod B)  = \trace(A)B$ for all $A \in \tc{H_1}$, $B \in \tc{H_2}$.
\begin{prop}
  \label{prop:partial-trace}
  The operator $\trace_1$ is well-defined, linear, continuous, and satisfies
  \begin{equation}
    \label{eq:continuity-of-partial-trace}
    \tnorm{\trace_1(A) } \leq \tnorm{A}, \quad A \in \tc{H_1 \tensor H_2}
  \end{equation}
  Furthermore,
  \begin{equation}
    \trace(S\trace_1(T)) = \trace((\id_1 \kprod S)T), \qquad T \in \tc{H_1 \tensor H_2}, S
    \in \bounded{H_2},
    \label{eq:partial-trace-characterization}
  \end{equation}
  where $\id_1$ is the identity operator on $H_1$.
\begin{proof}
  Let us start by proving that the operator $\trace_1(\cdot)$ is well defined. By
  Lemma~\ref{lma:approx-of-trace-class-operator-on-tensor-space}, the space
  \[
    B_0 = \{\sum_{i =1}^n A_i
      \kprod B_i : A_i \in \tc{H_1}, B_i \in \tc{H_2}, n =1,2,\ldots \}
  \]
  is  a dense subset of $\tc{H_1
    \tensor H_2}$. We therefore only need to show that $\trace_1(\cdot)$ is continuous on $B_0$.
  Let $T \in B_0$, $T = \sum_{i = 1}^n A_i \kprod B_i$. Then, for any $S \in \bounded{H_2}$,  we have
  \begin{align*}
    \trace(S \trace_1(T)) &= \sum_{i = 1}^n \trace(A_i)\trace(S B_i)= \sum_{i = 1}^n \trace\left( (\id_1 \kprod S)( A_i \kprod B_i)
    \right)=
    \\ &=  \trace\left((\id_1 \kprod S)\left[ \sum_{i = 1}^n A_i \kprod B_i \right] \right)=  \trace\left((\id_1 \kprod S) T \right).
  \end{align*}
  Hence, using the following formula for the trace norm,
  \[
    \tnorm{T} = \sup \left\{ |\trace(ST)| : \opnorm{S}=1 \right\},
  \]
  we get  $\tnorm{\trace_1(T)} \leq \tnorm{T}$ for all $T \in
  B_0$. Thus
  $\trace_1(\cdot)$  can be extended by continuity to $H_1 \tensor H_2$, and
  \eqref{eq:continuity-of-partial-trace} (of the paper) holds.

  Let us now show \eqref{eq:partial-trace-characterization} (of the paper). Fix $S \in \bounded{H_2}$, and define the
  linear functionals $g_S, h_S : \tc{H_1 \tensor H_2} \rightarrow \bR$ by $g_S(T) = \trace\left(  (S \kprod
    \id_2) T \right)$ and $h_S(T) = \trace(S \trace_1(T))$, for $T \in  \tc{H_1 \tensor H_2}$.
  By H\"older's inequality and \eqref{eq:continuity-of-partial-trace} (of the paper), $g_S$ and $h_S$ are both continuous. Since
  they are equal on the dense subset $B_0$, they are in fact equal everywhere, and
  \eqref{eq:partial-trace-characterization} (of the paper) follows.
\end{proof}
\end{prop}

We can also define $\trace_2: \tc{H_1 \tensor H_2} \rightarrow \tc{H_1}$ analogously.
The following result gives an explicit formula for the partial traces of integral operators with continuous
kernels.

\begin{prop}
  \label{prop:partial-trace-of-kernel-operator-with-continous-kernels}
Let $D_s \subset \bR^p, D_t \subset \bR^q$ be compact subsets,
  $H_1 = L^2\left( D_s, \bR \right), H_2 = L^2\left( D_t, \bR \right)$, and $H = L^2(D_s \times D_t, \bR) =
  H_1 \tensor H_2$.
If $C \in \tc{L^2(D_s \times D_t, \bR)}$  is a positive definite operator with symmetric continuous kernel
$c=c(s,t,s',t')$, i.e.\ $c(s,t,s',t')=c(s',t',s,t)$ for all $s,s' \in D_s, t,t' \in D_t$, and
  \[
    Cf(s,t) = \iint_{D_s \times D_t} c(s,t,s',t') f(s',t') ds' dt', \qquad f \in L^2(D_s \times D_t,
    \bR),
  \]
  then $\trace_1(C)$ is the integral operator on $L^2(D_t, \bR)$ with kernel $k(t,t') = \int_{D_s}
  c(s,t,s,t') ds$. The analogous result also holds for $\trace_2(C)$.
\begin{proof}
  Let
  $\vep > 0$. By Lemma~\ref{lma:approx-of-trace-class-operator-with-continuous-kernel-on-tensor-space}, we
  know that
  there exists an integral operator $C'$ with continuous kernel $c'$ such that
  $\tnorm{C - C'} \leq \vep/2$ and $\hnorm{c - c'}_\infty \leq \vep/2,$
  where $C' = \sum_{n = 1}^N A_n \kprod B_n$, and each $A_n, B_n$ are finite rank operators, with continuous
  kernels $a_n$, respectively $b_n$, and $\hnorm{g}_\infty~=~\sup_{x} |g(x)|$.
  We have
  \begin{align}
    \snorm{\trace_1(C) - \int_{D_s} c(s,\cdot,s,\cdot) ds }_2 & \leq \snorm{\trace_1(C) - \trace_1(C')}_2
    \label{eq:partial-trace-integral-formula-proof}
    \\ & \quad +
    \snorm{\trace_1(C') - \int_{D_s} c'(s,\cdot,s,\cdot) ds }_2
    \nonumber
    \\ & \quad + \snorm{\int_{D_s} c'(s,\cdot,s,\cdot) ds - \int_{D_s} c(s,\cdot,s,\cdot) ds }_2.
    \nonumber
  \end{align}
  The first term is bounded $\tnorm{\trace_1(C) - \trace_1(C')} \leq \tnorm{C - C'} \leq \vep/2$. The second
  term is equal to zero since
  \begin{align*}
    \trace_1(\sum_{n = 1}^N A_n \kprod B_n) &= \sum_{n = 1}^N \trace(A_n)B_n
    \\ & = \sum_{n = 1}^N \int_{D_s} A_n(s,s) ds B_n
    \\ & = \int_{D_s} \left( \sum_{n = 1}^N A_n \kprod B_n \right)(s,\cdot,s,\cdot) ds.
    \\ & = \int_{D_s} c'(s,\cdot,s,\cdot) ds,
  \end{align*}
  where the second equality comes from the fact that $A_n$ is a finite rank operator (hence trace-class) with
  continuous kernel. The third term of \eqref{eq:partial-trace-integral-formula-proof}  is
  \begin{multline*}
    \left( \iint_{D_t \times D_t} \left( \int_{D_s} \left[ c'(s,t,s,t') - c(s,t,s,t') \right] ds \right)^2 dt dt'
    \right)^{1/2} \\
    \leq |D_s| \: |D_t| \: \hnorm{c' - c}_\infty \leq \vep/2,
  \end{multline*}
  where $|D_s| = \int_{D_s} dx$ and $|D_t| = \int_{D_t} dy$.
  Therefore,
  \[
    \snorm{\trace_1(C) - \int_{D_s} c(s,t,s,t') ds}_2 \leq \vep.
  \]
  Since this holds for any $\vep > 0$,
  $\trace_1(C)$ is equal to the operator with kernel $k(t,t') =   \int_{D_s} c(s,t,s,t') ds$.
  The proof of the analogous result for $\trace_2(C)$ is similar.
\end{proof}
\end{prop}
The next result states that the partial trace of a Gaussian random trace-class operator is also Gaussian.
\begin{prop}
  \label{prop:partial-trace-of-Gaussian}
  Let $Z \in \tc{H_1 \tensor H_2}$ be a Gaussian random element. Then $\trace_1(Z) \in \tc{H_2}$ is a Gaussian random element.
  \begin{proof}
    The proof is finished by noticing that  $A \in \bounded{H_2}$, we have
      $\trace(A \trace_1(Z)) = \trace\left( (\id \kprod A) Z \right)$,
    where the right-hand side is obviously Gaussian.
  \end{proof}
\end{prop}

\section{Background Results}
\label{app:background-results-main-paper}

This section presents some background results that are used in the
paper. Some references for these results are
\citet{Zhu:2007,Gohberg:1971,gohberg:1990,Kadison:1997vol1,Kadison:1997vol2,Ringrose:1971}.

\subsection{Tensor Products Hilbert Spaces, and Hilbert--Schmidt Operators}
\label{app:Hilbert-spaces}

Let $H_1,H_2$ be two real separable Hilbert spaces, whose inner
products are denoted by $\sc{\cdot,\cdot}_1$ and
$\sc{\cdot,\cdot}_2$, respectively. Let  $H_1 \tensor H_2$ denote
the Hilbert space  obtained as the completion of the space of finite
linear combinations of simple tensors $u \tensor v, u \in H_1, v \in
H_2$ under the inner product
\[ \sc{u \tensor v, u' \tensor v'} = \sc{u, u'}_1\sc{v,v'}_2. \]

The Hilbert space $H_1 \tensor H_2$ is actually isometrically isomorphic to the space of \emph{Hilbert--Schmidt
  operators} from $H_2$ to $H_1$, denoted by $\schatten_2(H_2, H_1)$, which consists of all continuous linear
operators $T : H_2 \rightarrow H_1$ satisfying
\[
  \snorm{T}_2^2 = \sum_{n \geq 1} \hnorm{T e_n}^2,
\]
where the sum extends over any orthonormal basis $(e_n)_{n \geq 1}$ of $H_2$. The norm $\snorm{\cdot}$ is
actually induced by the inner-product
 inner product
 \[
 \HSsc{T, S} = \sum_{n \geq 1} \sc{T e_n, S e_n}_1, \quad T,S \in \HS{H_2,H_1},
 \]
 which is independent of the choice of the basis (the space $\HS{H_2, H_1}$ is therefore itself a Hilbert
 space).
The isomorphism between $H_1 \tensor H_2$ and $\HS{H_2,H_1}$ is given by the mapping
$\Phi: H_1 \tensor H_2 \rightarrow \schatten_2(H_2, H_1)$, defined by
$\Phi(u \tensor v) = u \tensort v$ for all $u \in H_1, v
\in H_2$, where $u \tensort v (v') = \sc{v', v}u$ for $u \in H_1, v,v' \in H_2$.
We therefore identify these two spaces, and might write $u
\tensor v$ instead of $u \tensort v$ hereafter.

Notice that since $\mathcal H = \HS{H_1} = \HS{H_1 , H_1}$ is itself a Hilbert space, if $A, B \in
\mathcal H$, the operator $A \btensort B \in \HS{\mathcal H, \mathcal H}$ is defined by $\left( A \btensort B
\right)(C) = \HSsc{C,B}A$, for $A,B,C \in \mathcal H$.
Here are some properties of the tensor product $\cdot\tensort \cdot$:
\begin{prop} \label{prop:properties-tensor}
  Let $H$ be a real separable Hilbert space.
For any $u,v,f,g \in H$, $A,B \in \schatten_{2}(H)$
    \begin{enumerate}
        \item $ \cdot \tensort \cdot $ is linear on the left, and conjugate-linear on the right,
        \item $\HSsc{u\tensort v, f\tensort g } = \sc{u,f}\sc{g,v} = \sc{(u \tensort v)g, f}$,
        \item $\HSsc{A, u\tensort v} = \sc{Av, u} = \HSsc{v \tensort u, A^{\dag}},$
                \item $\trace(u \tensort v)= \sc{u,v}$,
                \item $\tnorm{u \tensort v} = \snorm{u \tensort v}_{2} = \hnorm{u}\hnorm{v}$,
        \item $(u \tensort v)(f \tensort g) = \sc{f, v}u \tensort g$,
        \item $(u \tensort v)^{\dag} = v \tensort u,$
                \item  $(A \btensort B)^{\dag} = B \btensort A$.
    \end{enumerate}
        \begin{proof}
          The proof follows from the definition and the properties of the inner product, and is therefore
          omitted.
        \end{proof}
\end{prop}

Recall that for $A \in \bounded{H_1}, B \in \bounded{H_2}$,  the operator $(A
\kprod B) \in \bounded{ H_1 \tensor H_2 }$  is defined by the linear
extension of
\[
  (A \kprod B)(u \tensor v) = Au \tensor Bv, u \in H_1, v \in H_2
\]
Furthermore, we have $(A \kprod B)^{\dag} = A^{\dag}
\kprod B^{\dag}$, $\opnorm{A \kprod B} = \opnorm{A}\opnorm{B}$, and
$(A \kprod B)(C \kprod D) = (AC \kprod CD)$ for $A,C \in
\bounded{H_1}$, $B,D \in \bounded{H_2}$. For $u,v \in H_1, f,g \in
H_2$, $(u\tensort v)\kprod (f\tensort g) = (u \tensor f) \btensort
(v \tensor g)$ If $A \in \tc{H_1}, B \in \tc{H_2}$, then $A \kprod B
\in \tc{H_1 \tensor H_2}$, 
\[
  \tnorm{A \kprod B} \leq \tnorm{A}\tnorm{B},
\]
and $\trace(A \kprod B) = \trace(A) \trace(B)$.

In the case $H_1 = L^2\left( [-S, S]^d, \bR \right)$, $H_2 = L^2\left( [0,T], \bR \right)$, with $S, T > 0$,
if $A \in \HS{H_1}$, $B \in \HS{H_2}$ are
Hilbert--Schmidt operators (hence also integral operators, with
kernels $a(s,s), b(t,t)$, respectively), the operator $A \kprod B
\in \HS{H_1 \tensor H_2} = \HS{L^2\left( [-S, S]^d \times [0,T], \bR \right)}$ is also an integral operator with
kernel $k(s,t, s',t') = a(s,s') {b(t,t')},$ that is, 
\[
  (A \kprod B) u (s,t) = \int_{0}^T \int_{[-S, S]^d} k(s,t,s',t')u(s',t')ds dt.
\]

\subsection{Random Elements in Banach Spaces}

We understand random elements of  a separable Banach space  $(B,
\hnorm{\cdot})$  in the Bochner sense \citep[e.g.][]{Ryan:2002}. A
random element $X \in B$ satisfying $\ee\hnorm{X} < \infty$ has a
mean $\ee X \in B$, which satisfies $S (\ee X) = \ee(SX)$ for all
bounded linear operator $S : B \rightarrow B'$, where $B'$ is
another Banach space.

\subsection{Random Trace-class Operators}

If $X \in \tc{H}$ is a random element satisfying $\ee \tnorm{X} <
\infty$, i.e.\ a \emph{random trace-class operator}, then $\ee{\trace(A X)} = \trace(A \ee X)$ for any $A \in
\bounded{H}$. Furthermore, if $X' \in \tc{H}$ is another random
element such that $\ee(\tnorm{X}\tnorm{X'}) < \infty$, then
\begin{align*}
  \ee \left( \trace\left[ A X \right] \trace[A' X']\right) &= \ee \left( \trace\left[ AX \bkprod A' X' \right] \right)
  \\ &= \ee \left(  \trace\left[ (A \bkprod A') (X \bkprod X') \right] \right)
  \\ &=   \trace\left[ (A \bkprod A') \ee\left(  X \bkprod X'  \right)\right],
\end{align*}
for any $A, A' \in \bounded{H}$.

The second-order structure of a random element $X \in \tc{H}$
satisfying $\ee \tnorm{X}^2 < \infty$ is encoded by the covariance
functional $\Gamma: \bounded{H} \times \bounded{H} \rightarrow \bR$,
which is defined by \[
  \Gamma(A,B) = \cov \left( \trace[AX], \trace[BY] \right).
\]
Since
\begin{align*}
\Gamma(A,B) &= \ee\left( \trace[A(X-\mu)]\trace[B(X-\mu)] \right)
  \\ &= \trace\left[ (A \bkprod B) \ee \left( (X-\mu) \bkprod (X - \mu) \right) \right],
\end{align*}
the second-order structure is also encoded by the \emph{ generalized
covariance operator } $$\Gamma = \ee\left( (X - \mu) \bkprod (X -
\mu) \right) \in \tc{H \tensor H}.$$

\begin{prop}
  \label{prop:kronecker-product-of-gaussian-and-fixed-operator}
  Let $H_1, H_2$ be real separable Hilbert spaces. Let $Y \in \tc{H_1}$ be a Gaussian
  random element such that $\ee \tnorm{Y}^2 < \infty$.  Then, for any $T \in \tc{H_2}$
  fixed, $Y \kprod T$ is a Gaussian random element of $\tc{H_1 \tensor H_2}$.
  \begin{proof}
    We need to show that for all $S \in \bounded{H_1 \tensor H_2}$,
    $\trace(S (Y \kprod T))$ is Gaussian. This can be reduced to showing that
    $\trace(S_n(Y \kprod T))$ is Gaussian for all $n \geq 1$, where $(S_n)$ is a sequence
    of operators in $\bounded{H_1 \tensor H_2}$ that converges weakly to $S$. Indeed,
    letting $D_n = S_n - S$,  we
    have
    \begin{align*}
      \eee{ \left( \trace(S_n (Y \kprod T)) - \trace(S (Y \kprod T) \right)^2 } &= \eee{
    \trace(D_n (Y \kprod T))^2}
    \\ &= \ee \trace\left( (D_n \kprod D_n)  (Y \kprod T \kprod Y \kprod T) \right)
    \\ &= \trace\left( (D_n \kprod D_n)  \ee (Y \kprod T \kprod Y \kprod T) \right),
    \end{align*}
    where the last equality is valid since 
    \[
      \ee \tnorm{(D_n \kprod D_n) (Y \kprod T \kprod Y \kprod T)  }  \leq  \opnorm{D_n}^2 \tnorm{T}^2 \ee
      \tnorm{Y}^2 < \infty.
    \]
    Lemma~\ref{lma:weak-convergence-and-kronecker-product} tells us that $D_n \kprod D_n$
    converges weakly to zero, and since 
    \[
      \tnorm{\ee (Y \kprod T \kprod Y \kprod T)
      } \leq \tnorm{T}^2 \ee \tnorm{Y}^2 < \infty,
    \]
  Lemma~\ref{lma:weak-convergence-and-trace} tells us that $\eee{ \left( \trace(S_n (Y
  \kprod T)) - \trace(S (Y \kprod T) \right)^2 } \raz$ as $n \rainf$. Therefore, since the
  space of Gaussian random variables is close under the $L^2(\Omega, \mathbb P)$ norm,
  $\trace(S(Y \kprod T))$ is Gaussian if $\trace(S_n(Y \kprod T))$ is
  Gaussian for all $n \geq 1$.
  Lemma~\ref{lma:weak-approximation-of-bounded-operator-by-sum-of-kprod} tells us that we
  can choose $S_n = \sum_{i = 1}^n A_i \kprod B_i$, where $A_i \in \bounded{H_1}$ and $B_i
  \in \bounded{H_2}$. In this case,
  \[
    \trace(S_n(Y \kprod T)) = \sum_{i = 1}^n \trace\left( (A_i \kprod B_i) (Y \kprod T)
    \right) = \sum_{i = 1}^n \trace(A_i Y) \trace(B_i T),
  \]
 which is Gaussian since $\trace(A_i Y)$ is Gaussian for each $i \geq 1$.
  \end{proof}
\end{prop}

\subsection{Technical results}
\label{sec:technical}

Recall that $(T_n)_{n \geq 1} \subset \bounded{H}$ is said to
converge weakly to $T \in \bounded{H}$ if for all $u,v \in H$,
$\sc{T_n u, v} \rightarrow \sc{Tu,v}$ as $ n \rainf$.

\begin{lma}
  \label{lma:weak-approximation-of-bounded-operator-by-sum-of-kprod}
  Let $H_1, H_2$ be real separable Hilbert spaces.
  For any $S \in \bounded{H_1 \tensor H_2}$, there exists a sequences of operators
  $(S_n)_{n \geq 1} \subset \bounded{H_1 \tensor H_2}$ of the form $S_n = \sum_{i = 1}^n A_i \kprod
  B_i$ with $A_i \in \bounded{H_1}$, $B_i \in \bounded{H_2}$, such that $S_n$ converges
  weakly to $S$.
  \begin{proof}
  For $S \in \bounded{H_1 \tensor H_2}$, define
  \begin{align*}
    S_N &= \sum_{n,n',m,m'=1}^N \sc{S e_n \tensor f_m, e_{n'} \tensor f_{m'}} (e_{n'} \tensor f_{m'}) \tensort
    (e_n \tensor f_m),
    \\ &= \sum_{n,n',m,m'=1}^N \sc{S e_n \tensor f_m, e_{n'} \tensor f_{m'}} (e_{n'}
    \tensor e_n) \kprod (f_{m'} \tensor f_m),
  \end{align*}
  where $(e_n)_{n \geq 1}$ is an orthonormal basis of $H_1$, and $(f_m)_{m \geq 1}$ is an orthonormal basis of
  $H_2$.
  First,  notice that we have the following equality:
  \begin{align}
    \sc{S_N e_i \tensor f_j, e_k \tensor f_l} &= \sum_{n,m=1}^N \sum_{n',m'=1}^N \sc{S e_n \tensor f_m, e_{n'} \tensor
  f_{m'}} \sc{e_n \tensor f_m, e_i \tensor f_j} 
  \sc{e_{n'} \tensor f_{m'}, e_k \tensor f_l}
  \nonumber
  \\ &= \sc{S e_i \tensor f_j, e_k \tensor f_l} \ind{i \leq N} \ind{j \leq N}\ind{k \leq
N}\ind{l \leq N}. \label{eq:weak-conv-on-simple-tensor}
  \end{align}
  Therefore, for general $g,h \in H_1 \tensor H_2$,  $g = \sum_{i,j \geq 1} \alpha_{ij} e_i \tensor f_j$ and
  \[
    h = \sum_{k,l \geq 1} \beta_{kl} e_k \tensor f_l,
  \]
  we have
  \begin{align*}
    \sc{S_N h, g} &= \sum_{i,j \geq 1}\sum_{k,l \geq 1}\alpha_{ij} \beta_{kl} \sc{S_N(e_i \tensor f_j), e_k \tensor f_l}
    \\ &= \sum_{i,j \geq 1}\sum_{k,l \geq 1} \alpha_{ij} \beta_{kl} \sc{S(e_i \tensor f_j), e_k \tensor f_l} \ind{i \leq N} \ind{j \leq N}\ind{k \leq
N}\ind{l \leq N}
    \tag{Using \eqref{eq:weak-conv-on-simple-tensor}}
    \\ &=  \sc{S \left( \sum_{i,j = 1}^N  \alpha_{ij} e_i \tensor f_j \right), \sum_{k,l =
1}^N \beta_{kl} e_k \tensor f_l}.
  \end{align*}
  Therefore, by continuity of the inner product and the continuity of $S$, we have
  $\lim_{N \rainf} \sc{S_N h, g} = \sc{S h, g}$.
  \end{proof}
\end{lma}

\begin{lma}
  \label{lma:weak-convergence-and-kronecker-product}
  Let $(S_n)_{n \geq 1} \subset \bounded{H}$ be a sequence of operators converging
  weakly to $S \in \bounded{H}$. Then, $S_n \kprod S_n$ converges weakly to $S \kprod S
  \in \bounded{H \tensor H}$.
  \begin{proof}
    For $u,v,z,w \in H$, we have
    \begin{align*}
      \lim_{n \rainf } \sc{\left( S_n \kprod S_n \right) (u \tensor v), z \tensor w } & =
      \lim_{n \rainf} \sc{S_n u,
      z}\sc{S_n v, w}
      \\ & = \sc{S u, z}\sc{S v, w }
      \\ & = \sc{\left( S \kprod S
      \right)(u \tensor v), z \tensor w}.
    \end{align*}
    Now for general $g,h \in H \tensor H$, let us write  $g = \sum_{i,j \geq 1} \alpha_{ij}
    e_i \tensor e_j$, $h = \sum_{k,l \geq 1}
    \beta_{kl} e_k \tensor e_l$,
    where $(e_i)_{i \geq 1}$ is an orthonormal basis of $H$, and let
    $g^{N} = \sum_{i,j = 1}^{N} \alpha_{ij} e_i \tensor e_j$, $h^N = \sum_{k,l=1}^N
    \beta_{kl} e_k \tensor e_l$ for $N \geq 1$. Also, let 
    \[
      K = \max\{ \sup_{n \geq 1} \opnorm{S_n}, \opnorm{S} \}, 
    \]
    and notice that $K < \infty$ by the uniform boundedness principle \citep[e.g.][]{rudin:1991}.
    We have
    \begin{align*}
      \left| \sc{(S_n \kprod S_n) g, h}  - \sc{(S \kprod S)g, h} \right| & \leq
      | \sc{(S_n \kprod S_n) g, h}| + |\sc{(S_n \kprod S_n) g^N, h^N}|
      \\ & \quad + \left| \sc{(S_n \kprod S_n) g^N, h^N}  - \sc{(S \kprod S)g^N, h^N}
      \right|
      \\ & \quad + \left| \sc{(S \kprod S) g^N, h^N}  - \sc{(S \kprod S)g, h}
      \right|
      \\ &  \leq K^2 \hnorm{g} \hnorm{h - h^N} + K^2 \hnorm{g - g^N} \hnorm{h^N}
      \\ & \quad + \Big| \sum_{i,j,k,l = 1}^N \alpha_{ij} \beta_{kl} \big[ \sc{(S_n -
  S)e_i, e_k}\sc{S_n e_j, e_l}
\\ & \quad\quad \quad\quad\quad\quad\quad\quad \quad +   \sc{Se_i,e_k}\sc{(S_n - S)e_j, e_l} \big]
      \Big|
      \\ & \quad + K^2 \hnorm{g} \hnorm{h - h^N} + K^2 \hnorm{g^N - g} \hnorm{h}
      \\ &  \leq 4 K^2 \max\left\{ \hnorm{g} , \hnorm{h}
    \right\} \max\left\{ \hnorm{h - h^N} , \hnorm{g - g^N} \right\}
      \\ &  \quad + 2 K N^4 \left( \hnorm{g} + \hnorm{h} \right) \max_{1 \leq i,k \leq N}
    \left\{ \left|
          \sc{(S_n - S) e_i, e_k} \right| \right\}.
    \end{align*}
    Now, for any $\vep > 0$, choose $N > 1$ such that $\max\left\{ \hnorm{h - h^N},
    \hnorm{g - g^N} \right\} \leq \vep\left( 6 K^2 \max\left\{ \hnorm{g} , \hnorm{h}
    \right\}\right)^{-1}$. Since $N$ is fixed, we can find an $n' \geq 0$ such that
    \[
      \max_{1 \leq i,k \leq N} \left| \sc{(S_n - S) e_i, e_k} \right| \leq
      \frac{\vep}{6 K N^4 ( \hnorm{g} + \hnorm{h} )}, \quad \text{for all } n \geq n'.
    \]
    Then, for all $n \geq n'$, we have
      $| \ssc{(S_n \kprod S_n) g, h}  - \ssc{(S \kprod S)g, h} |  \leq \vep$, therefore
      $S_n \kprod S_n$ converges weakly to $S \kprod S$.
  \end{proof}
\end{lma}

\begin{lma}
  \label{lma:weak-convergence-and-trace}
  Let $(S_n)_{n \geq 1} \subset \bounded{H}$ be a sequence of operators converging
  weakly to $S \in \bounded{H}$. Then, for all $T \in \tc{H}$, we have
  \[
    \trace(S_n T) \rightarrow \trace(ST), \quad n \rainf.
  \]
  \begin{proof}
    Let $T = \sum_{l \geq 1 } \lambda_l u_l \tensort v_l$ be the singular value
    decomposition of $T$. Without loss of generality, $\left( v_l \right)_{l \geq 1}$ is an
    orthonormal basis of $H$. We have
    \begin{align*}
      \lim_{n \rainf} \trace(S_n T) & = \lim_{n \rainf} \sum_{l \geq 1} \sc{v_l, S_n T v_l}
      \\ &= \lim_{n \rainf} \sum_{l \geq 1} \lambda_l \sc{v_l, S_n u_l}
      \\ &= \sum_{l \geq 1} \lambda_l \lim_{n \rainf} \sc{v_l, S_n u_l}
      \\ &= \sum_{l \geq 1} \lambda_l \sc{v_l, S u_l}
      \\ &= \trace(S T),
    \end{align*}
    where the third equality is justified by the dominated convergence theorem since
    $\sum_{l \geq 1} |\lambda_l| |\sc{v_l, S_n u_l}| \leq \sup_{n \geq 1} \left\{
      \opnorm{S_n}
    \right\} \tnorm{T} < \infty$ by the uniform boundedness theorem
  \end{proof}
\end{lma}

\begin{lma}
  \label{lma:approx-of-trace-class-operator-on-tensor-space}
  The operators of the form $\sum_{n=1}^N A_n \kprod B_n$, where  $A_n: H_1 \rightarrow H_1$ and
  $B_n:H_2 \rightarrow H_2$ are finite rank operators, and $N < \infty$, are dense in the Banach space $\tc{H_1 \tensor H_2}$.
  \begin{proof}
    Let $T \in \tc{H_1 \tensor H_2}$. Then $T = \sum_{n \geq 1} \lambda_n U_n \tensort V_n$, with convergence
    in trace norm, where
    $(\lambda_n)_{n \geq 1}$ is a summable decreasing sequence of positive numbers, and
    $(U_n)_{n \geq 1} \subset H_1 \tensor H_2$ and $(V_n)_{n \geq 1} \subset H_1 \tensor
    H_2$ are orthonormal sequences. Each $U_n$ can be written as $U_n = \sum_{l \geq 1}
    \lambda_{n,l} u_{n,l}^{(1)} \tensor u_{n,l}^{(2)}$, with convergence in the norm of $H_1 \tensor
    H_2$, that we denote by $\hnorm{\cdot}_2$. Similarly, $V_n = \sum_{l \geq 1} \gamma_{n,l}
    v_{n,l}^{(1)} \tensor v_{n,l}^{(2)}$. Let 
    \[
      U_n^M = \sum_{l = 1}^M \lambda_{n,l} u_{n,l}^{(1)} \tensor u_{n,l}^{(2)} \quad \text{ and } \quad
      V_n^M = \sum_{l = 1} \gamma_{n,l} v_{n,l}^{(1)} \tensor v_{n,l}^{(2)}.
    \]
    Fix $\vep > 0$, and choose $N$ such that $\tnorm{T - \sum_{n = 1}^{N}\lambda_n U_n \tensort V_n} \leq
    \vep/2$. For $M \geq 1$ fixed, we have
    \begin{align*}
      \tnorm{T - \sum_{n = 1}^{N}\lambda_n U_n^M \tensort V_n^M} &\leq \tnorm{T - \sum_{n = 1}^{N}\lambda_n
    U_n \tensort V_n}
    \\ & \hspace{1cm}+ \tnorm{\sum_{n = 1}^{N} \lambda_n \left[ (U_n - U_n^M) \tensort V_n + U_n^M \tensort (V_n -
    V_n^M)   \right]}
 \\   &\leq \vep/2 + \sum_{n = 1}^N \lambda_n \left[ \tnorm{(U_n - U_n^M) \tensort V_n } + \tnorm{U_n^M \tensort
(V_n - V_n^M)} \right]
 \\   &\leq \vep/2 + \sum_{n = 1}^N \lambda_n \left(\hnorm{U_n - U_n^M}_2 \hnorm{ V_n }_2 + \hnorm{U_n^M}_2
 \hnorm{V_n - V_n^M}_2 \right).
    \end{align*}
    Take $M \geq 1$ such that
    \[
      \max_{n=1,\ldots, N} \left\{ \hnorm{U_n - U_n^M}_2, \hnorm{V_n - V_n^M}_2
      \right\} \leq \min\left\{\frac{\vep}{6 \trace{C}}, 1\right\}.
    \]
    Then,
    since $U_n, V_n$ have unit length, and  $\hnorm{U_n^M}_2 \leq 1 + \hnorm{U_n - U_n^M} \leq 2$ for
    $n=1,\ldots, N$, we have
    \begin{align*}
      \tnorm{T - \sum_{n = 1}^{N}\lambda_n U_n^M \tensort V_n^M} &\leq \vep/2 +  3 \max_{n=1,\ldots, N} \left\{
    \hnorm{U_n - U_n^M}_2, \hnorm{V_n - V_n^M}_2 \right\} \cdot \left( \sum_{n = 1}^N \lambda_n \right)
    \\ &\leq \vep/2 + 3\frac{\vep}{6 \trace{C}} \trace{C}
    \\ &=\vep.
    \end{align*}
Since $\sum_{n = 1}^{N}\lambda_n U_n^M \tensort V_n^M = \sum_{n =
1}^{N} \sum_{j,l = 1}^{K} \left( \lambda_n \lambda_{n,l}
\gamma_{n,j} u_{n,l}^{(1)}  \tensor v_{n,j}^{(1)} \right) \kprod
\left( u_{n,l}^{(2)} \tensor v_{n,j}^{(2)} \right)$ the proof is
finished.
  \end{proof}
\end{lma}

If $H=L^2\left( D_s \times D_t, \bR \right)$, we can approximate certain integral operators  in a stronger sense:
\begin{lma}
  \label{lma:approx-of-trace-class-operator-with-continuous-kernel-on-tensor-space}
  Let $D_s \subset \bR^p, D_t \subset \bR^q$ be compact subsets, and $C \in \tc{L^2(D_s \times D_t, \bR)}$  be
  a positive definite
  integral operator with symmetric continuous kernel
  $c=c(s,t,s',t')$, i.e.\ $c(s,t,s',t')=c(s',t',s,t)$ for all $s,s' \in D_s, t,t' \in D_t$.

  For any $\vep > 0$, there exists
  an operator $C' = \sum_{n = 1 }^N A_n \kprod B_n$, where $A_n : L^2(D_s, \bR) \rightarrow L^2(D_s, \bR) ,
  B_n : L^2(D_t, \bR) \rightarrow  L^2(D_t, \bR)$ are finite rank
  operators with continuous kernels $a_n$, respectively $b_n$, such that
  \begin{enumerate}
    \item $\tnorm{C - C'} \leq \vep$,
    \item $\sup_{s,s' \in D_s, t,t' \in D_t} \left| c(s,t,s',t') - c'(s,t,s',t') \right| \leq \vep$, where $c'$
      is the kernel of the operator $C'$,
  \end{enumerate}
  \begin{proof}
    By Mercer's Theorem, there exists continuous orthonormal functions $(U_n)_{n \geq 1} \subset L^2(D_s
    \times D_t, \bR)$
    and    $(\lambda_n)_{n \geq 1} \subset \bR$ is a summable decreasing sequence of positive numbers,
    such that
    \begin{equation}
      \label{eq:mercer-continuous}
      c(s,t,s',t') = \sum_{n \geq 1} \lambda_n U_n(s,t)U_n(s',t'),
    \end{equation}
    where the convergence is uniform
    in $(s,t,s',t')$.

    Let $C^N = \sum_{n = 1}^N \lambda_n U_n \tensort U_n$, and let $c^N$ denote its kernel. Fix $\vep > 0$,
    and let $\hnorm{g}_\infty~=~\sup_{x} |g(x)|$.
    We
    have that for $N$ large enough, both $\tnorm{C - C^N}$ and $\hnorm{c - c^N}_\infty$ are bounded by
    $\vep/2$, since $C$ is positive and \eqref{eq:mercer-continuous} is also its singular value decomposition.

We can now approximate each of the continuous functions $U_n(s,t), n
= 1, \ldots, N,$ by tensor products of continuous functions
\citep{cheney1986multivariate}. Let $U_n^M = \sum_{l = 1}^M
u_{n,l}^{(1)} \tensor u_{n,l}^{(2)}$, where $u_{n,l}^{(1)} \in
L^2(D_s, \bR), u_{n,l}^{(2)} \in L^2(D_t, \bR),  l \geq 1$, are
\emph{continuous} functions  such that
\[
  \hnorm{U_n - U_n^{M}}_\infty \leq \min\left\{ \frac{\vep}{6\kappa \trace{C} }, \kappa\right\},
\]
where $\kappa = \max_{n = 1,\ldots, N} \hnorm{ U_n }_\infty$ (notice
that $\kappa < \infty$ since each $U_n$ is continuous). Writing
$C^{N,M} = \sum_{n=1}^N \lambda_n U_n^M \tensort U_n^M$, and
denoting by $c^{N,M}$ its kernel, we have
\begin{align*}
  \hnorm{c^N - c^{N,M}}_\infty &\leq \sum_{n = 1}^N \lambda_n \left[ \hnorm{U_n - U_n^M}_\infty
\hnorm{U_n}_\infty + \hnorm{U_n^M}_\infty \hnorm{U_n - U_n^M}_\infty
\right]
\\ & \leq 3 \kappa \max_{n = 1,\ldots, N} \hnorm{U_n - U_n^M}_\infty \cdot \trace(C)
\\ & \leq \vep/2.
\end{align*}
Furthermore, we also have
\begin{align*}
  \tnorm{C^N - C^{N,M}} &\leq \sum_{n = 1}^N \lambda_n \left[ \tnorm{ (U_n - U_n^M) \tensort
U_n} + \tnorm{U_n^M \tensort (U_n - U_n^M)} \right]
\\ &\leq \sum_{n = 1}^N \lambda_n \left[ \hnorm{ U_n - U_n^M}_2 \hnorm{U_n}_2 + \hnorm{U_n^M}_2 \hnorm{ U_n -
U_n^M}_2 \right]
\\ &\leq \sum_{n = 1}^N \lambda_n \left[ \hnorm{ U_n - U_n^M}_\infty \hnorm{U_n}_\infty + \hnorm{U_n^M}_\infty
\hnorm{ U_n - U_n^M}_\infty \right]
\\ & \leq \vep/2.
\end{align*}
Since $C^{N,M} = \sum_{n = 1}^{N} \sum_{j,l = 1}^{K} \left(
\lambda_n u_{n,l}^{(1)}  \tensor
  u_{n,j}^{(1)} \right) \kprod \left(  u_{n,l}^{(2)}  \tensor u_{n,j}^{(2)} \right)$,
the proof is finished.
  \end{proof}
\end{lma}

\section{Implementation details}
\label{app:implementation}

All the implementation details described here are implemented in the R package \texttt{covsep}
\citep{R:covsep}.

In practice, random elements of $H_1 \tensor H_2$ are first
projected onto a truncated basis of $H_1 \tensor H_2$. We shall
assume that the truncated basis is of the form $(e_i \tensor
f_j)_{i=1,\ldots, d_1; j=1, \ldots , d_2}$, for some $d_1, d_2
<\infty$,  where $(e_i)_{i \geq 1} \subset H_1$, respectively
$(f_j)_{j\geq 1} \subset H_2$, is an orthonormal basis of $H_1$,
respectively $H_2$. In this way, one can encode (and approximate) an
element $X \in H_1 \tensor H_2$ by a $d_1 \times d_2$ matrix
$\mathbf{X} \in \bR^{d_1 \times d_2}$, whose $(k,l)$-th coordinate
is given by $\mathbf{X}(k,l) = \sc{X, e_k \tensor f_l}, k=1,\ldots,
d_1; l=1,\ldots, d_2$. We therefore assume from now on that only
need to describe the implementation of $T_N(r,s)$ for $H_1 \tensor
H_2 = \bR^{d_1} \tensor \bR^{d_2} = \bR^{d_1 \times d_2}$. In this
case, $\mathbf{X}$ is a random element of $\bR^{d_1 \times d_2}$,
i.e.\  a random $d_1 \times d_2$ matrix, and we observe
$\mathbf{X}_1,\ldots, \mathbf{X}_n \simiid \mathbf{X}$. We have
  \begin{align*}
    \mathbf{C}_{1,N}(k,k') &= \frac{\widetilde{ \mathbf{C}}_{1,N}(k,k')}{\sqrt{\trace( \widetilde{ \mathbf{C}}_{1,N}) }},
    & \mathbf{C}_{2,N}(l,l') &= \frac{\widetilde{ \mathbf{C}}_{2,N}(l,l')}{\sqrt{\trace( \widetilde{ \mathbf{C}}_{2,N}) }},
\end{align*}
and
\begin{align*}
  \widetilde{\mathbf{ C}}_{1,N}(k,k') &= \frac{1}{N} \sum_{i=1}^N \sum_{l=1}^{d_2} \left( \mathbf{X}_i(k,l) - \overline{\mathbf{X}}(k,l)
  \right)\left( \mathbf{X}_i(k',l) - \overline{\mathbf{X}}(k',l)\right) 
  \\ &= \left(  \frac{1}{N} \sum_{i=1}^N \left( \mathbf{X}_i - \overline{\mathbf{X}}
    \right) \left( \mathbf{X}_i - \overline{\mathbf{X}}\right)^\tp  \right)_{k,k'} =
  \sum_{l=1}^{d_2} \mathbf{C}_N(k,l,k',l) ,
  \\  \widetilde{\mathbf{C}}_{2,N}(l,l') &= \frac{1}{N} \sum_{i=1}^N \sum_{k=1}^{d_1}  \left( \mathbf{X}_i(k,l) - \overline{\mathbf{X}}(k,l)
  \right)\left( \mathbf{X}_i(k,l') - \overline{\mathbf{X}}(k,l')\right)  
  \\ &= \left(  \frac{1}{N} \sum_{i=1}^N \left( \mathbf{X}_i - \overline{\mathbf{X}}
    \right)^\tp \left( \mathbf{X}_i - \overline{\mathbf{X}}\right)  \right)_{l,l'} 
  = \sum_{k=1}^{d_1} \mathbf{C}_N(k,l,k,l') ,
\\ \overline{\mathbf{X}}(k,l) &= \frac{1}{N} \sum_{i=1}^N \mathbf{X}_i(k,l),
\\ \mathbf{C}_N(k,l,k',l') & = \frac{1}{N} \sum_{i=1}^N
\left( \mathbf{X}_i(k,l) - \overline{\mathbf{X}}(k,l) \right)\left(
\mathbf{X}_i(k',l') -
  \overline{\mathbf{X}}(k',l')\right),
  \end{align*}
for all $k,k' =1,\ldots, d_1; l,l' = 1,\ldots, d_2$.

The computation of $\widetilde{\mathbf{ C}}_{1,N}(k,k')$ using the
above formula is not efficient in \texttt{R}, when implemented using
a double \texttt{for} loop. However, if we denote by $\mathbf{A}_k$
the $N \times d_2$ matrix with $\left( \mathbf{A}_k \right)_{il} =
\mathbf{X}_i(k,l) - \overline{\mathbf{X}}(k,l)$, $n=1,\ldots, N;
k=1,\ldots, d_1; l=1,\ldots, d_2$, by $\vec(\mathbf{A}_k)$ the
vector obtained by stacking the columns of $\mathbf{A}_k$ into a
vector of length $N d_2$, and by $Y_n$ the $n$-th row of the $Nd_2
\times d_1$ matrix $A = \left( \vec(A_1), \ldots, \vec(A_{d_1})
\right)$, we get $\overline Y = (N d_2)^{-1}\sum_{n = 1}^{N d_2} =
0$ and
\[
  (N d_2)^{-1} \sum_{n=1}^{N d_2} (Y_n)_k
  (Y_n)_{k'} = (N d_2)^{-1} \sum_{i=1}^n \sum_{l=1}^{d_2} \left( \mathbf{X}_i(k,l) -
    \overline{\mathbf{X}}_i(k,l) \right).
\]
Therefore $\tilde{\mathbf{C}}_{1,N} = \frac{N d_2 -1}{N}
\texttt{cov}(A)$, where \texttt{cov} is the standard \texttt{R}
function returning the covariance, and the computation is very fast.
The computation of $\tilde{\mathbf{C}}_{2,N}$ can be done similarly.

If we denote by $(\hat{ \lambda}_r, \hat{ \mathbf{u}}_r)$,
respectively $(\hat \gamma_s, \hat{ \mathbf{v}}_s)$, the $r$-th
eigenvalue/eigenvector pair of $\mathbf{C}_{1,N}$, respectively the
$s$-th eigenvalue/eigenvector pair of $\mathbf{C}_{2,N}$, we have
\[
  T_N(r,s) =  T_N(r,s| X_1,\ldots, X_N) = \sqrt{N} \left[ \frac{1}{N} \sum_{i=1}^N \left(
      \hat{\mathbf{u}}_r^\tp(\mathbf{X}_i
      -
      \overline{\mathbf{X}})\hat{\mathbf{v}}_s\right)^2 - \hat \lambda_r \hat \gamma_s\right],
\]
where $\cdot^\tp$ denotes matrix transposition. The variance of $T_N(r,s| X_1,\ldots, X_N)$ is estimated by
\begin{multline}
  \hat \sigma^2(r,s | X_1,\ldots, X_N) =
  \\ \frac{2 \hat \lambda_r^2 \hat \gamma_s^2  \left( \trace(\mathbf{C}_{1,N})^2 + \hsnorm{\mathbf{C}_{1,N}}^2
      - 2 \hat \lambda_r \trace(\mathbf{C}_{1,N}) \right) \left( \trace(\mathbf{C}_{2,N})^2 +
      \hsnorm{\mathbf{C}_{2,N}}^2 - 2 \hat \gamma_s \trace(\mathbf{C}_{2,N})
    \right)}{\trace(\mathbf{C}_{1,N})^2\trace(\mathbf{C}_{2,N})^2 },
\end{multline}
where $\hsnorm{\mathbf{A}}^2 = \sum_{i=1}^{d_1} \sum_{j=1}^{d_2}
[(\mathbf{A})_{ij}]^2$ for a $d_1 \times d_2$ matrix $\bf A$.

{
If $\mathcal I = \left\{ 1, \ldots, p \right\} \times \left\{ 1,
  \ldots, q \right\}$, then 
$\left( T_N(r,s) \right)_{(r,s) \in \mathcal I}$ is asymptotically a mean zero Gaussian random $p \times q$
matrix, with separable covariance. Its left (row) covariances are consistently estimated by the $p \times p$ matrix
$\hat \Sigma_{L, \mathcal I} = \hat \Sigma_{L, \mathcal I}(X_1, \ldots, X_N)$  with entries
\begin{equation}
  \left( \hat  \Sigma_{L, \mathcal I} \right)_{r,r'} = 
  \frac{
\sqrt{2 \hat \lambda_r \hat \lambda_{r'}}   \left( \delta_{rr'} \trace(\widehat{C}_{1,N})^2 + \hsnorm{\widehat{C}_{1,N}}^2 - (
  \hat \lambda_r+ \hat \lambda_{r'}) \trace(\widehat{C}_{1,N}) \right)}{\trace(\widehat{C}_{1,N})\trace(\widehat{C}_{2,N}) }  ,
  \label{eq:hat-sigma-left}
\end{equation}
$r,r' \in \left\{ 1, \ldots, p \right\}$,
and it right (column) covariances are consistently estimated by the $q \times q$ matrix $\hat \Sigma_{R,
  \mathcal I} = \hat \Sigma_{R, \mathcal I}(X_1, \ldots, X_N)$ 
with entries
\begin{equation}
  \left( \hat \Sigma_{R, \mathcal I} \right)_{s,s'} = 
  \frac{
    \sqrt{ 2 \hat \gamma_s \hat \gamma_{s'} }\left( \delta_{ss'} \trace(\widehat{C}_{2,N})^2 + \hsnorm{\widehat{C}_{2,N}}^2 - (\hat \gamma_s+ \hat \gamma_{s'})
              \trace(\widehat{C}_{2,N}) \right)
  }{\trace(\widehat{C}_{1,N})\trace(\widehat{C}_{2,N}) } 
  \label{<++>}
\end{equation}
$s,s' \in \left\{ 1, \ldots, q \right\}$,
}

The computation of $\hsnorm{D_N}$ can be done without storing the
full covariance $C_N$ in memory. The following pseudo-code returns
$\hsnorm{D_N}$:
\begin{enumerate}[I.]
  \item Compute and store $\mathbf{C}_{1,N}$ and $\mathbf{C}_{2,N}$, and set $s = 0$.
  \item Replace $\mathbf{X}_n$ by $\mathbf{X}_n - \overline{ \mathbf{X}}$ for each $n=1,\ldots, N$.
\item For $i,k = 1,\ldots, d_1; j,l=1,\ldots, d_2$,
  \begin{enumerate}
    \item Compute $y = N^{-1}\sum_{n = 1}^N \mathbf{X}_n(i,j)\mathbf{X}_n(k,l)$.
    \item Set $s = s + \left(y - \mathbf{ C}_{1,N}(i,k)\mathbf{ C}_{2,N}(j,l) \right)^2$.
  \end{enumerate}
\item Return $s$.
\end{enumerate}
The computation of $\hsnorm{D_N^* - D_N}$ requires a slight
modification of the pseudo-code. Given $X_1,\ldots, X_N$ and
$X_1^*,\ldots, X_N^*$,
\begin{enumerate}[I.]
  \item Compute and store $\mathbf{ C}_{1,N}$ and $\mathbf{ C}_{2,N}$, $\mathbf{ C}^*_{1,N}$ and $\mathbf{
C}^*_{2,N}$, and set $s = 0$.
  \item Replace $\mathbf{X}_n$ by $\mathbf{X}_n - \overline{ \mathbf{X}}$, and $\mathbf{X}_n^*$ by
    $\mathbf{X}_n^* - \overline{ \mathbf{X}^*}$ for each $n = 1,\ldots, N$.
\item For $i,k = 1,\ldots, d_1; j,l=1,\ldots, d_2$,
  \begin{enumerate}
    \item Compute $y = N^{-1}\sum_{n = 1}^N \left( \mathbf{X}_n(i,j)\mathbf{X}_n(k,l) -
        \mathbf{X}^*_n(i,j)\mathbf{X}^*_n(k,l) \right)$.
    \item Set $s = s + \left(y - \mathbf{ C}^*_{1,N}(i,k)\mathbf{ C}^*_{2,N}(j,l) + \mathbf{ C}_{1,N}(i,k)\mathbf{ C}_{2,N}(j,l) \right)^2$.
  \end{enumerate}
\item Return $s$.
\end{enumerate}
Finally, Algorithms \ref{alg:parametric-boot-one-direction} and
\ref{alg:empirical-boot-one-direction} describe the procedure to
approximate the p-values for the tests based on parametric and
empirical bootstrap, respectively.

\begin{algorithm}
\caption{Parametric Bootstrap $p$-value approximation for $H_N$}
  \label{alg:parametric-boot-one-direction}
Given $X_1,\ldots, X_N$,
\begin{enumerate}[I.]
  \item compute $\overline X$, $C_{1,N} \kprod C_{2,N}$, and $H_N = H_N(X_1,\ldots,
    X_N)$.
  \item For $b = 1, \ldots, B$,
    \begin{enumerate}
      \item Create bootstrap samples $\mathbf{X}^b = \{X_1^b,\ldots,X_N^b\}$, where \newline $X_i^b \simiid F\left(\: \overline X, C_{1,N} \kprod C_{2,N} \right)$.
      \item Compute $H_N^b = H_N(\mathbf{X}^b)$,
    \end{enumerate}
  \item Compute the estimated bootstrap $p$-value
    \[
      p=\frac{1}{B} \sum_{b=1}^B \ind{H_N^b > H_N},
    \]
    where $\ind{ A} = 1$ if $A$ is true, and zero otherwise.
\end{enumerate}
\end{algorithm}

\begin{algorithm}
\caption{Empirical Bootstrap $p$-value approximation for $H_N$}
  \label{alg:empirical-boot-one-direction}
  Given $\mathbf X = \left\{ X_1,\ldots, X_N \right\}$,
\begin{enumerate}[I.]
  \item Compute $C_{1,N} \kprod C_{2,N}$, and $H_N = H_N(\mathbf X)$.
  \item For $b = 1, \ldots, B$,
    \begin{enumerate}
      \item Create the bootstrap sample $\mathbf{X}^b = \{X_1^b,\ldots,X_N^b\}$ by drawing with
        repetition from $X_1,\ldots, X_N$.
      \item For each bootstrap sample, compute $\Delta_N^b = \Delta_N(\mathbf{X}^b; \mathbf X)$.
    \end{enumerate}
  \item Compute the estimated bootstrap $p$-value
    \[
      p=\frac{1}{B} \sum_{b=1}^B \ind{\Delta_N^b > H_N},
    \]
    where $\ind{ A} = 1$ if $A$ is true, and zero otherwise.
\end{enumerate}
\end{algorithm}

\section{Additional results from the simulation studies}
\label{app:additional-simulations-studies}

Figure~\ref{fig:sim-tdist-increasing-proj} shows the empirical powers empirical bootstrap version of the tests
$\widetilde G_N(\mathcal I)$ for increasing projection subspaces, i.e.\ for $\mathcal I = \mathcal I_l,
l=1,2,3$, where  $\mathcal I_1 = \{(1,1)\}, \mathcal I_2 = \{ (i,j) : i,j = 1,2 \}$ and $\mathcal I _3 = \{
  (i,j) : i = 1,\ldots, 4; j=1,\ldots, 10 \}$, when data are generated from a multivariate $t$ distribution
with $6$ degrees of freedom (the Non-Gaussian scenario in the paper).
Figure~\ref{fig:simulations-gaussian-and-t-R2} shows the empirical power for the asymptotic test, the
parametric and empirical bootstrap tests based on the test statistic $\widetilde G_N(\mathcal{I}_2)$, as well
as parametric and bootstrap tests based on the test statistics $ G_N(\mathcal{I})$, $\widetilde
G^a_N(\mathcal{I}_2)$ where $\mathcal{I}_2= \{ (i,j) : i,j = 1,2 \}$.
Figure~\ref{fig:simulations-gaussian-and-t-R3} shows the analogous results for the projection set $\mathcal
I_3$. Tables~\ref{tab:empirical-size-r1},
\ref{tab:empirical-size-r2} and \ref{tab:empirical-size-r3} give the true levels of the tests for $\mathcal I
= \mathcal I_1, \mathcal I_2,$ and $\mathcal I =  \mathcal I_3$, respectively.

Figure~\ref{fig:grid-level-power-gaussian-3x2} shows the empirical size and power of the
separability test, in the Gaussian scenario, as functions of the projection set
\[
  \mathcal I_{r,s} = \left\{ (i,j) : 1 \leq i \leq r, 1 \leq j \leq s \right\},
\]
for all possible choices of $(r,s)$.
The test used is $\widetilde G_N(\mathcal I)$, with distribution
    approximated by the empirical bootstrap with $B=1000$.  
    Figure~\ref{fig:grid-level-power-student-3x2} is analogous plot for the Non-Gaussian scenario.

\begin{figure}[h!]
\centering
\includegraphics[width=\linewidth]{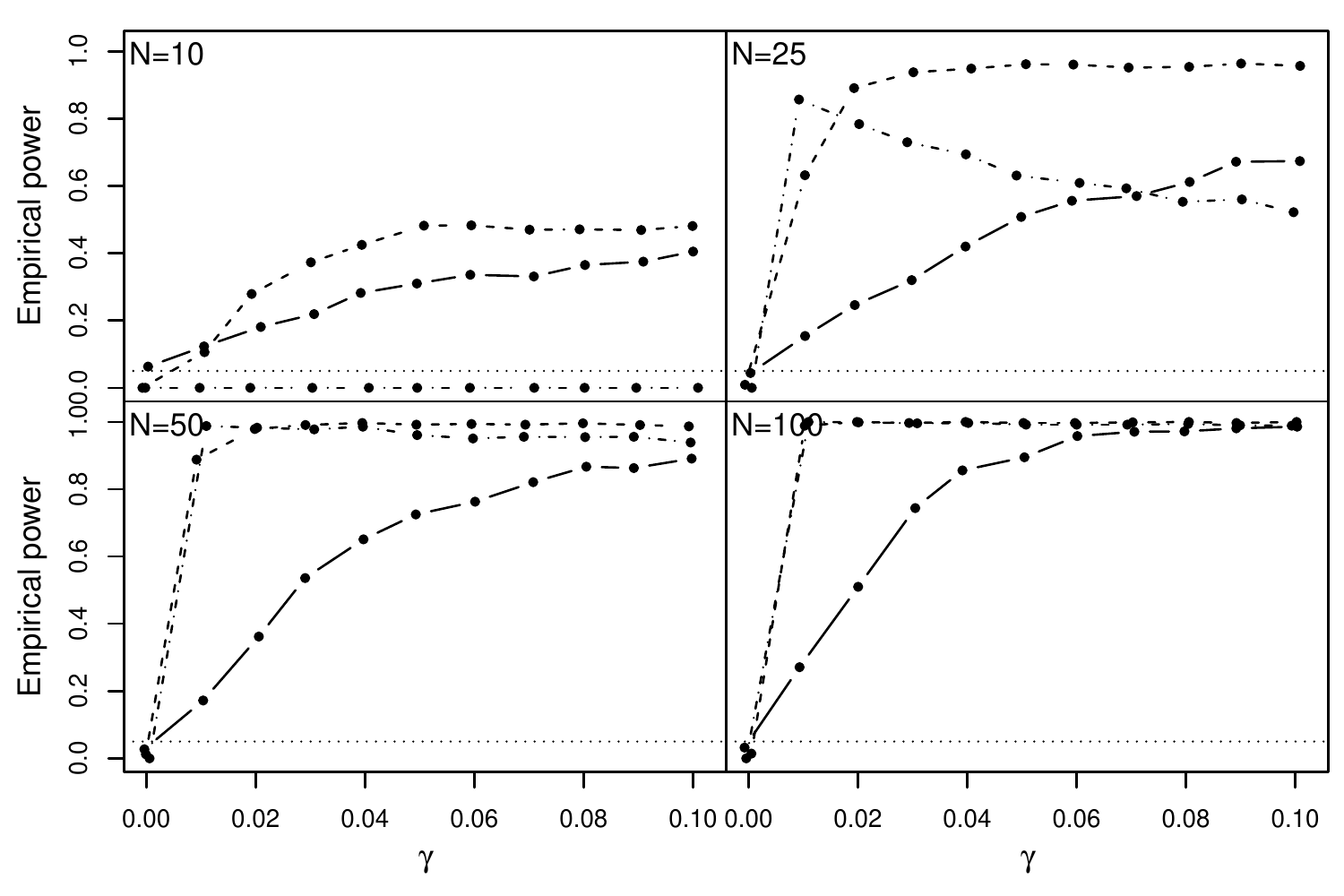}
\caption{Empirical power of the empirical bootstrap version of $\widetilde G_N(\mathcal I_l)$, for $l = 1$
  (solid line), $l=2$ (dashed line) and $l=3$ (dash-dotted line), in the \emph{non-Gaussian} scenario.  The
  horizontal dotted line indicates the nominal level ($5\%$) of the test.
  Note that the points have been horizontally jittered for better visibility.
}
\label{fig:sim-tdist-increasing-proj}
\end{figure}

\begin{figure}[h]
  \begin{center}
\includegraphics[width=\linewidth, trim=0 0 0 5,
clip=TRUE]{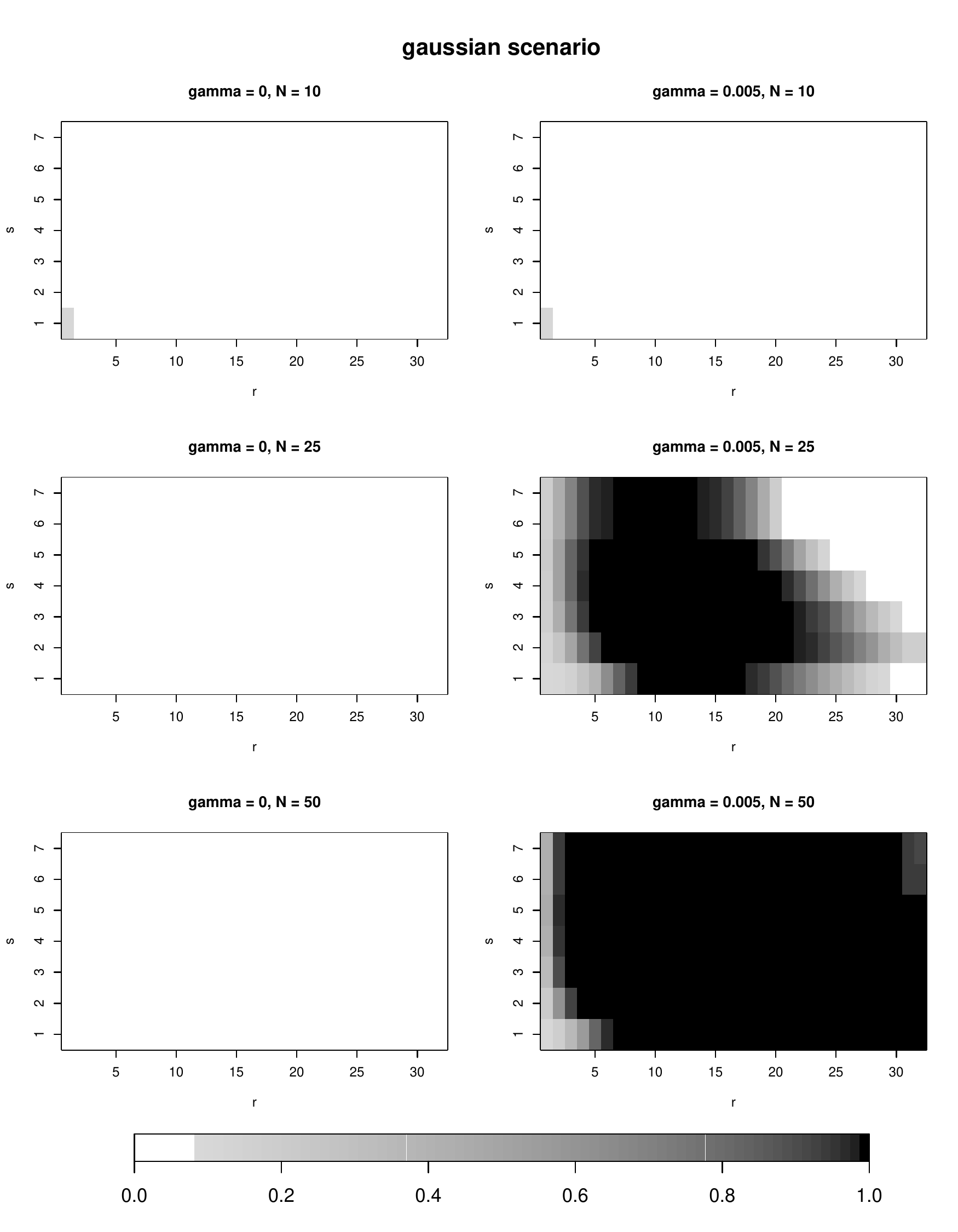}
  \end{center}
  \caption{Empirical size (left column) and power (right column) of the separability test as functions of the projection set
    $\mathcal I$. The test used is $\widetilde G_N(\mathcal I)$, with distribution approximated by the
    empirical bootstrap with $B=1000$. The left plots, respectively the right plots, were simulated from the Gaussian scenario with $\gamma
    =0$, respectively $\gamma = 0.005$. Each row corresponds to a different sample size: $N=10$ (top), $N=25$
    (middle), $N=50$ (bottom). Each $(r,s)$ rectangle represents the level/power of the test based on the projection
    set $\mathcal I = \left\{ (i,j) : 1 \leq i \leq r, 1 \leq j \leq s \right\}$.}
  \label{fig:grid-level-power-gaussian-3x2}
\end{figure}

\begin{figure}[h]
  \begin{center}
\includegraphics[width=\linewidth, trim=0 0 0 5,
clip=TRUE]{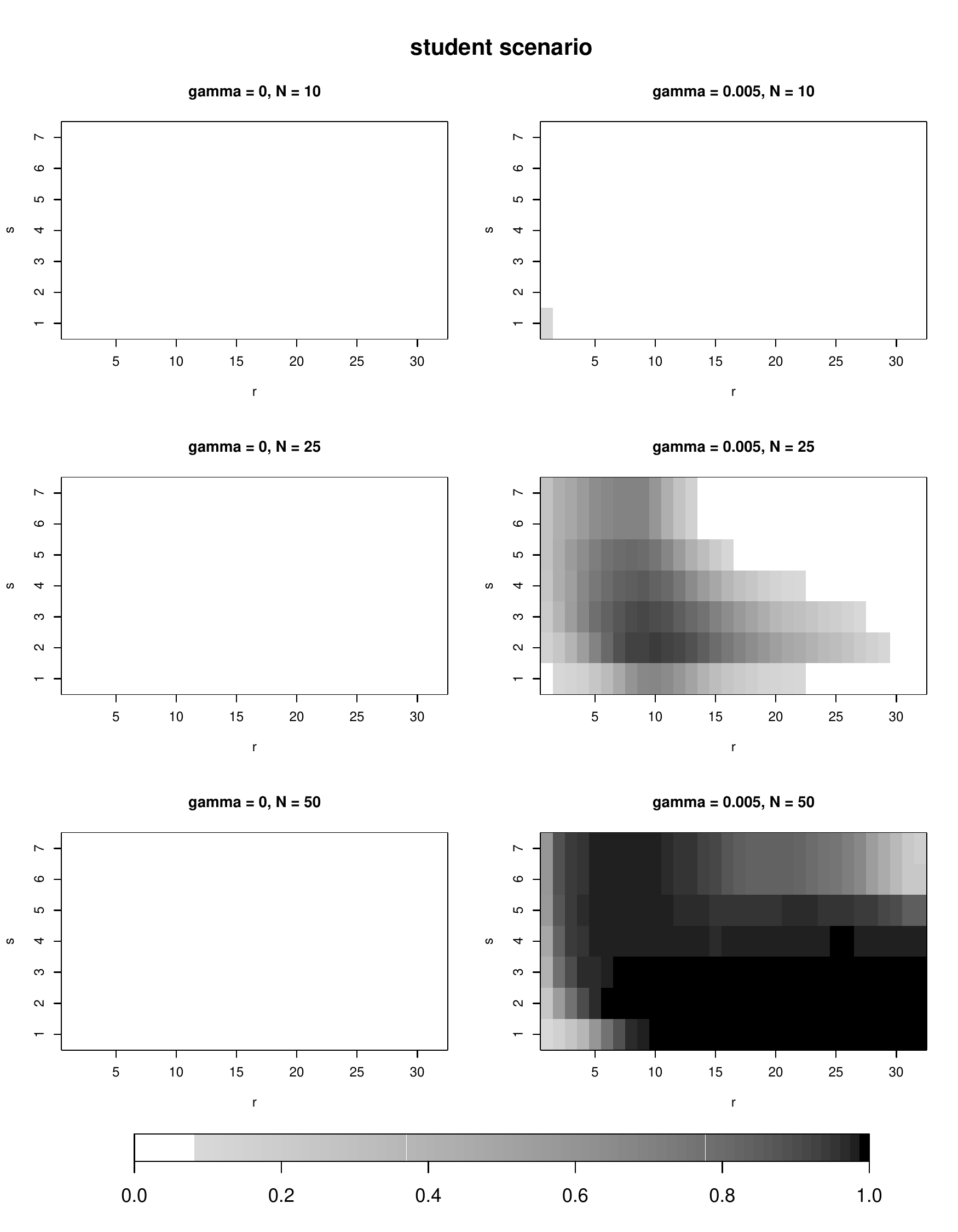}
  \end{center}
  \caption{Empirical size (left column) and power (right column) of the separability test as functions of the projection set
    $\mathcal I$. The test used is $\widetilde G_N(\mathcal I)$, with distribution approximated by the
    empirical bootstrap with $B=1000$. The left plots, respectively the right plots, were simulated from the
    Non-Gaussian scenario with $\gamma
    =0$, respectively $\gamma = 0.005$. Each row corresponds to a different sample size: $N=10$ (top), $N=25$
    (middle), $N=50$ (bottom). Each $(r,s)$ rectangle represents the level/power of the test based on the projection
    set $\mathcal I = \left\{ (i,j) : 1 \leq i \leq r, 1 \leq j \leq s \right\}$.}
  \label{fig:grid-level-power-student-3x2}
\end{figure}

\begin{figure}[p]
  \centering 
  \begin{subfigure}[c]{\linewidth}
    \includegraphics[width=\linewidth, trim=0 0 0 5, clip=TRUE]{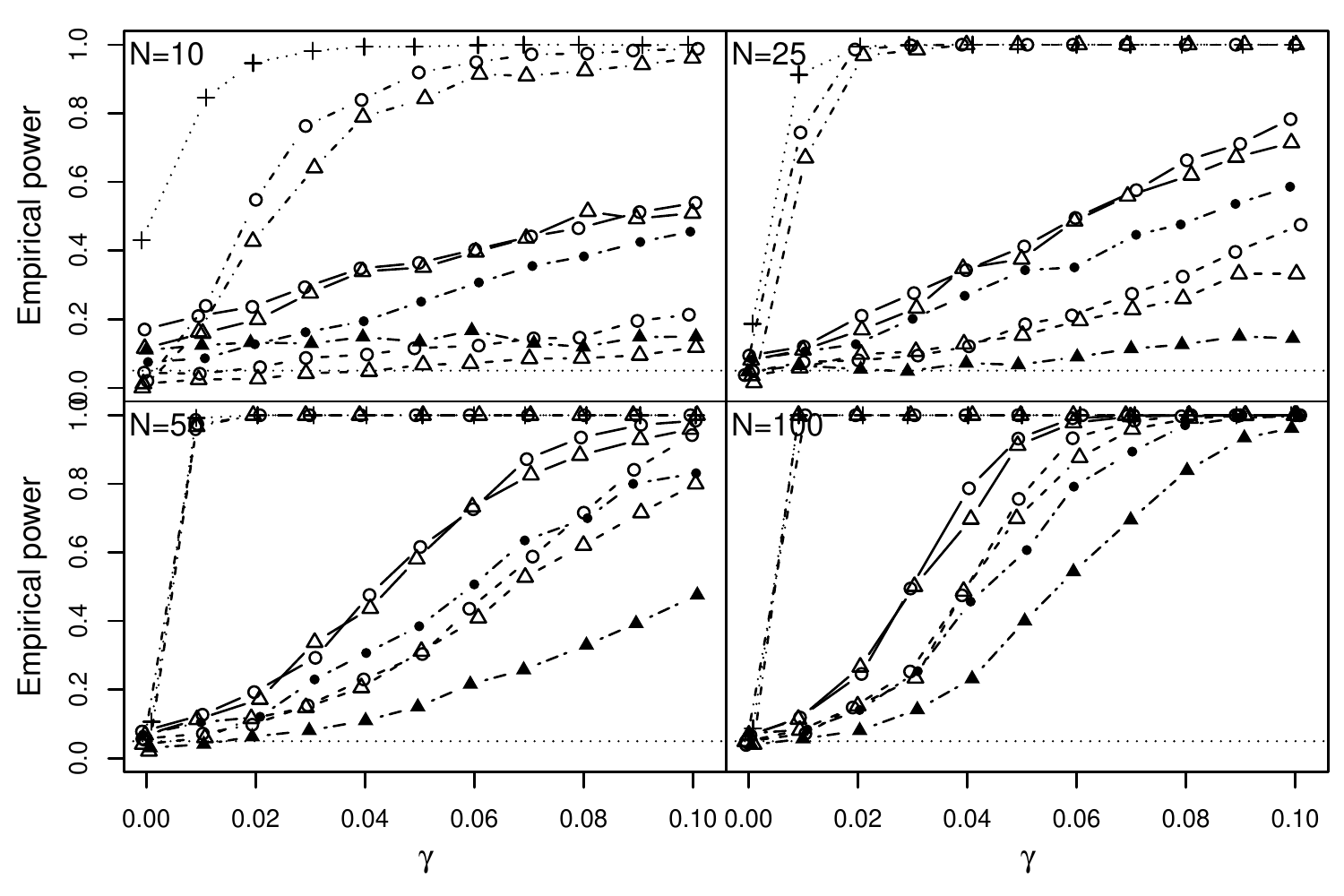}
    \caption{Gaussian scenario}
  \end{subfigure}

  \begin{subfigure}[c]{\linewidth}
    \includegraphics[width=\linewidth, trim=0 0 0 5, clip=TRUE]{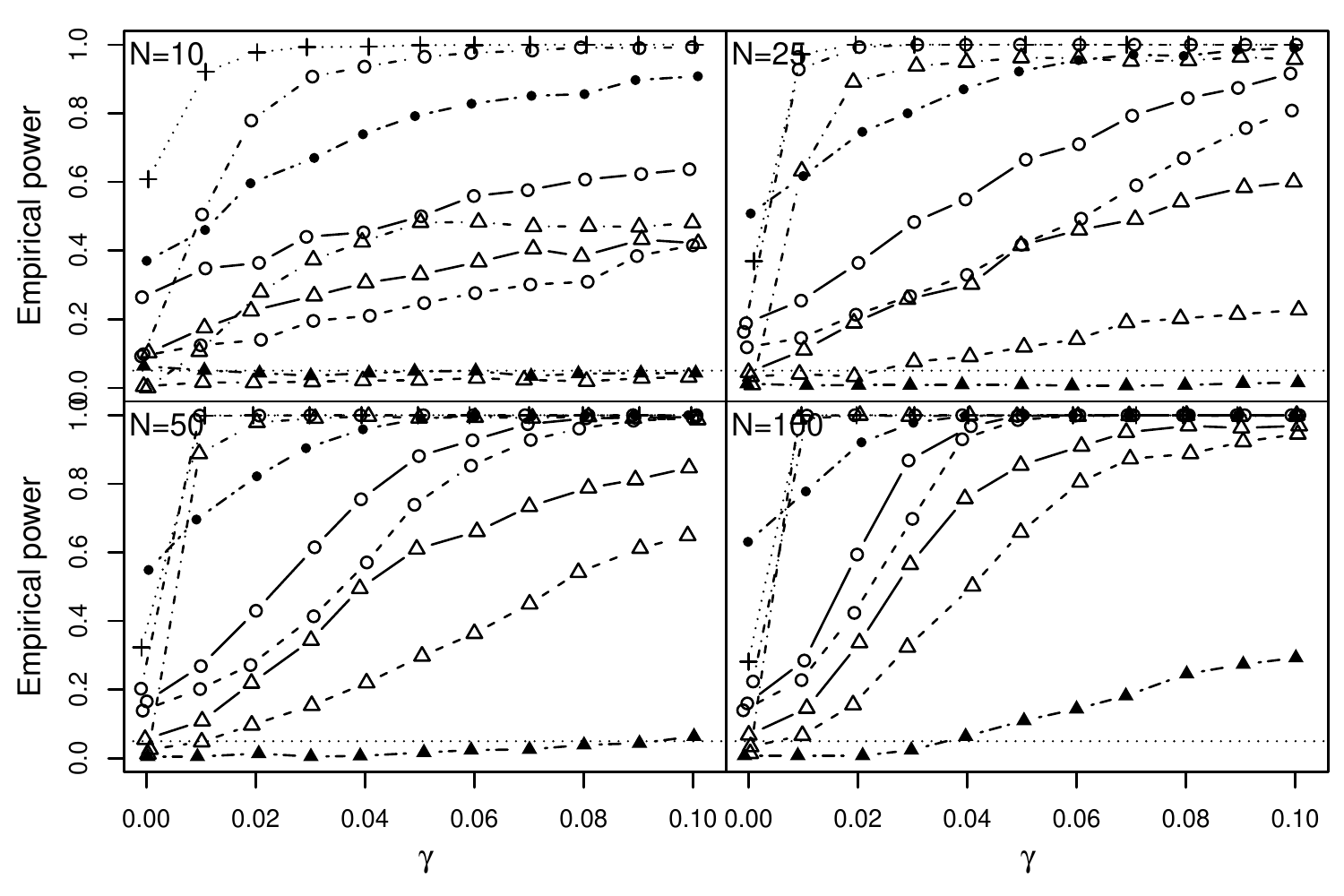}
    \caption{Non-Gaussian scenario}
  \end{subfigure}
  \caption{Empirical power of the testing procedures in the
    \emph{Gaussian} scenario (panel (a)) and \emph{non-Gaussian} scenario (panel (b)), for $N=10, 25, 50, 100$ and $\mathcal I = \mathcal I_3$. The
    results shown correspond to the test \eqref{eq:asymptotically-sigma2-times-chi2-test} based on its
    asymptotic distribution ($\cdot\!\cdot\!\cdot\!\cdot\!+\!\cdot\!\cdot\!\cdot\!\cdot$), the Gaussian
    parametric bootstrap test 
    $\widetilde G_N(\mathcal I_2)$ (dash-dotted line with empty circles),
    $\widetilde G^a_N(\mathcal I_2)$ (dashed line with empty circles), and
    $G_N(\mathcal I_2)$ (solid line with empty circles), 
      the empirical bootstrap projection tests
     $\widetilde G_N(\mathcal I_2)$ (-- $\cdot$
     --$\bigtriangleup$-- $\cdot$ --),
     $\widetilde G^a_N(\mathcal I_2)$ (-- --$\bigtriangleup$-- --), and
     $G_N(\mathcal I_2)$ (---$\bigtriangleup$---),
     the Gaussian parametric Hilbert--Schmidt test (dash-dotted line with
    filled circles) and the empirical Hilbert--Schmidt test (dash-dotted line with filled triangles). The
    horizontal dotted line indicates the nominal level ($5\%$) of the test. 
  Note that the points have been horizontally jittered for better visibility.
  }
  \label{fig:simulations-gaussian-and-t-R2}
\end{figure}

\begin{figure}[p]
  \centering 
  \begin{subfigure}[c]{\linewidth}
    \includegraphics[width=\linewidth, trim=0 0 0 5, clip=TRUE]{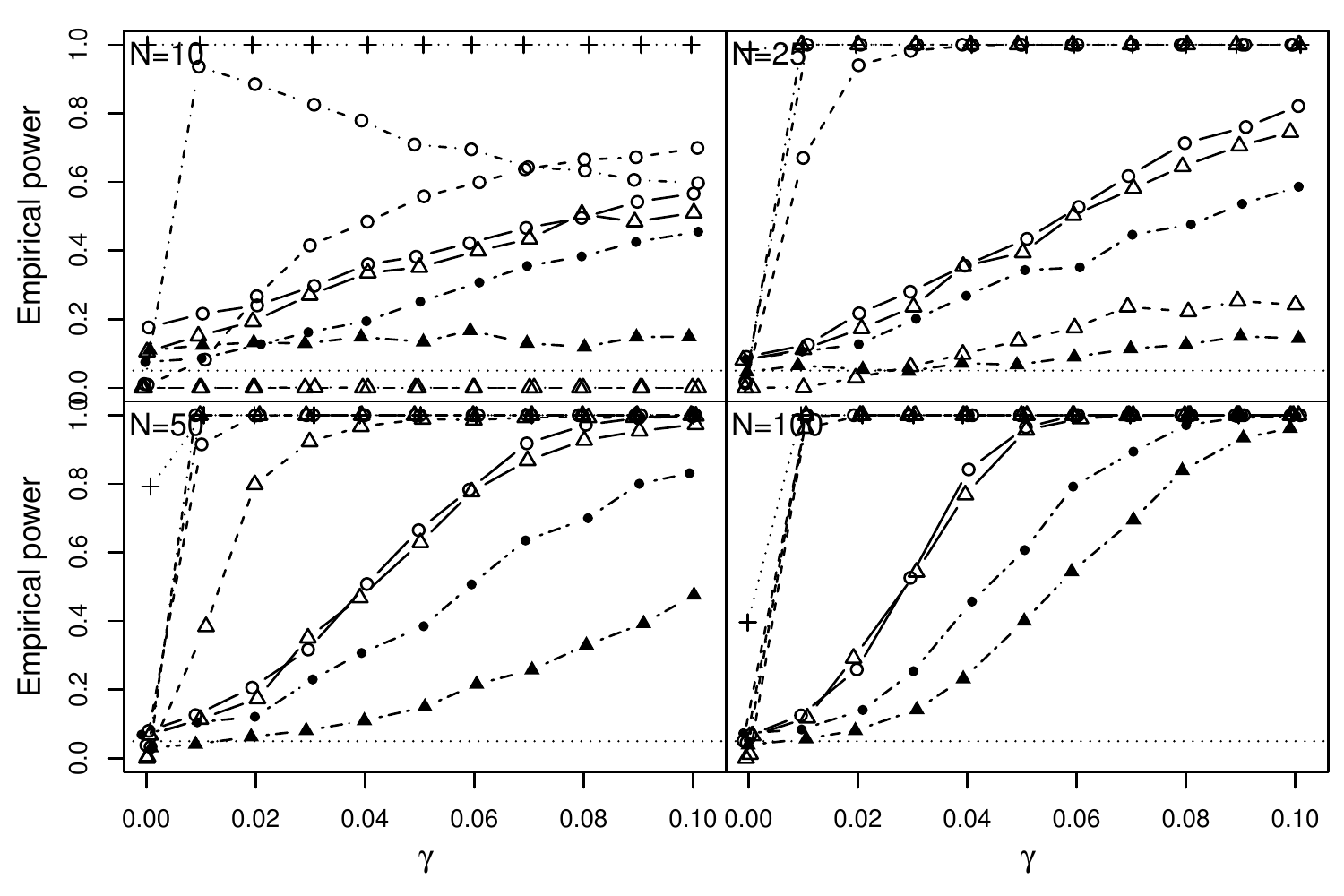}
    \caption{Gaussian scenario}
  \end{subfigure}

  \begin{subfigure}[c]{\linewidth}
    \includegraphics[width=\linewidth, trim=0 0 0 5, clip=TRUE]{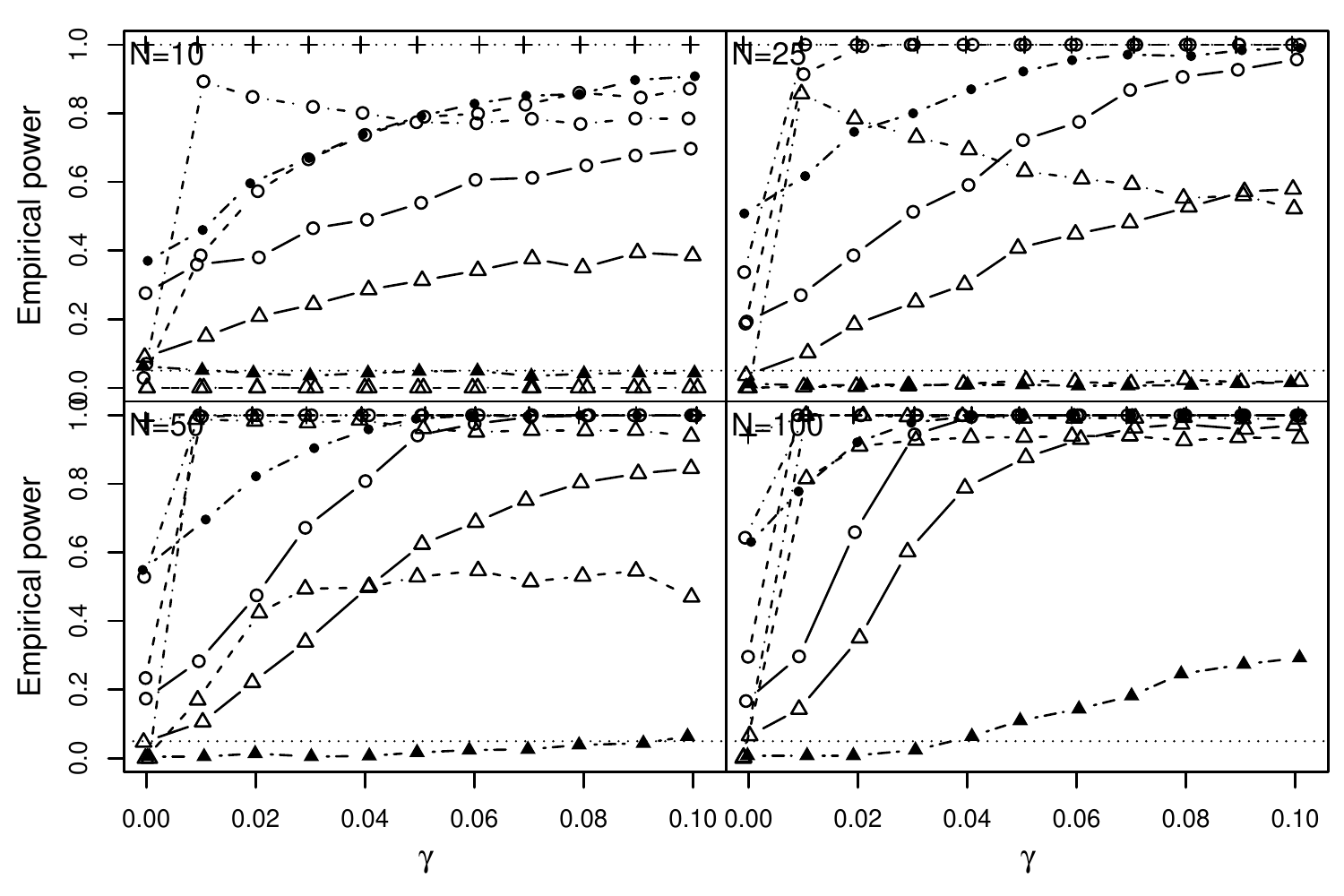}
    \caption{Non-Gaussian scenario}
  \end{subfigure}
  \caption{Empirical power of the testing procedures in the
    \emph{Gaussian} scenario (panel (a)) and \emph{non-Gaussian} scenario (panel (b)), for $N=10, 25, 50, 100$ and $\mathcal I = \mathcal I_3$. The
    results shown correspond to the test \eqref{eq:asymptotically-sigma2-times-chi2-test} based on its
    asymptotic distribution ($\cdot\!\cdot\!\cdot\!\cdot\!+\!\cdot\!\cdot\!\cdot\!\cdot$), the Gaussian
    parametric bootstrap test 
    $\widetilde G_N(\mathcal I_3)$ (dash-dotted line with empty circles),
    $\widetilde G^a_N(\mathcal I_3)$ (dashed line with empty circles), and
    $G_N(\mathcal I_3)$ (solid line with empty circles), 
      the empirical bootstrap projection tests
     $\widetilde G_N(\mathcal I_3)$ (-- $\cdot$
     --$\bigtriangleup$-- $\cdot$ --),
     $\widetilde G^a_N(\mathcal I_3)$ (-- --$\bigtriangleup$-- --), and
     $G_N(\mathcal I_3)$ (---$\bigtriangleup$---),
     the Gaussian parametric Hilbert--Schmidt test (dash-dotted line with
    filled circles) and the empirical Hilbert--Schmidt test (dash-dotted line with filled triangles). The
    horizontal dotted line indicates the nominal level ($5\%$) of the test. 
  Note that the points have been horizontally jittered for better visibility.
  }
  \label{fig:simulations-gaussian-and-t-R3}
\end{figure}

\begin{table}[p]
    \caption{Empirical size of the testing procedures (with $\alpha = 0.05$), for
      $\mathcal I = \mathcal I_1$.}
   \label{tab:empirical-size-r1}
   \vspace{.5cm}
  \begin{subtable}[c]{\linewidth}
    \centering
    \input{graphics/true-levels-GaussScenario-r1.tex}
    \caption{\emph{Gaussian} scenario}
  \end{subtable}
  \begin{subtable}[c]{\linewidth}
    \centering
    \input{graphics/true-levels-StudScenario-r1.tex}
    \caption{\emph{Non-Gaussian} scenario}
  \end{subtable}
\end{table}

\begin{table}[p]
    \caption{Empirical size of the testing procedures (with $\alpha = 0.05$), for
      $\mathcal I = \mathcal I_2$.}
   \label{tab:empirical-size-r2}
   \vspace{.5cm}
  \begin{subtable}[c]{\linewidth}
    \centering
    \input{graphics/true-levels-GaussScenario-r2.tex}
    \caption{\emph{Gaussian} scenario}
  \end{subtable}
  \begin{subtable}[c]{\linewidth}
    \centering
    \input{graphics/true-levels-StudScenario-r2.tex}
    \caption{\emph{Non-Gaussian} scenario}
  \end{subtable}
\end{table}

\begin{table}[p]
    \caption{Empirical size of the testing procedures (with $\alpha = 0.05$), for
      $\mathcal I = \mathcal I_3$.}
   \label{tab:empirical-size-r3}
   \vspace{.5cm}
  \begin{subtable}[c]{\linewidth}
    \centering
    \input{graphics/true-levels-GaussScenario-r3.tex}
    \caption{\emph{Gaussian} scenario}
  \end{subtable}
  \begin{subtable}[c]{\linewidth}
    \centering
    \input{graphics/true-levels-StudScenario-r3.tex}
    \caption{\emph{Non-Gaussian} scenario}
  \end{subtable}
\end{table}

\section*{Acknowledgments}

We wish to thank the editor, associate editor, and the referees for their comments that have led to an
improved version of the paper. We also wish to thank Victor Panaretos for interesting discussions.

\clearpage

\pagebreak

\bibliographystyle{agsm}
\bibliography{references}

\end{document}

%% file: graphics/true-levels-GaussScenario-r1.tex
\begin{tabular}{rrrrr}
  & N=10 & N=25 & N=50 & N=100 \\ 
  \hline
\hline
Asymptotic Distribution & 0.17 & 0.09 & 0.08 & 0.06 \\ 
   \hline
Gaussian parametric bootstrap (non-Studentized) & 0.22 & 0.10 & 0.08 & 0.08 \\ 
   (diag Studentized) & 0.04 & 0.04 & 0.06 & 0.05 \\ 
   (full Studentized) & 0.02 & 0.04 & 0.05 & 0.05 \\ 
   \hline
Empirical bootstrap (non-Studentized) & 0.20 & 0.11 & 0.08 & 0.07 \\ 
  (diag Studentized) & 0.10 & 0.05 & 0.05 & 0.06 \\ 
  (full Studentized) & 0.10 & 0.07 & 0.06 & 0.06 \\ 
   \hline
Gaussian parametric Hilbert--Schmidt & 0.07 & 0.08 & 0.07 & 0.07 \\ 
   \hline
Empirical Hilbert--Schmidt & 0.11 & 0.05 & 0.03 & 0.04 \\ 
  \end{tabular}

%% file: graphics/true-levels-StudScenario-r1.tex
\begin{tabular}{rrrrr}
  & N=10 & N=25 & N=50 & N=100 \\ 
  \hline
\hline
Asymptotic Distribution & 0.29 & 0.20 & 0.18 & 0.15 \\ 
   \hline
Gaussian parametric bootstrap (non-Studentized) & 0.31 & 0.21 & 0.18 & 0.17 \\ 
   (diag Studentized) & 0.08 & 0.13 & 0.14 & 0.15 \\ 
   (full Studentized) & 0.08 & 0.12 & 0.14 & 0.14 \\ 
   \hline
Empirical bootstrap (non-Studentized) & 0.20 & 0.07 & 0.06 & 0.08 \\ 
  (diag Studentized) & 0.07 & 0.06 & 0.04 & 0.04 \\ 
  (full Studentized) & 0.06 & 0.04 & 0.03 & 0.03 \\ 
   \hline
Gaussian parametric Hilbert--Schmidt & 0.37 & 0.51 & 0.55 & 0.63 \\ 
   \hline
Empirical Hilbert--Schmidt & 0.06 & 0.01 & 0.01 & 0.01 \\ 
  \end{tabular}

%% file: graphics/true-levels-GaussScenario-r2.tex
\begin{tabular}{rrrrr}
  & N=10 & N=25 & N=50 & N=100 \\ 
  \hline
\hline
Asymptotic Distribution & 0.43 & 0.19 & 0.11 & 0.09 \\ 
   \hline
Gaussian parametric bootstrap (non-Studentized) & 0.17 & 0.09 & 0.08 & 0.07 \\ 
   (diag Studentized) & 0.04 & 0.05 & 0.06 & 0.05 \\ 
   (full Studentized) & 0.02 & 0.04 & 0.05 & 0.04 \\ 
   \hline
Empirical bootstrap (non-Studentized) & 0.12 & 0.08 & 0.07 & 0.07 \\ 
  (diag Studentized) & 0.01 & 0.04 & 0.04 & 0.05 \\ 
  (full Studentized) & 0.00 & 0.01 & 0.02 & 0.04 \\ 
   \hline
Gaussian parametric Hilbert--Schmidt & 0.07 & 0.08 & 0.07 & 0.07 \\ 
   \hline
Empirical Hilbert--Schmidt & 0.11 & 0.05 & 0.03 & 0.04 \\ 
  \end{tabular}

%% file: graphics/true-levels-StudScenario-r2.tex
\begin{tabular}{rrrrr}
  & N=10 & N=25 & N=50 & N=100 \\ 
  \hline
\hline
Asymptotic Distribution & 0.61 & 0.37 & 0.32 & 0.28 \\ 
   \hline
Gaussian parametric bootstrap (non-Studentized) & 0.26 & 0.19 & 0.17 & 0.16 \\ 
   (diag Studentized) & 0.09 & 0.12 & 0.14 & 0.14 \\ 
   (full Studentized) & 0.10 & 0.16 & 0.20 & 0.22 \\ 
   \hline
Empirical bootstrap (non-Studentized) & 0.10 & 0.04 & 0.06 & 0.07 \\ 
  (diag Studentized) & 0.00 & 0.03 & 0.03 & 0.03 \\ 
  (full Studentized) & 0.00 & 0.01 & 0.01 & 0.01 \\ 
   \hline
Gaussian parametric Hilbert--Schmidt & 0.37 & 0.51 & 0.55 & 0.63 \\ 
   \hline
Empirical Hilbert--Schmidt & 0.06 & 0.01 & 0.01 & 0.01 \\ 
  \end{tabular}

%% file: graphics/true-levels-GaussScenario-r3.tex
\begin{tabular}{rrrrr}
  & N=10 & N=25 & N=50 & N=100 \\ 
  \hline
\hline
Asymptotic Distribution & 1.00 & 0.98 & 0.79 & 0.40 \\ 
   \hline
Gaussian parametric bootstrap (non-Studentized) & 0.18 & 0.09 & 0.08 & 0.07 \\ 
   (diag Studentized) & 0.01 & 0.01 & 0.04 & 0.05 \\ 
   (full Studentized) & 0.01 & 0.02 & 0.03 & 0.05 \\ 
   \hline
Empirical bootstrap (non-Studentized) & 0.10 & 0.08 & 0.07 & 0.07 \\ 
  (diag Studentized) & 0.00 & 0.00 & 0.01 & 0.01 \\ 
  (full Studentized) & 0.00 & 0.00 & 0.00 & 0.00 \\ 
   \hline
Gaussian parametric Hilbert--Schmidt & 0.07 & 0.08 & 0.07 & 0.07 \\ 
   \hline
Empirical Hilbert--Schmidt & 0.11 & 0.05 & 0.03 & 0.04 \\ 
  \end{tabular}

%% file: graphics/true-levels-StudScenario-r3.tex
\begin{tabular}{rrrrr}
  & N=10 & N=25 & N=50 & N=100 \\ 
  \hline
\hline
Asymptotic Distribution & 1.00 & 1.00 & 0.98 & 0.94 \\ 
   \hline
Gaussian parametric bootstrap (non-Studentized) & 0.28 & 0.19 & 0.17 & 0.17 \\ 
   (diag Studentized) & 0.03 & 0.19 & 0.23 & 0.30 \\ 
   (full Studentized) & 0.07 & 0.34 & 0.53 & 0.64 \\ 
   \hline
Empirical bootstrap (non-Studentized) & 0.09 & 0.04 & 0.05 & 0.07 \\ 
  (diag Studentized) & 0.00 & 0.00 & 0.00 & 0.00 \\ 
  (full Studentized) & 0.00 & 0.00 & 0.00 & 0.00 \\ 
   \hline
Gaussian parametric Hilbert--Schmidt & 0.37 & 0.51 & 0.55 & 0.63 \\ 
   \hline
Empirical Hilbert--Schmidt & 0.06 & 0.01 & 0.01 & 0.01 \\ 
  \end{tabular}